\documentclass[%
 reprints,nofootinbib
]{revtex4-2}
\usepackage{iceberg}

\usepackage[caption=false]{subfig}
\graphicspath{
    {figures/}
}

\usepackage{quiver}
\usetikzlibrary{arrows,decorations.pathmorphing}
\usetikzlibrary{shapes.geometric,calc}

\def\llb{\llbracket}
\def\rrb{\rrbracket}
\def\bra{\langle}
\def\ket{\rangle}
\def\lr{\leftrightarrow}
\def\backslash{\symbol{92}}
\def\ve{\varepsilon}
\def\weight{\mathfrak{w}}
\def\qubit{\mathfrak{q}}
\def\bg{\bm{\mathrm{g}}}
\def\Weight{\mathfrak{W}}
\def\pauli{P}
\def\be{\bm{\mathrm{e}}}
\DeclareMathOperator{\XZ}{XZ}
\DeclareMathOperator{\sgn}{sign}
\DeclareMathOperator{\cone}{cone}

\DeclareMathAlphabet{\dutchcal}{U}{dutchcal}{m}{n}

\usepackage{accents}
\newcommand{\rh}[1]{\accentset{\rightharpoonup}{#1}}
\newcommand{\lh}[1]{\accentset{\leftharpoonup}{#1}}

\makeatletter
\renewenvironment{algorithm}[1][]{%
  \par\addvspace{\intextsep}%
  \refstepcounter{algorithm}%
  \hrule height.8pt depth0pt \kern2pt%
  \def\caption##1{{\raggedright\textbf{\ALG@name~\thealgorithm}\ ##1\par}%
    \addcontentsline{loa}{algorithm}{\protect\numberline{\thealgorithm}##1}%
    \kern2pt\hrule\kern2pt\nobreak}%
}{%
  \kern2pt\hrule\par\addvspace{\intextsep}%
}
\makeatother

\begin{document}

\title{Non-CSS Quantum Code Embedding}

\author{Andrew C. Yuan}
\author{Nouédyn Baspin}
\affiliation{%
 Iceberg Quantum\\
 Sydney, Australia
}%
\author{{\scriptsize\texttt{\{andrew, nouedyn\}@iceberg-quantum.com}}}

\begin{abstract}
    We generalize the unified framework in \cite{yuan2026unified} to accommodate arbitrary stabilizer codes (referred as non-CSS for emphasis).
    The generalization has immediate consequences in areas including logical measurement, quantum weight reduction and Euclidean embedding.
    For example, currently CSS (pure $X,Z$-type) logical measurements are well-understood based on the (height-1) cone, while non-CSS (e.g., $Y$-type) logical measurements are ad hoc (e.g., rely on local Clifford transform).
    Similarly, quantum weight reduction and optimal Euclidean embeddings only exist for CSS input codes.
    Here, we show that our generalized framework addresses this issue, and thus many constructions, including qLDPC surgery, Layer Codes and quantum weight reduction generalize to non-CSS codes in a relatively straightforward fashion.
    The previously mentioned examples are derived in detail for clarity.
\end{abstract}

\maketitle

{\centering
\setlength{\fboxsep}{8pt}%
\fbox{\begin{minipage}{0.791\columnwidth}
\small
\textbf{On the length of this paper.}
The main result is Theorem~\ref{thm:symplectic-embedding}, whose statement and proof together occupy roughly two pages (Section~\ref{sec:result}); with the symplectic preliminaries of Section~\ref{sec:prelim}, it is self-contained.
Everything after that derives its applications in detail -- non-CSS surgery, Layer codes and weight reduction in Sections~\ref{sec:surgery}, \ref{sec:layer} and \ref{sec:weight}, with fault tolerance and low-level implementation in the appendices.
These are largely independent of one another and may be read in any order.
The reader should not be deterred by the length of this paper.
\end{minipage}}\par}
\vspace{\baselineskip}

\tableofcontents

\section{Introduction}
\label{sec:intro}

Quantum code embedding (or iterative mapping cones) \cite{yuan2026unified} has emerged as a common structural theme across several seemingly distinct problems in quantum error correction (QEC).
In many settings, one begins with an input code and modifies it by adding physical qubits and stabilizer checks. The purpose of the modification may vary: to measure a logical operator fault-tolerantly \cite{cohen2022low,ide2025fault,williamson2026low,he2025extractors,swaroop2026universal,cowtan2026parallel,baspin2025fast,chang2026constant}, to embed a quantum low-density parity-check (qLDPC) code into Euclidean space \cite{portnoy2023local,williamson2023layer,lin2023geometrically,li2026almost,yuan20264d,balasubramanian2026passive} or to reduce stabilizer weight and qubit degree \cite{hastings2016weight,hastings2021quantum,hsieh2025simplified,yuan2026quantum}.
Yet in all cases the central requirement is the same.
The output code must retain the logical content of the input code in a controlled and natural way, so that the embedding changes the physical realization of the code without inadvertently changing the encoded information.

The unified framework for quantum code embedding provides an algebraic language for this requirement.
For Calderbank–Shor–Steane (CSS) codes, the separation between $X$- and $Z$- type stabilizers allows the code to be represented as a chain complex, with logical operators identified through the (co)homology.
Code modifications can then be organized as an (iterative) mapping-cone-type construction, which naturally preserves the logical degrees of freedom under general regularity conditions.

In the CSS setting, fault-tolerant measurement of an $X$- or $Z$-type logical operator is well understood through the conventional (height-1) cone \cite{ide2025fault}.
One attaches an ancilla complex to the data code so that the target logical operator becomes a stabilizer whose value can be inferred by syndrome extraction.
This gives a clean homological interpretation of code surgery and explains why logical measurement is closely related to quantum weight reduction: a high-weight logical operator is effectively promoted to a stabilizer check and then replaced by a low-weight, fault-tolerant measurement ancilla.

The same embedding perspective has also played a central role in the search for geometrically local quantum codes that saturate the Bravyi–Poulin–Terhal (BPT) bounds \cite{bravyi2009nogo, bravyi2010tradeoffs}. After the discovery of \textit{good} qLDPC codes \cite{panteleev2022asymptotically,leverrier2022quantum,dinur2023good}, whose encoding rates and relative distance are constants, this problem became one of \textit{optimal} embedding: how can an abstract qLDPC code be realized in Euclidean space while preserving its logical degrees of freedom and asymptotic parameters?
Important progress came from random embedding methods which established nearly -- up to polylogarithmic corrections -- optimal Euclidean embeddings \cite{portnoy2023local,li2026almost,lin2023geometrically}, and from Layer codes which gave an explicit, optimal and modular construction for CSS qLDPC codes using surface-code layers and local defects \cite{williamson2023layer,yuan20264d}.
From the viewpoint of the unified framework, these developments illustrate how logical preservation, geometric locality, and optimal parameter scaling can be controlled within a common homological language.

Beyond optimal code parameters, one may also prefer embedded quantum codes to possess stronger physical properties.
A particularly important property is (passive) self-correction, where encoded quantum information is protected thermally without active error correction.
The Layer codes, for example, were shown not to be fully self-correcting \cite{baspin2025free}, but only partially self-correcting \cite{williamson2023layer,gu2026layer}, despite their optimal parameter scaling in Euclidean space.
A recent breakthrough goes beyond this limitation by constructing a three-dimensional self-correcting quantum memory \cite{balasubramanian2026passive}, thereby resolving the longstanding problem of whether self-correction can exist in three spatial dimensions.
Conceptually, this construction can again be viewed as an embedding procedure: the code is iteratively modified while preserving a single logical qubit, maintaining locality in $\R^3$, and increasing the syndrome cost of both $X$- and $Z$-type errors across scales. Thus, the unified embedding framework provides a natural algebraic lens through which both BPT-saturating embeddings and self-correcting constructions can be understood.

Quantum weight reduction provides another important area where the embedding framework is useful.
Even when a code is qLDPC in the asymptotic sense, the constant check weights and qubit degrees may still be too large for practical implementation; weight reduction aims to sparsify the code while preserving its logical subspace and distance.
A central requirement is that this sparsification incurs low ancilla overhead.
Hastings pioneered quantum weight reduction for CSS codes, transforming codes with check weight $\weight$ and qubit degree $\qubit$ into constant weight and degree codes, though with polynomial overhead $O(\poly (\weight,\qubit))$ and, in early versions, a loss in distance \cite{hastings2016weight,hastings2021quantum}.
Subsequent developments improved this picture: Hsieh \emph{et al.} \cite{hsieh2025simplified} achieved substantially lower overhead $O(\weight \qubit \log (\weight \qubit))$ while preserving distance up to a constant factor, while Yuan \emph{et al.} \cite{yuan2026quantum} provides a simple and explicit construction with (state of the art) check weight and total qubit degree $=6$, at the cost of larger asymptotic overhead $O(\weight^4 \qubit^4)$.
These constructions again fit naturally into the embedding viewpoint: high-weight checks are replaced by auxiliary low-weight structures, and the main algebraic task is to guarantee that the original logical information is preserved.

\begin{table*}
    \centering
    \begin{tabular}{ |l||l|l|l| }
    \hline
    \textbf{Application} & \textbf{CSS (Prior)} & \textbf{Non-CSS (Prior)} & \textbf{Non-CSS (This work)} \\
    \hline
    Logical Measurement & Height-1 cone \cite{ide2025fault} & Ad hoc; Clifford to CSS & Section \ref{sec:surgery} \\
    & & representative \cite{cowtan2026parallel,yuan2026parsimonious} & \\
    \hline
    Euclidean Embedding & Layer codes \cite{williamson2023layer,yuan20264d}; & None known & Section \ref{sec:layer} \\
    & Random embeddings \cite{portnoy2023local,lin2023geometrically,li2026almost} & & \\
    \hline
    Weight Reduction & Hastings weight reduction \cite{hastings2016weight,hastings2021quantum}; & None known & Section \ref{sec:weight} \\
    & Low-overhead variants \cite{hsieh2025simplified,yuan2026quantum} & & \\
    \hline
    \end{tabular}
    \caption{The CSS-only restriction, organized by application.
    The generalization presented here resolves the qualitative restriction in all three rows simultaneously.
    It does not, by itself, improve upon the best CSS-specific parameters, which are compared in Table \ref{tab:applications}.}
    \label{tab:css-restriction}
\end{table*}

\subsection{Existing Limitations}
\label{sec:existing-limitations}
Despite its breadth, the framework -- and most constructions it unifies -- is restricted to CSS codes.
For logical measurement there is at least a stopgap: a single-qubit Clifford reduces a non-CSS logical to a CSS representative, after which CSS surgery applies \cite{cowtan2026parallel,yuan2026parsimonious}.
This is a workaround due to the limitation of the framework, not an account of non-CSS logical measurement in its own terms.
Nor does it survive the passage from a single operator to a whole code: there need be no single-qubit Clifford making every generator simultaneously pure $X$- or $Z$-type, since generators overlapping on a qubit can demand incompatible rotations there.
The five-qubit code $\llb 5,1,3\rrb$ is the standard example: no single-qubit Clifford maps it to a CSS code.
In particular, the stopgap does not transfer to Euclidean embedding and weight reduction, and thus non-CSS analogues of the CSS constructions are simply unknown.
Table \ref{tab:css-restriction} summarizes the situation across the three applications.

The limitation also bears on Majorana fermion codes, whose structure shares the symplectic character of non-CSS codes (see Example \ref{ex:majorana}).
A Majorana code on $2n$ operators $c_1,\dots,c_{2n}$ is specified by even-weight products $\prod_{i\in S}c_i$ generating an isotropic subspace of a quadratic form over $\F_2$ -- the same data as the commutation structure of a Pauli group, with the Majorana operators playing the role of the $2n$ symplectic coordinates rather than $n$ qubit pairs. No distinguished $X$/$Z$ split is available\footnote{A Jordan--Wigner transformation restores an $X$/$Z$-like split, but only after choosing an ordering of the $2n$ modes, so the split depends on that choice rather than being canonical; its locality also degrades once the modes are not laid out along a 1D chain, which is why it is used mainly in one dimension, with auxiliary modes needed to restore locality in higher dimensions -- see \cite{verstraete2005mapping}.}, so a Majorana code is closer to a generic non-CSS code than to a CSS one.
Bravyi \emph{et al.} \cite{bravyi2010majorana} maps any Majorana code locally to a CSS code, but at the cost of doubling the logical degrees of freedom, so the encoded information is not preserved exactly.
Whether logical measurement, embedding, and weight reduction extend without this doubling is a question a genuinely non-CSS framework is positioned to address, though we do not attempt it here.

\subsection{Result}

In this manuscript, we generalize the CSS framework \cite{yuan2026unified} to general stabilizer codes, which we refer as non-CSS for emphasis.
Specifically, Theorem \ref{thm:symplectic-embedding} formalizes the statement in terms of symplectic complexes, which describes both non-CSS codes (Example \ref{ex:nonCSS}) and Majorana fermion codes (Example \ref{ex:majorana}).

Similar to \cite{yuan2026unified}, the result unifies multiple constructions across distinct problems in QEC.
In particular, while the conventional (height-1) cone was sufficient to describe CSS logical measurement \cite{ide2025fault}, we show in Section \ref{sec:surgery} that for non-CSS logical measurement, the \textit{symplectic} height-2 cone presents as a more natural language.
Moreover, in contrast to \cite{yuan2026unified}, this manuscript is not merely organizational.
In fact, it permits the immediate generalization of optimal embeddings and weight reduction to non-CSS codes, which we describe utilizing the example of conventional CSS Layer codes in Section \ref{sec:layer} and \ref{sec:weight}, respectively.
These examples are tabulated in Table \ref{tab:applications}.
Given the generality of the symplectic formalism, we expect that many other CSS constructions can similarly be extended to non-CSS stabilizer codes, and potentially to Majorana fermion codes.

\begin{table*}
    \centering
    \begin{tabular}{ |l||l|l| }
    \hline
    \textbf{Subfield} & \textbf{CSS (Existing)} & \textbf{Non-CSS (This work)} \\
    \hline
    Single Measurement & $O(W\log W)$ \cite{yuan2026parsimonious}  &  $O(W\log W)$\\
    \hline
    3D Layer Code & Each edge hosts $\le 3$ qubits & Each edge hosts $\le 3$ qubits \\
    & $O(\qubit^2)n$ qubits & $O(\weight \qubit^3)n$ qubits \\
    &Reduced weight, total qubit degree $=6$ \cite{williamson2023layer} & Reduced weight $=9$, total qubit degree $=8$\\
    \hline
    Weight Reduction & Overhead $O(\weight^4 \qubit^4)n$ qubits & Overhead $O(\weight^4 \qubit^4)n$ qubits\\
    &Reduced weight, total qubit degree $=6$ \cite{yuan2026quantum} & Reduced weight $=9$, total qubit degree $=8$ \\
    \hline
    \end{tabular}
    \caption{Comparison between constructions for CSS and their generalizations to non-CSS codes. For single logical measurement, $W$ is the weight of the targeted logical, so that the entries are the ancilla size.
    For 3D Layer codes and weight reduction, $n,\weight,\qubit$ is the number of qubits, maximum check weight, and total qubit degree of the input code, respectively.
    }
    \label{tab:applications}
\end{table*}

We emphasize that the main theorem is self-contained in Section~\ref{sec:result}, assuming only familiarity with the preliminaries of symplectic forms reviewed in Section~\ref{sec:prelim}.
The applications are then developed separately: explicit examples of non-CSS qLDPC surgery, Layer codes, and weight reduction are presented in Sections~\ref{sec:surgery}, \ref{sec:layer}, and \ref{sec:weight}, respectively.
As in Ref.~\cite{yuan2026unified}, these application sections are largely independent and may be read in any order.

\subsection{Outlook and Open Questions}

One open question is whether the symplectic framework enables parallel non-CSS logical measurement natively, at low spatial overhead: given a set $\cL^{\star}$ of commuting non-CSS logicals, possibly densely overlapping at qubits, can they be measured in parallel without additional $Y$-type ancillas?
Cowtan \emph{et al.} \cite{cowtan2026parallel} address this mostly for the case where $\cL^{\star}$ is \textit{quasi-CSS}:

\begin{definition}[Quasi-CSS]
    \label{def:quasi-css}
    A commuting set of logicals $\cL^{\star}$ is \textit{quasi-CSS} if there is a single-qubit Clifford $U=\bigotimes_i u_i$ such that
    \begin{equation}
        \label{eq:quasi-css}
        U\,\ell\,U^{\dagger} \text{ is pure $X$- or $Z$-type for every } \ell\in\cL^{\star}.
    \end{equation}
\end{definition}

A single non-CSS logical is always quasi-CSS on its own. The content of the definition is that \emph{one} Clifford works for the whole set simultaneously, which fails as soon as two logicals act as $X$ and $Y$ on the same qubit.
Whether the more complicated (non-quasi-CSS) case can be addressed remains open.

This also bears on constant-time (fast) surgery \cite{baspin2025fast,chang2026constant}, which measures multiple commuting logicals in $O(1)$ rounds instead of the usual $O(d)$ \cite{gottesman2013fault}.
Although the framework for fast surgery \cite{baspin2025fast} generalizes to non-CSS codes using the symplectic framework (see Appendix \ref{sec:fast-surgery}), the actual ancilla \cite{baspin2025fast,chang2026constant} is expected to only generalize for quasi-CSS $\cL^{\star}$; for arbitrarily commuting sets, the analogue ancilla construction is not immediate and needs further study.

We consider this the framework's most consequential open question. It is not hypothetical: \cite{cowtan2026parallel}, \cite{baspin2025fast}, and \cite{chang2026constant} all appeared within the past year, all are organized around measuring many logicals in parallel or in constant time, and the quasi-CSS restriction is a bottleneck they hit rather than one we anticipate on their behalf. There is also a precedent for what a unifying framework can and cannot do. \cite{yuan2026unified} did not itself derive the 4D and 5D Layer Codes of \cite{yuan20264d} -- those needed a genuinely new combinatorial idea (color routing), absent from the framework entirely; what \cite{yuan2026unified} supplied was the language in which the harder problem could be posed and checked.
If the symplectic cone plays that role for parallel non-CSS measurement, that would be the clearest evidence that generalizing past CSS matters.

A second, more practical question concerns Majorana fermion codes.
The CSS literature's focus is partly a matter of available tools, i.e., chain complexes making homological methods directly applicable; the applications here show those tools survive the passage to the symplectic setting, and Majorana codes carry the same structure (Section~\ref{sec:existing-limitations}).
Admitting no CSS description at all without the logical doubling of \cite{bravyi2010majorana}, they are arguably a more natural target for the symplectic formalism than qubit non-CSS codes are.
This raises the broader question of whether symplectic methods can make Majorana codes more useful in practical fault-tolerant architectures.

\subsection{Acknowledgment: Use of AI}
\label{sec:ai-statement}

The results of this paper -- Theorem \ref{thm:symplectic-embedding}, its proof, and the constructions and analyses derived from it -- were obtained by the authors without the use of artificial intelligence.
AI tools were used for writing and polishing only: improving the exposition, and checking the manuscript for notational consistency and typographical errors.
All such changes were reviewed by the authors, who take full responsibility for the content of this paper.
\section{Main Result}
\label{sec:result}
As demonstrated in Section \ref{sec:symplectic-complex} (and, specifically, Example \ref{ex:nonCSS}), every non-CSS code can be written in the form of a symplectic complex with basis, and thus we formulate the code embedding theorem as follows.
\begin{theorem}[Symplectic Embedding]
    \label{thm:symplectic-embedding}
    Let
    \begin{equation}
        D=S^{D} \xrightarrow{\sigma^D} P^D\xrightarrow{\hat{\sigma}^D} \bar{S}^{D}
    \end{equation}
    be a symplectic complex with basis, symplectic map $\lambda^{D}$ so that $P^D=Q_X^D\oplus Q_Z^D$, and isometry (Definition \ref{def:isometry}) $\phi^D:S^D \to \bar{S}^D$.
    Let ($A^\top$) $A$ be a (co)complex with basis, and with (co)differential $\partial^A$ ($\delta^A$) given by
    \begin{align}
        A^{\top} &=X^A\to Q^A_X\to \bar{Z}^{A} \\
        A        &=Z^A \to Q_Z^A \to \bar{X}^{A}
    \end{align}
    where $\bar{X}^A,\bar{Z}^A$ are (isometric) copies of $X^A,Z^A$, respectively, and $Q_X^A,Q_Z^A$ are (isometric) copies of, say, $Q^A$.

    Consider the height-2 cone $C=S\to P\to \bar{S}$ \cite{yuan2026unified} given by the following diagram
    \begin{equation}
    \begin{tikzpicture}[baseline]
    \matrix(a)[matrix of math nodes, nodes in empty cells, nodes={minimum size=25pt},
    row sep=2em, column sep=2em,
    text height=1.25ex, text depth=0.25ex]
    {&& X^{A}  & Q^{A}_X & \bar{Z}^{A}\\
    & S^{D}  & P^{D}  & \bar{S}^{D} &\\
    Z^{A} & Q^{A}_Z & \bar{X}^{A} &&\\};
    \path[->,font=\scriptsize]
    (a-1-3) edge node[above]{$\delta^A$}  (a-1-4)
    (a-1-4) edge node[above]{$\delta^A$}  (a-1-5)
    (a-2-2) edge node[above]{$\sigma^{D}$}  (a-2-3)
    (a-2-3) edge node[above]{$\hat{\sigma}^{D}$}  (a-2-4)
    (a-3-1) edge node[above]{$\partial^A$}  (a-3-2)
    (a-3-2) edge node[above]{$\partial^A$}  (a-3-3);
    \path[->,font=\scriptsize]
    (a-1-3) edge node[right]{$g$}  (a-2-3)
    (a-1-4) edge node[right]{$h$}  (a-2-4)
    (a-2-2) edge node[right]{$\hat{h}=h^{\top}\phi^D$}  (a-3-2)
    (a-2-3) edge node[right]{$\hat{g}=g^{\top}\lambda^D$}  (a-3-3);
    \path[->,dashed,black!30!green,font=\scriptsize]
    (a-1-3) edge[bend right=80] node[left]{$p$} (a-3-2)
    (a-1-4) edge[bend left=80] node[right]{$p^{\top}$} (a-3-3);
    \end{tikzpicture}
    \end{equation}
    where gluing maps $(g,h)$ is a chain map\footnote{And thus, by definition,  $(\hat{h},\hat{g})$ is also a chain map}, and $g,\hat{g}$ are compatible with defect map\footnote{Equivalently, $(p,p^{\top})$ is a chain homotopy for $\hat{g}g$ and the zero map} $p$, i.e., $\hat{g}g=\partial^A p+ p^{\top} \delta^A$.
    Equip the height-2 cone $C$ with (mixed) symplectic map
    \begin{equation}
        \lambda =
        \begin{pmatrix}
            0 & 0 & I\\
            0 & \lambda^{D} & 0\\
            I & 0 & 0
        \end{pmatrix}
    \end{equation}
    so that $C$ has symplectic basis $P=Q_X\oplus Q_Z$ where
    \begin{equation}
        \label{eq:symplectic-decomposition}
        Q_X=Q^{A}_X\oplus Q_X^D, \quad Q_Z=Q^{A}_Z\oplus Q_Z^D
    \end{equation}
    and with isometry $\phi:S\to \bar{S}$ via
    \begin{equation}
        \phi=
        \begin{pmatrix}
        0&0&I\\
        0&\phi^D&0\\
        I&0&0
        \end{pmatrix}
    \end{equation}
    Equip the embedded column complex $C^{\bg}$ below
    \begin{equation}
        C^{\bg}= H^0(A) \xrightarrow{[g]} H_1(D) \xrightarrow{[\hat{g}]} H_{0}(A)
    \end{equation}
    with the \textit{induced}\footnote{Lemma \ref{lem:induced-symplectic}} symplectic form $[\Lambda^{D}]$ and \textit{canonically identify}\footnote{Lemma \ref{lem:cohomology-iso}} $H_0(A)=H^0(A)^*$.
    Then $C,C^{\bg}$ are both symplectic complexes.
    In particular, as in \cite{yuan2026unified}, if $H_1(A)=0$, then $H_1(C)\cong H_1(C^{\bg})$ with the standard isomorphism.
\end{theorem}

\begin{remark}[Basis]
    Note that Theorem \ref{thm:symplectic-embedding} can be formulated in a canonical (basis independent) manner (similar to Theorem 1 of \cite{yuan2026unified}). However, for applications in non-CSS (especially the code distance), a corresponding basis is always chosen, and thus we've chosen the slightly more concrete formulation with basis.
    In particular, the distance $d(C)$ of $C$ can be lower bounded in the conventional manner using the Cleaning Lemma in \cite{yuan2026unified}.
\end{remark}

\begin{remark}[Canonical Necessity]
    Note that the the embedded column $C^{\bg}$ technically does not possess a basis even if the construction of $C$ is equipped with a basis, and thus it's also necessary to understand the terminology in a canonical manner.
    The relation between the canonical formalism and the basis dependent formalism is elaborated in the Prelimary Section \ref{sec:prelim}.
    However, for most practical applications, the embedded code corresponds to an input code with chosen basis, and thus the convention is to choose a basis for the embedded code which matches that of the input code.
\end{remark}

\begin{remark}[Ancilla Dressing]
    \label{rem:ancilla-dressing}
    Note that the ancilla $A$ is CSS on its own: before gluing, $X^A$ acts purely on $Q_X^A$ (via $\delta^A$) and $Z^A$ acts purely on $Q_Z^A$ (via $\partial^A$), with no coupling between the two sectors.
    This changes once $C$ is assembled.
    The map $g,p$ indicates that a generator $x\in X^A$ -- an ancilla check that was pure $X$-type in isolation -- has a contribution $gx$ to the data code $D$, and, whenever $p\neq 0$, an additional $Z$-type action $px$ on the ancilla's \textit{own} qubits $Q_Z^A$.
    In other words, the ancilla itself generically becomes non-CSS once attached, its former $X$-checks dressed by $Z$-type support from the defect map $p$.
    See Fig. \ref{fig:surface-ancilla} for an example.
\end{remark}

\begin{remark}[Homological Perturbation Lemma]
    As noted in Ref. \cite{balasubramanian2026passive}, the CSS framework \cite{yuan2026unified} can be understood via the homological perturbation lemma instead of the original iterative mapping cone picture.
    It's expected that a similar observation applies to the non-CSS (symplectic) framework.
\end{remark}
\begin{proof}[Proof of Theorem \ref{thm:symplectic-embedding}]
    Most of the proof follows similarly as in \cite{yuan2026unified}. The key observation and improvement is the construction of the height-2 cone in a manner which preserves the \textit{symplectic duality.}
    Hence, it's sufficient to prove that the height-2 cone $C$ and the embedded complex complex $C^{\bg}$ are symplectic complexes.

    First note that $\Lambda$ is indeed a symplectic form, since $\lambda$ is non-degenerate (bijective).
    Further note that the generating map $\sigma:S\to P$ of $C$ thus has block-matrix form
    \begin{equation}
        \sigma =
        \begin{pmatrix}
            \delta^A & 0 & 0\\
            g & \sigma^{D} & 0\\
            p & \hat{h} & \partial^A
        \end{pmatrix}
    \end{equation}
    Similarly, the map $\ve :P\to \bar{S}$ of $C$ has block-matrix form
    \begin{equation}
        \ve =
        \begin{pmatrix}
            \delta^A & 0 & 0\\
            h & \hat{\sigma}^{D} & 0\\
            p^{\top} & \hat{g} & \partial^A
        \end{pmatrix}
    \end{equation}
    Utilizing the isometry $\phi:S\to \bar{S}$, it's straightforward to check that $\ve = \phi \sigma^{\top} \lambda \equiv \hat{\sigma}$ and thus $C$ is indeed a symplectic complex (with basis).

    Now let us consider the embedded column complex.
    By Lemma \ref{lem:cohomology-iso} and \ref{lem:induced-symplectic}, we see that if $[\ell^{D}]\in H_1(D)$ and $[a^{\top}]\in H^0(A)=\ker \delta^A$, then the canonical isomorphism $\psi:H_{0}(A)\to H^0(A)^*$ implies that
    \begin{align}
        (\psi[\hat{g}][\ell^{D}]) [a^{\top}]&=(\psi [\hat{g}\ell^{D}])[a^{\top}]\\
        &=a^{\top}(\hat{g}\ell^{D})\\
        &=\Lambda^{D}(\ell^{D},g a^{\top})\\
        &=[\Lambda^D]([\ell^D],[g][a^{\top}])
    \end{align}
    where we utilized the fact that $[g]$ maps $H^0(A)$ into $H_1(D)$.
    Hence, $C^{\bg}$ is also a symplectic complex.
\end{proof}
\section{Preliminaries}
\label{sec:prelim}

\subsection{General}
\begin{definition}
    We write $[n]=\{1,...,n\}$ and $[n)=\{1,...,n-1\}$.
    We also write $[a,b]$ to only denote the integer values within the interval and similarly for $[a,b),(a,b],(a,b)$.
\end{definition}


In the following sections, we will often deal with vector spaces $V,W$ equipped with some basis, such that there is a one-to-one correspondence between basis elements, and thus warrants the following (slightly different from convention) definition.
\begin{definition}[Isometry]
    \label{def:isometry}
    Let $V,W$ be spaces equipped with basis and $\phi:V\to W$ be an isomorphism that maps the the basis of $V,W$ in a one-to-one correspondence, then we refer to $\phi$ as an \textbf{isometry}, and $V,W$ are \textbf{isometric}, denoted by $V\cong W$.
\end{definition}

\subsection{Complexes}

\subsubsection{Canonical Formalism}
\begin{definition}
\label{def:chain-complex}
A \textbf{(chain) complex} $C$ (of \textbf{length} $n$) is a sequence of (finite-dimensional) $\F_2$-vector spaces $C_{i}$ (whose elements are referred as $i$-\textbf{chains}) together with linear $\partial_{i}:C_{i} \to C_{i-1}$, called the \textbf{differentials} of $C$, such that $\partial_{i}\partial_{i+1}=0$ where the subscripts are often omitted. We write
\begin{equation}
    C = C_n\cdots \to  C_{i} \xrightarrow{\partial_{i}} C_{i-1} \to \cdots  C_0
\end{equation}
Note that $\im \partial_{i+1} \subseteq \ker \partial_i$ for all $i$, and thus the \textbf{$i$-homology} of $C$ is defined as
\begin{equation}
    H_i(C)\equiv  \ker \partial_i/\im \partial_{i+1}
\end{equation}
Denote the equivalences class of $i$-chain $\ell$ as $[\ell]\in H_i(C)$ -- conversely, we may also write $[\ell] \in H_i(C)$ without specifying the representation $\ell$.
\end{definition}

\begin{definition}[Cocomplex]
Let $C$ denote a complex.
Then the \textbf{cocomplex} $C^{\top}$ is defined as
\begin{equation}
    C^\top = C_n^* \cdots \leftarrow C_{i}^* \xleftarrow{\delta_{i}} C_{i-1}^* \leftarrow \cdots C_0^*
\end{equation}
where $C_i^*$ is the dual space of $C_i$, whose elements are referred as $i$-\textbf{cochains}, and \textbf{codifferential} $\delta_i$ is defined as
\begin{equation}
    (\delta_i c_{i-1}^{\top})(c_i)= c_{i-1}^{\top} (\partial_i c_i)
\end{equation}
where $c_i$ is a chain of $C$ and $c_{i-1}^{\top}$ is a cochain.
The \textbf{$i$-cohomology} of $C$ is $H^i(C) = \ker \delta_{i+1}/\im \delta_i$
\end{definition}

\begin{lemma}[(Co)homology Isomorphism]
    \label{lem:cohomology-iso}
    Let $C$ be a complex. Then $\delta \delta =0$ and there exists a canonical isomorphism $H_i(C)\to H^i(C)^*$.
\end{lemma}

\begin{proof}
    It's straightforward to check that $\delta \delta=0$ and thus let us restrict our attention to the canonical isomorphism.
    Without loss of generality, we shall consider the canonical isomorphism $\phi :H^i(C)\to H_i(C)^*$, since the dual of a dual is canonically isomorphic to itself.
    Specifically, define $\phi$ as follows
    \begin{equation}
        (\phi [c^{\top}]) [c] =c^{\top}(c)
    \end{equation}
    where $c,c^{\top}$ are arbitrary representations of the equivalence class $[c],[c^{\top}]$, respectively.
    First note that $\phi$ is well-defined. Indeed, note that
    \begin{align}
        (c^{\top} +\delta f^{\top})(c+\partial g) &= c^{\top}(c) +f^{\top}(\partial c)\\
        & +(\delta c^{\top})(g) +f^{\top}(\partial^2 g) \\
        &= c^{\top}(c)
    \end{align}
    where we used the fact that $\partial c=0,\delta c^{\top}=0$ and $\partial^2=0$.
    It's straightforward to check that $\phi$ is linear. 
    
    Note that if $\phi[c^{\top}]=0$, then $c^{\top}(c)=0$ for any $c \in \ker \partial$.
    Define $f:\im \partial \to \F_2$ as $f(\partial c)=c^{\top} (c)$, which is well-defined since if $\partial c=\partial c'$, then $c^{\top}(c+c')=0$. Using simple basis extension, we can extend $f$ to some linear map on $C_i\to \F_2$, and thus $f(\partial c)=(\delta f)(c)$.
    Hence, $c^{\top} =\delta f\in \im \delta$, and thus $[c^{\top}]=0$.
    Hence, $\phi$ is injective.
    
    One can then easily check that $\dim H_i(C)^* =\dim H^i(C)$ and thus $\phi$ is an isomorphism.
    Note that although we chose a basis in the proof (for simple basis extension), the definition of $\phi$ is basis independent and thus canonical.
    Using the fact that the dual of dual is canonically isomorphic to itself, one can show that the corresponding canonical isomorphism $\psi:H_i(C)\to H^i(C)^*$ is given by
    \begin{equation}
        (\psi [c])[c^{\top}] = c^{\top} (c)
    \end{equation}
    
\end{proof}

\subsubsection{Basis Formalism}
\begin{definition}[Basis]
\label{def:complex-basis}
A complex $C$ is equipped with a \textbf{basis} $\cC$ (assumed henceforth) if each $C_i$ is equipped basis elements in $\cC_i$ referred as \textbf{$i$-cells}, which induces a nondegenerate bilinear form $\bra \cdot|\cdot\ket$ on $C_i$ (or equivalently, each $C_i=\F_2^{n_i}$ and the $i$-cells are the standard basis).
The  \textbf{(Hamming) weight} $|\ell|$ is then defined for any $\ell \in C_{i}$.
We say that cells $c_{i},c_{i-1} $ are \textbf{adjacent}  if $\bra c_{i-1}|\partial c_{i}\ket \ne 0$, and write $c_i \sim c_{i-1}$. 
\end{definition}

\begin{remark}[Riesz]
    If complex $C$ is equipped with a non-degenerate bilinear form, then by Riesz representation, the map $c\mapsto \bra c|\cdots \ket$ from $C_i\mapsto C_i^*$ is an isomorphism.
    This isomorphism also induces the relation $\delta = \partial^{\top}$ where $\partial^{\top}$ is that defined by the nondegenerate bilinear form.
    Hence, when complex $C$ is equipped with a basis, we shall identify the cochains in $C^{\top}$ as chains in $C$ and $\delta = \partial^{\top}$.
\end{remark}

\begin{definition}[Systolic Distance]
    Let $C$ be a complex. Then the $i$\textbf{-systolic distance} $d_i(C)$ is 
    \begin{equation}
        d_i(C) = \min_{\ell:0\ne [\ell]\in H_i(C)} |\ell|
    \end{equation}
    Similarly, define the \textbf{(co)systolic distance} $d^{i}(C)$ via $H^{i}(C)$.
\end{definition}
\begin{example}[Repetition Code]
    \label{ex:rep}
    The \textbf{repetition code} on $L$ bits is the complex $R \equiv R(L)$ with differential $\partial^{R}$.
    The 1-cells are denoted via $|i^+\ket$ for $i\in [L)$ where $i^\pm =i\pm 1/2$, and 0-cells via $|i\ket$ for $i\in [L]$ so that
    \begin{equation}
        \partial^{R} |i^+\ket = |i\ket +|i+1\ket
    \end{equation}
    We implicitly assume that $|i\ket, |i^\pm\ket =0$ if the label is not within the previous parameters.
\end{example}

\begin{example}[Graphs]
    Given a graph $\cG=(\cE,\cV)$, there is an associated complex $G=E\to V$ where $E,V$ are the $\F_2$ span of $\cE,\cV$ and the differential map is the adjacency relation. 
    Conversely, a complex $G =E\to V$ can define a graph $\cG=(\cE,\cV)$ if $|\partial e|=2$ for all 1-cells $e\in E$. In either case, we refer to $G$ as a \textbf{graph} complex with associated graph $\cG$.
    For example, the repetition code $R(L)$ is a graph complex. 
    We further call complex $C=C_2 \to C_1 \to C_0$ a \textbf{cell} complex if $C_1 \to C_0$ is a graph complex and $G$ is a \textbf{subgraph} of $C$ if $G$ is a subgraph of $C_1\to C_0$.
\end{example}


\begin{example}[CSS Codes]
    \label{ex:CSS}
    Consider a CSS code as a complex $A=Z\to Q\to X$ and convention that $Z,Q,X$ are the $Z$-type checks, qubits, $X$-type checks, respectively. 
    Then the (co)homology $H_1(C)$ ($H^1(C)$) corresponds to the collection of equivalence classes of $Z$- ($X$-) type logical operators, while the 1-(co)systolic distance $d_1(C)$ ($d^1(C)$) is the $Z$- ($X$-) type code distance.
\end{example}

\subsection{Symplectic Complex}
\label{sec:symplectic-complex}

\subsubsection{Canonical Formalism}
\begin{definition}
    A \textbf{symplectic form} on a $\F_2$ vector space $\pauli$ is a map $\Lambda:\pauli\times \pauli\to \F_2$ which is bilinear, alternating ($\Lambda(\ell,\ell)=0$ for any $\ell$) and non-degenerate ($\Lambda(\ell,\cdot)$ is an isomorphism on $\pauli$ for any $\ell$).
    When a symplectic form is given, write $V^{\perp}$ to be the \textbf{symplectic complement} of subspace $V\subseteq Q$, i.e., the collection of all $\ell$ such that $\Lambda(\ell,\ell')=0$ for all $\ell'\in V$. Note that since $\Lambda$ is non-degenerate, $V^\perp$ satisfies $\dim V + \dim V^\perp = \dim \pauli$.
\end{definition}
As we shall soon see in Example \ref{ex:nonCSS}, general (non-CSS) stabilizer codes can be written in the form of a complex, with additional properties.
Hence, this provides motivation for the following definition

\begin{definition}[Symplectic Complex]
    A 3-term (length-2) \textbf{symplectic complex} is a complex (not necessarily with basis)
    \begin{equation}
        C = S\xrightarrow{\sigma} \pauli\xrightarrow{\hat{\sigma}} S^*
    \end{equation}
    such that  $S^*$ is the dual space of $S$ and $\pauli$ is equipped with a symplectic form $\Lambda$  and  $\hat{\sigma}\ell = \Lambda(\ell,\sigma \cdots)$.
    We refer to $\sigma$ as the \textbf{stabilizer} map and $\hat{\sigma}$ as the \textbf{syndrome} map.
\end{definition}

\begin{lemma}[Induced Symplectic Form]
    \label{lem:induced-symplectic}
    Given a symplectic complex $C=S\to \pauli\to S^*$ with symplectic form $\Lambda$ on $Q$. Then $\Lambda$ induces a symplectic form $[\Lambda]$ on $H_1(C)$ such that 
    \begin{equation}
        [\Lambda]([\ell],[\ell'])=\Lambda(\ell,\ell')
    \end{equation}
\end{lemma}
\begin{proof}
    Given $\ell,\ell' \in \ker \hat{\sigma}$, note that
    \begin{align}
        \Lambda(\ell+\sigma g,\ell'+\sigma g') &= \Lambda(\ell,\ell') \\
        &\quad+\Lambda(\sigma g,\ell') +\Lambda(\ell,\sigma g) \\
        &\quad +\Lambda(\sigma g,\sigma g')\\
        &=\Lambda(\ell,\ell')
    \end{align}
    Hence, $([\ell] ,[\ell'])\mapsto \Lambda(\ell,\ell')$ is a well-defined map $H_1(C)\times H_1(C)\to\F_2$.
    It's straightforward to check that the map is bilinear and alternating and thus it remains to show that it's non-degenerate.

    Fix $[\ell]\in H_1(C)$ with representation $\ell\in \ker \hat{\sigma}$. 
    If $\Lambda(\ell,\ell')=0$ for all $\ell'\in \ker \hat{\sigma}$, then $\ell$ is in the symplectic complement of $\ker \hat{\sigma}$.
    However, note that $\ker \hat{\sigma}$ is the symplectic complement of $\im \sigma$ in $\pauli$, and thus $\ell \in \im \sigma$.
    Hence, the induced symplectic form is injective. 
    Since $\dim H_1(C)<\infty$, the induced symplectic form is an isomorphism. 
\end{proof}

\begin{lemma}[Symplectic Duality]
    \label{lem:symplectic-duality}
    Given a symplectic complex $C=S\to \pauli\to S^*$ with symplectic form $\Lambda$ on $Q$.
    Then its cocomplex $C^\top$ is also a symplectic complex with symplectic form $\Lambda^{\top}:P^*\times P^* \to \F_2$ defined as
    \begin{equation}
        \Lambda^{\top}(p^{\top},q^{\top})=\Lambda(p,q)
    \end{equation}
    where $p\in P$ is the unique element such that $p^{\top}=\Lambda(p,\cdots)$ and similarly for $q,q^\top$.
    Moreover, $C$ is isomorphic to its cocomplex in a canonical manner, given by the chain isomorphism below
    \begin{equation}
        \begin{tikzpicture}[baseline]
        \matrix(a)[matrix of math nodes, nodes in empty cells, nodes={minimum size=25pt},
        row sep=1.5em, column sep=3em,
        text height=1.25ex, text depth=0.25ex]
        {S &P & S^*\\
         S^{**} & P^* & S^*\\};
        \path[->,font=\scriptsize]
        (a-1-1) edge node[above]{$\sigma$} (a-1-2)
        (a-1-2) edge node[above]{$\hat{\sigma}$} (a-1-3)
        (a-2-1) edge node[below]{$\hat{\sigma}^\top$} (a-2-2)
        (a-2-2) edge node[below]{$\sigma^\top$} (a-2-3)
        (a-1-1) edge node[right]{$\iota$} (a-2-1)
        (a-1-2) edge node[right]{$p\mapsto \Lambda(p,\cdots)$} (a-2-2)
        (a-1-3) edge node[right]{$\mathrm{id}$} (a-2-3);
        \end{tikzpicture}
    \end{equation}
    where $\iota:S\to S^{**}$ is the canonical isomorphism $(\iota s)(s^{\top})=s^{\top}(s)$ where $s^{\top}\in S^*, s\in S$.
\end{lemma}
\begin{proof}
    It's straightforward to check that $\Lambda^{\top}$ is a symplectic form since $p\mapsto p^{\top}$ is an isomorphism.
    Note that for $s\in S,p\in P$, we have
    \begin{align}
        (\hat{\sigma}^\top \iota  s)(p)= (\hat{\sigma}p)(s) = \Lambda(p,\sigma s)
    \end{align}
    Similarly, note that
    \begin{equation}
        (\sigma^{\top}\Lambda(p,\cdots))(s)=\Lambda(p,\sigma s)=(\hat{\sigma}p)(s)
    \end{equation}
    Hence, the diagram is commuting.
    Note that the chain map is indeed an isomorphism since $\Lambda$ is non-degenerate.
\end{proof}

\subsubsection{Basis Formalism}

It's well know that
\begin{lemma}[Symplectic Basis]
    Let $\pauli$ be equipped with symplectic $\Lambda$. 
    Then $\dim \pauli$ is even and there exists a \textbf{symplectic basis}, i.e., basis $q_{X,1},...,q_{X,n},q_{Z,1},...,q_{Z,n}$ such that
    \begin{align}
        \Lambda(q_{X,i},q_{Z,j})&=\delta_{ij} \\
        \Lambda(q_{X,i},q_{X,j})&=\Lambda(q_{Z,i},q_{Z,j}) = 0
    \end{align}
\end{lemma}

\begin{definition}[Basis]
    Similar to Definition \ref{def:complex-basis}, a symplectic complex $C$ is equipped with a \textbf{basis}  (assumed henceforth) $\cC$ if $S$ is equipped with a basis $\cS$ that induces a nondegenerate bilinear form, and $S^*$ is equipped with a basis $\cS^*$ which defines an isometry\footnote{Recall we only deal with finite-dimensional spaces}  $\phi:S\to S^*$, and that $\pauli$ is equipped with a symplectic basis $\cQ_X\sqcup \cQ_Z$ (which also induces a nondegenerate bilinear form $\bra\cdot|\cdot\ket$) and we write $\pauli=Q_X\oplus Q_Z$ with the understanding that basis elements $\cQ_X,\cQ_Z$ are in one-to-one correspondence with say ordered elements $\cQ$, which we denote implicitly when writing $q_X\in \cQ_X,q_Z\in \cQ_Z$ and $q\in \cQ$. 
\end{definition}

\begin{remark}[Riesz]
    Given a nondegenerate bilinear form $\bra \cdot|\cdot\ket$ and a symplectic form $\Lambda$ on $\pauli$, there exists a unique linear map $\lambda:\pauli\to \pauli$ such that $\bra \ell |\lambda \ell'\ket =\Lambda(\ell,\ell')$, and thus we will use $\Lambda$ and $\lambda$ interchangeably, where $\lambda$ will be referred to as the \textbf{symplectic map}.
    With respect the symplectic basis, $\lambda$ is in \textbf{standard form}
    \begin{equation}
        \lambda =
        \begin{pmatrix}
            0 &I\\
            I &0
        \end{pmatrix}
    \end{equation}
    Conversely, given matrix $\lambda$ in standard form relative to some basis, the corresponding $\Lambda(p,q)=p^{\top}\lambda q$ is a well-defined symplectic form and the basis is a symplectic basis.
    
    In particular, a symplectic complex (with basis) can be rewritten as
    \begin{equation}
        C=S\xrightarrow{\sigma} \pauli \xrightarrow{\hat{\sigma}} \bar{S}
    \end{equation}
    where $\bar{S},S$ are isometric via the isometry $\phi:S\to \bar{S}$ and $\hat{\sigma}=\phi \sigma^{\top}\lambda$. 
\end{remark}


\begin{example}[Non-CSS Codes]
    \label{ex:nonCSS}
    For a general stabilizer code, it is not necessarily possible to find parity checks that are purely $X$- or $Z$- type Paulis and thus will be referred as \textbf{non-CSS} codes for emphasis, despite containing CSS codes as a special case.
    Non-CSS codes can then be written as a symplectic complex with basis 
    \begin{equation}
        C=S\xrightarrow{\sigma} \pauli\xrightarrow{\hat{\sigma}} \bar{S}
    \end{equation}
    so that $S$ is the space of \textbf{checks} with support defined by the stabilizer map $\sigma$, (i.e., each column of $\sigma$ denotes the abelianization of the corresponding check), $P$ denotes the \textbf{Pauli space} $=\F_2^{2n}$ with symplectic basis $P=Q_X\oplus Q_Z$ so that $Q_X,Q_Z$ are isometric copies of $Q=\F_2^n$, and $\bar{S}$ denote the space of \textbf{syndromes}, isometric to $S$.
    
    With this identification, $H_1(C)$ is the collection of (equivalence class of) logical operators of the non-CSS code, with the 1-systolic distance $d_1(C)=d(C)$ equal to the code distance.
    Note that by duality in Lemma \ref{lem:symplectic-duality}, the 1-systolic and 1-cosystolic distance are equal.
    Since stabilizer codes have commuting parity checks, $\hat{\sigma}\sigma =0$ and thus $C$ is indeed a complex.

\end{example}

\begin{example}[CSS Codes as Symplectic Complex]
    \label{ex:CSS-symplectic}
    By Example \ref{ex:CSS}, a CSS code can be written as a complex $C_{\rm{CSS}}=Z\to Q_Z\to \bar{X}$ with (co)differential $\partial$ ($\delta$). 
    By Example \ref{ex:nonCSS}, it can also be rewritten as a symplectic complex
    \begin{equation}
        C=\underbrace{X\oplus Z}_{S} \xrightarrow{\sigma} \underbrace{Q_X\oplus Q_Z}_{\pauli}\xrightarrow{\hat{\sigma}} \underbrace{\bar{Z}\oplus \bar{X}}_{\bar{S}}
    \end{equation}
    where the bar notation is the differentiate between checks and syndromes, and $Q_X\cong Q_Z\cong Q$ denote the $X,Z$-sectors of qubits $Q$.
    In particular, we may write $x,z$ to denote checks (cells) in $X,Z$ and $\bar{z},\bar{x}$ to denote syndrome (cells) in $\bar{Z},\bar{X}$, respectively.
    Similarly, write $q_X,q_Z$ to denote qubits (cells) in the $X,Z$-sector $Q_X,Q_Z$ corresponding to a physical qubit $q\in Q$, respectively.
    
    Further note that the symplectic map $\lambda$ is in standard form, and stabilizer map
    \begin{equation}
        \sigma =
        \begin{pmatrix}
            \partial & 0\\
            0 & \delta
        \end{pmatrix}
    \end{equation}
    and syndrome map given by
    \begin{equation}
        \hat{\sigma} = 
        \phi
        \sigma^{\top} \lambda 
        =
        \begin{pmatrix}
            \partial & 0\\
            0 & \delta
        \end{pmatrix}
    \end{equation}
    where $\phi:S\to \bar{S}$ is the isometry given by
    \begin{equation}
        \phi = 
        \begin{pmatrix}
            0 & I\\
            I & 0
        \end{pmatrix}
    \end{equation}
\end{example}

\begin{example}[Majorana Codes]
    \label{ex:majorana}
    Although this will not be the topic of focus, it's worth mentioning that Majorana codes can also be written as a symplectic complex. Specifically, the space of \textbf{Majorana modes} $M =\F_2^{2n}$ (with standard basis $|i\ket$) has symplectic form
    \begin{equation}
        \Lambda(a,b)=\bra a|b\ket + \bra a|\cM\ket \bra \cM|b\ket
    \end{equation}
    where $\cM$ is the sum of all basis elements in $M$. A possible symplectic basis is given by the Jordan-Wigner transform
    \begin{align}
        X_i \equiv  \sum_{s < 2i} |s\ket, \quad Z_i\equiv |2i-1\ket+ |2i\ket
    \end{align}
    
\end{example}



\section{Non-CSS Surgery}
\label{sec:surgery}

In this section, we apply the main result in Theorem \ref{thm:symplectic-embedding} to qLDPC surgery, and show that non-CSS (e.g., $Y$-type) logical measurements falls within this framework.
For concreteness, we restrict our attention to the measurement of a single non-CSS logical. 
To guide the reader on the abstract construction, we shall repeatedly consider the example of measuring a $\bar{Y}$-logical of the standard surface code, as shown in Fig. \ref{fig:surface-surgery}.

As recognized initially in \cite{ide2025fault}, for CSS (pure $X$- or $Z$-type) logical measurement of CSS qLDPC codes, the conventional (height-1) cone is sufficient. 
However, it was previously unknown whether a similar formulation can be achieved for non-CSS logicals, even when restricted to CSS codes, since the addition of non-CSS logicals results in the \textit{deformed code} to be non-CSS.
Here, we show that our non-CSS embedding framework addresses this problem -- in fact, it also permits the discussion of non-CSS logical measurement on non-CSS codes.

Note that \cite{yuan2026parsimonious} provides the currently known lowest ancilla overhead, i.e., $O(\Weight \log \Weight)$ given logical of weight $\Weight$, and thus serves as the fundamental basis of a wide variety of known qLDPC CSS surgery.
Hence, we shall generalize the construction to qLDPC non-CSS surgery.
For simplicity (and ease of comparison with current literature), we shall restrict our attention to the case where the data code is CSS -- the extension to measuring non-CSS logicals on a non-CSS code is straightforward and collected in the last Subsection \ref{sec:ext-non-CSS}.

\begin{figure}[ht]
\subfloat[\label{fig:surface}]{%
    \centering
    \includegraphics[width=0.3\columnwidth]{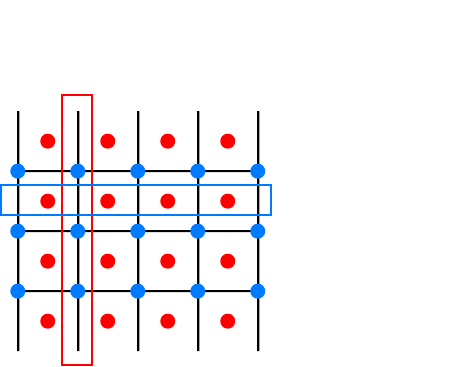}
}
\quad
\subfloat[\label{fig:surface-ancilla}]{%
    \centering
    \includegraphics[width=0.3\columnwidth]{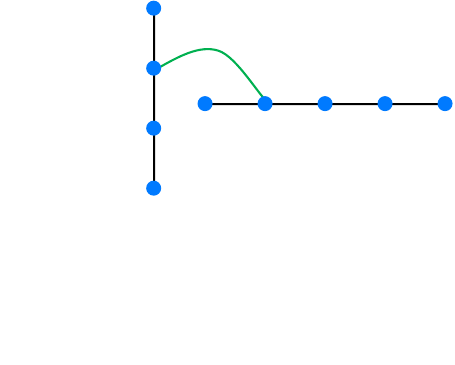}
}
\\
\subfloat[\label{fig:surface-cone}]{%
    \centering
    \includegraphics[width=0.3\columnwidth]{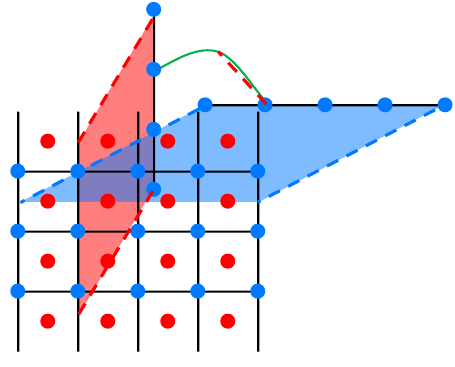}
}
\quad
\subfloat[\label{fig:surface-cone-stabilizer}]{%
    \centering
    \includegraphics[width=0.3\columnwidth]{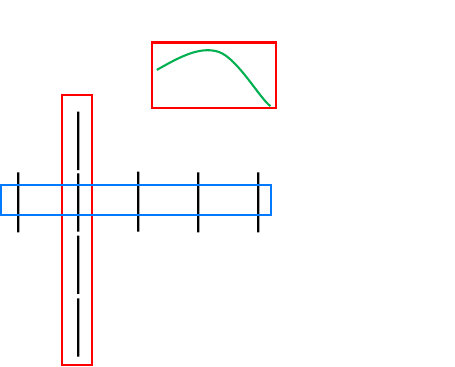}
}
\caption{$\bar{Y}$ Measurement. (a) denotes a surface code with edges hosting qubits, where blue vertices (red faces) indicates $X$ ($Z$) checks acting on adjacent edges. Blue (red) solid lines encompass a $\bar{X}$ ($\bar{Z}$) logical representation so that the product denotes a $\bar{Y}$ logical. (b) denotes the measurement graph as constructed in Definition \ref{def:measurement-graph}, where for simplicity, we set $\cE_{\rm{extra}} =\varnothing$. Note that the green edge denotes $\cE_{X\wedge Z}$ since the $\bar{X},\bar{Z}$ representations overlap.
Note that this measurement graph has no cycles and thus we can skip Theorem \ref{thm:parsimonious-cone}, and take the ancilla code to be the measurement graph so that blue vertices indicate $X$-checks acting on adjacent edges hosting qubits. 
(c) denotes the symplectic cone constructed in Theorem \ref{thm:symplectic-embedding} or \ref{thm:surgery}. The red shaded region indicates that checks in the ancilla (data) code further act as a $Z$ Pauli on the corresponding qubit in the data (ancilla) code. Similarly, the blue shaded region indicates that checks in the ancilla (data) code further act as an $X$ Pauli on the corresponding qubit in the data (ancilla) code.
The red dashed line connecting the green edge indicates that the neighboring blue vertex must also act as a $Z$ Pauli on the green edge, corresponding to the defect $p$ in Theorem \ref{thm:symplectic-embedding}.
(d) denotes the product of all ancilla checks in Fig. \ref{fig:surface-cone} (and thus a stabilizer), where qubits encompassed by red (blue) lines are acted upon by a $Z$ ($X$)-Pauli.
See Remark \ref{rem:recovery-and-threshold}.
}
\label{fig:surface-surgery}
\end{figure}

\subsection{CSS Data Code}
By Example \ref{ex:CSS-symplectic}, let the data CSS be written as a symplectic complex 
\begin{equation}
    \label{eq:data-symplectic}
    D=X^D\oplus Z^D\xrightarrow{\sigma^D} Q_X^D\oplus Q_Z^D \xrightarrow{\hat{\sigma}^D} \bar{Z}^D\oplus \bar{X}^D
\end{equation}
with checks $x,z$ in $X^D,Z^D$ (and syndromes $\bar{x},\bar{z}$ differentiated by the bar notation), and qubits $q$ with corresponding $X,Z$-sector qubits $q_X,q_Z$.
Note that we use the superscript $D$ to emphasize that the corresponding objects (e.g., stabilizer map, syndrome map, symplectic form, etc) are in the data code.
As an example, consider $D$ to be the surface code in Fig. \ref{fig:surface}.

\subsection{Ancilla Code}

As discussed in Section II of \cite{yuan2026parsimonious}, the first step of constructing an ancilla $A$ for the data code $D$ is to consider the measurement graph induced by the logical.
Here, we generalize the idea to non-CSS logicals of $D$.
Specifically, let $\ell^{\star} =\ell_X^{\star}\oplus \ell_Z^{\star}\in P^D$ denote the abelianization of a logical operator where we also regard $\ell_X^{\star} \subseteq \cQ_X^D\cong \cQ^D$ as a subset, and similarly for $\ell_Z^{\star}$.
Let $\Weight=|\ell^{\star}|=|\ell_X^{\star}|+|\ell_Z^{\star}|$ denote the weight of the logical.
In the example of Fig. \ref{fig:surface}, the solid blue and red lines indicate $\ell_X^\star,\ell_Z^{\star}$, respectively.

\begin{definition}[Measurement (Multi)Graph, Fig. \ref{fig:surface-ancilla}]
    \label{def:measurement-graph}
    Let $\ell^{\star}=\ell_X^{\star}\oplus \ell_Z^{\star}$ be a logical, i.e., $\ell^{\star} \in \ker \hat{\sigma}^D$, with weight $\Weight$.
    Let $\cV_X$ denote a copy of collection $\ell_X^{\star} \subseteq \cQ_X^D$ with elements (vertices) denotes as $|q_X\ket$ for $q_X\in \ell_X$.
    For every $z$-check $z\in Z^D$ with support overlapping with $\ell_X^{\star}$, we note that the overlap must be even in cardinality by commutation relations (i.e., $\ell^{\star} \in \ker \hat{\sigma}$) and thus can be paired up arbitrarily $q_{X}\tilde{q}_{X}$. 
    Let $\cE_X$ denote the collection 
    $\|q_{X}\tilde{q}_{X};z\ket$ over all pairs and over all overlapping $z$ checks so that $\cG_X=(\cE_X,\cV_X)$ is a (multi) graph.
    Similarly define $\cV_Z,\cE_Z$ so that $\cG_Z=(\cE_Z,\cV_Z)$ is a (multi) graph.

    Consider the disjoint union of graphs $\cG_X^D \sqcup \cG_Z^D$.
    If $\ell_X^{\star},\ell_Z^{\star}$ regarded as subsets of $\cQ^D$ has nontrivial overlap $\ell_X^{\star}\wedge \ell_Z^{\star} \ne \varnothing$, then add an edge connecting $|q_X\ket \in \cV_X$ and $|q_Z\ket \in \cV_Z$ where $q_X,q_Z$ correspond to the same qubit in the overlap $\ell_X\wedge \ell_Z$.
    Let $\cE_{X\wedge Z}$ denote the collection of such edges $\|q\ket$ where $q\in \cQ^D$ corresponds to $q_X,q_Z$. 
    An example of the resulting (multi)graph is shown in Fig. \ref{fig:surface-ancilla}.

    Further add arbitrary edges $\cE_{\rm{extra}}$ while keeping the max degree fixed so that the resulting (multi) graph $\cG=(\cE,\cV)$ is connected and expanding with $\Omega(1)$ Cheeger constant\footnote{This can be achieved by adding random edges, or using more sophisticated methods~\cite{lubotzky1988ramanujan,hoory2006expander,alon2008elementary}.}.
    Then $\cG$ is the \textbf{measurement (multi) graph} associated with $\ell^{\star}$, with associated graph complex $G:E\to V$.
\end{definition}

Note that 
\begin{align}
    \cV&= \cV_X \sqcup  \cV_Z \\
    \cE&= \cE_X\sqcup \cE_Z \sqcup \cE_{X\wedge Z} \sqcup  \cE_{\rm{extra}}
\end{align}
Also note that the max degree $\Delta(G) \le \qubit+1$ where $\qubit$ is the total qubit degree of data code $D$.
By \cite{yuan2026parsimonious}, we have the following ancilla construction
\begin{theorem}[Parsimonious Cone]
    \label{thm:parsimonious-cone}
    Let $G=E\to V$ denote a graph complex with bounded degree $\Delta = \Delta(G)$.
    Then there exists cell complex $C$ such that $G$ is a subgraph of $C$ with 
    \begin{equation}
        \dim C_0 = O(|\cV| \log|\cV|),
    \end{equation}
    and
    \begin{align}
        H_2(C)=H_1(C)=0, &\quad H_0(C) \cong \F_2\\
        \label{eq:reduced-weights-ancilla}
        C_2 \xrightleftharpoons[4+\Delta]{5} C_1 \xrightleftharpoons[9+\Delta]{2} C_0.
    \end{align}
\end{theorem}

Given measurement graph $G$ of logical $\ell^{\star}$, we can then apply Theorem \ref{thm:parsimonious-cone} to construct connected cell complex 
\begin{equation}
    A=Z^A \to Q_Z^A \to \bar{X}^A 
\end{equation}
with faces, edges, and vertices correspond to $Z$-check, qubits, and $X$-checks, respectively. Similarly, the cocomplex is given by
\begin{equation}
    A^{\top} = X^A \to Q_X^A \to \bar{Z}^A
\end{equation}
where the notation is chosen to be consistent with the main result in Theorem \ref{thm:symplectic-embedding}.

\subsection{Attaching Ancilla to Data}
\subsubsection{Gluing Maps}

Consider the following diagram, copied from Theorem \ref{thm:symplectic-embedding} with the ancilla $A$ and data $D$,
\begin{equation}
    \label{eq:surgery}
    \begin{tikzpicture}[baseline]
    \matrix(a)[matrix of math nodes, nodes in empty cells, nodes={minimum size=25pt},
    row sep=2em, column sep=2em,
    text height=1.25ex, text depth=0.25ex]
    {&& X^{A}  & Q^{A}_X & \bar{Z}^{A}\\
    & S^{D}  & P^{D}  & \bar{S}^{D} &\\
    Z^{A} & Q^{A}_Z & \bar{X}^{A} &&\\};
    \path[->,font=\scriptsize]
    (a-1-3) edge node[above]{$\delta^A$}  (a-1-4)
    (a-1-4) edge node[above]{$\delta^A$}  (a-1-5)
    (a-2-2) edge node[above]{$\sigma^{D}$}  (a-2-3)
    (a-2-3) edge node[above]{$\hat{\sigma}^{D}$}  (a-2-4)
    (a-3-1) edge node[above]{$\partial^A$}  (a-3-2)
    (a-3-2) edge node[above]{$\partial^A$}  (a-3-3);
    \path[->,font=\scriptsize]
    (a-1-3) edge node[right]{$g$}  (a-2-3)
    (a-1-4) edge node[right]{$h$}  (a-2-4)
    (a-2-2) edge node[right]{$\hat{h}$}  (a-3-2)
    (a-2-3) edge node[right]{$\hat{g}$}  (a-3-3);
    \path[->,dashed,black!30!green,font=\scriptsize]
    (a-1-3) edge[bend right=80] node[left]{$p$} (a-3-2)
    (a-1-4) edge[bend left=80] node[right]{$p^{\top}$} (a-3-3);
    \end{tikzpicture}
\end{equation}
Let $g$ be the inclusion map which maps vertex
\begin{equation}
    \label{eq:g-map}
    |q_\alpha\ket^{A} \mapsto q_\alpha, \quad \alpha=X,Z
\end{equation}
And zero otherwise, and $h$ be that which maps the edge
\begin{equation}
    \label{eq:h-map}
    \|q_X\tilde{q}_X;z\ket_X^A \mapsto \bar{z}, \quad \|q_Z\tilde{q}_Z;x\ket_X^A \mapsto \bar{x}
\end{equation}
Then it's straightforward to check that
\begin{lemma}[Chain Map]
    \label{lem:chain-surgery}
    $g,h$ is a chain map
\end{lemma}
\begin{proof}
    Let $\phi^D:S^D \to \bar{S}^D$ denote the isomorphism.
    Note that $\hat{\sigma}^D g$ maps
    \begin{align}
        |q_X\ket^{A} &\xmapsto{g} q_X \\
        &\xmapsto{\lambda^{D}} q_Z \\
        &\xmapsto{\sigma^{D\top}} \sum_{z\sim q} z \\
        &\xmapsto{\phi^{D}} \sum_{z\sim q} \bar{z}
    \end{align}
    where $z\sim q$ is the adjacency relation to the CSS code $D_{\rm{CSS}}$ and $q\in \cQ^D$ is the qubit corresponding to $q_Z$.
    Also note that $\hat{h}\sigma^S$ maps
    \begin{align}
        |q_X\ket^{A} &\xmapsto{\sigma^S} \sum_{z\sim q} \|q_X \tilde{q}_X;z\ket^{A}_X  +\cdots\\
        &\xmapsto{h} \sum_{z\sim q} \bar{z}
    \end{align}
    where $\tilde{q}_X$ is the unique other qubit that was paired up with $q_X$, and $\cdots$ denotes possible other edges adjacent to $|q_X\ket^{A}$, which are irrelevant since $h$ maps those edges to zero.
    The case is similar for $|q_Z\ket^{A}$ and thus the statement follows.
\end{proof}

\subsubsection{Defect Map}

Note that $\hat{g}g$ maps
\begin{align}
    \label{eq:compat-g-start}
    |q_X\ket^{A} &\xmapsto{g} q_X \\
    &\xmapsto{\lambda^{Q}} q_Z \\
    &\xmapsto{g^{\top}} |\bar{q}_Z\ket^{A} 1\{q\in \ell_X\wedge \ell_Z\}
    \label{eq:compat-g-end}
\end{align}
where $q\in \cQ^D$ is the qubit corresponding to $q_X$, and recall the bar notation\footnote{Technically, should be written as $\overline{|q_Z\ket}^A$.} is to indicate a syndrome. The case for $|q_Z\ket^{A}$ is similar.
Hence, $\hat{g}g$ indicates the overlap $\ell_X^{\star}\wedge \ell_Z^{\star} \subseteq \cQ^D$.
If $\ell_X^{\star}\wedge \ell_Z^{\star} =\varnothing$, then $\hat{g}g=0$ and thus a defect map $p$ can be set $=0$.
However, more generally, a defect line similar to the Layer Codes \cite{williamson2023layer,yuan2026unified,yuan2026quantum} must be introduced to define $p$.
In particular, if vertex $q$ (corresponding to $q_X,q_Z$) is in the overlap $\ell_X^{\star}\wedge \ell_Z^{\star}$, then let $p$ map
\begin{equation}
    \label{eq:defect-map}
    |q_X\ket^{A} \mapsto \|q\ket_Z^A
\end{equation}
and be the zero map otherwise.

\begin{lemma}[Compatibility]
    \label{lem:compatible-surgery}
    The gluing maps $g,\hat{g}$ are compatible with the defect map $p$, i.e.,
    \begin{equation}
        \hat{g}g= \partial^A p+ p^{\top} \delta^A
    \end{equation}
\end{lemma}
\begin{proof}
    By Eq. \eqref{eq:compat-g-start}-\eqref{eq:compat-g-end}, we see that if $q\in \ell_X\wedge \ell_Z$ (with corresponding $q_X,q_Z$), then
    \begin{align}
        |q_X\ket^{A} &\xmapsto{\hat{g}g} |\bar{q}_Z\ket^A\\
        |q_Z\ket^{A} &\xmapsto{\hat{g}g} |\bar{q}_X\ket^A
    \end{align}
    Consider the map $\partial^A p$ acting on
    \begin{align}
        |q_X\ket^A &\xmapsto{p} \|q \ket^A_Z \\
        &\xmapsto{\partial^A} |\bar{q}_X \ket^A + |\bar{q}_Z\ket^A
    \end{align}
    And acting on
    \begin{equation}
        |q_Z\ket^{A} \xmapsto{p} 0  \xmapsto{\partial^A} 0
    \end{equation}
    Conversely, consider the map $p^{\top}\delta^{A}$ as follows
    \begin{align}
        |q_X\ket^{A} &\xmapsto{\delta^{A}} \|q\ket^{A}_X +\cdots \\
        &\xmapsto{p^{\top}} |\bar{q}_X\ket^A
    \end{align}
    where technically, $\delta^{A}$ maps $|q_X\ket^{A}$ also to other edges. However, since $p^{\top}$ maps the other edges to zero, they are irrelevant and thus omitted as $\cdots$.
    The map also acts as follows
    \begin{align}
        |q_Z\ket^{A} &\xmapsto{\delta^A} \|q \ket^{A} +\cdots \\
        &\xmapsto{p^{\top}} |\bar{q}_X\ket^{A}
    \end{align}
    Hence, the statement then follows.
\end{proof}

\subsection{Non-CSS Logical Measurement}

By Theorem \ref{thm:symplectic-embedding}, logical measurement follows.

\begin{theorem}[Fig. \ref{fig:surface-cone}]
    \label{thm:surgery}
    Let $D$ be a $\llb n,k,d\rrb$ CSS qLDPC code in Eq. \eqref{eq:data-symplectic}, and $\ell^{\star}=\ell_X^{\star}\oplus \ell_Z^{\star}$ be a logical operator in $D$ with weight $\Weight$.
    Then there exists an ancilla CSS qLDPC code $A$ with $O(\Weight \log \Weight)$ qubits such that if the height-2 cone $C$ is constructed by Diagram \eqref{eq:surgery} with maps $g,h,p$ in Eq. \eqref{eq:g-map}, \eqref{eq:h-map} and \eqref{eq:defect-map}, then $C$ is a non-CSS qLDPC code with 
    \begin{equation}
        n(C)=n+O(\Weight \log \Weight)
    \end{equation}
    qubits.
    $C$ also has $k-1$ logical qubits, or more specifically,
    \begin{equation}
        H_1(C) \cong [\ell^{\star}]^{\perp}/ [\ell^{\star}]
    \end{equation}
    where $[\ell^{\star}]^{\perp}$ is the symplectic complement of $[\ell^{\star}]$ in $H_1(D)$ with respect to the induced symplectic form $[\Lambda^D]$.
    $C$ also has code distance
    \begin{equation}
        d(C)=\Omega(1) d(D)
    \end{equation}
    where the $\Omega(1)$ depends on the Cheeger constant of the measurement graph in Definition \ref{def:measurement-graph}, and $d(D)$ is that of $D$.
\end{theorem}
\begin{proof}
    Note that by Theorem \ref{thm:parsimonious-cone}, we have $H^0(A)\cong \F_2$ with unique basis element $[\cV]$ given by the (unique) connected component of the measurement graph.
    Note that
    \begin{equation}
        [\ell^{\star}] = [g][\cV]
    \end{equation}
    Conversely, note that $\ker[\hat{g}]=[\ell^{\star}]^{\perp}$ and thus the statement follows from Theorem \ref{thm:symplectic-embedding} via the fact that $H_1(A)=0$.
    Proof of the code distance is standard by utilizing the Cleaning Lemma in \cite{yuan2026unified}, see, e.g., Theorem V.1 of \cite{yuan2026parsimonious} or Theorem IV.2 of \cite{yuan2026quantum}.
\end{proof}

\begin{remark}[Recovery and Threshold, Fig. \ref{fig:surface-cone-stabilizer}]
    \label{rem:recovery-and-threshold}
    Note that in contrast to conventional CSS measurement, the original logical $\ell^{\star}$ is not a stabilizer of the deformed code $C$.
    Instead, as shown in Fig. \ref{fig:surface-cone-stabilizer}, $\ell^{\star}$ + some \textit{defects} of $Z$-type on the ancilla forms a stabilizer. Specifically, it contains the following stabilizer
    \begin{equation}
        \sigma \cV = 
        \begin{pmatrix}
            0 \\
            \ell^{\star} \\
            p\cV
        \end{pmatrix}
    \end{equation}
    where $\cV\in X^A$ is the unique connected component (sum over all vertices) of the ancilla cell complex, and thus in addition to $\ell^\star$, it also acts as $Z$-Paulis corresponding to $p\cV$ on the ancilla.
    Therefore, to recover the measurement of $\ell^{\star}$, the ancilla qubits must be measured out transversally in the $Z$-basis.
    Measurement of stabilizer $\sigma \cV$ + measuring out the ancilla will provide the information of the logical measurement of $\ell^{\star}$.

    To show that this process is fault tolerant (possesses a threshold), a corresponding spacetime fault complex is built from first principles (see Appendix \ref{sec:non-CSS-measure}).
    Note that even in conventional measurement of, say, $X$-type logical, this was necessary to recover the original $Z$-type stabilizers of the data code (see Appendix \ref{sec:CSS-measure}).
\end{remark}
\subsection{Non-CSS Logical on Non-CSS Code}
\label{sec:ext-non-CSS}
In this subsection, we shall extend the formulation slightly to incorporate non-CSS logical measurement on any non-CSS code.
Specifically, let a non-CSS data code be written in the form of a symplectic complex
\begin{equation}
    D=S^D \xrightarrow{\sigma^D} P^D \xrightarrow{\hat{\sigma}^D}\bar{S}^D
\end{equation}
where $S^D,\bar{S}^D$ are isometric, and $P^D=Q_X^D\oplus Q_Z^D$ denotes the corresponding symplectic basis so that $\Lambda^D$ is the symplectic form.
Given a non-CSS logical of $D$, we can thus define a measurement graph in an analogous manner.
\begin{definition}[Measurement (Multi-)Graph]
    \label{def:measurement-graph-non-CSS}
    Let $\ell^{\star}=\ell_X^{\star}\oplus \ell_Z^{\star}$ denote a logical of $D$, i.e., $\ell^{\star} \in \ker \hat{\sigma}^D$.
    Let $\cV_X$ denote a copy of collection $\ell_X^{\star}\subseteq \cQ_X^D$ so that vertices are denoted by $|q_X\ket$ for $q_X\in \ell_X^{\star}$ and similarly define $\cV_Z$ so that $\cV=\cV_X \sqcup \cV_Z$.
    Consider the collection of checks $s\in S$ which \textbf{nontrivially overlaps} with $\ell^{\star}$, i.e., anti-commutes with $\ell^{\star}$ at some qubit
    \begin{equation}
        \Lambda^D(\ell^{\star}|_{q},(\sigma^D s)|_q)=1
    \end{equation}
    where $\ell|_{q}=\ell_X|_{q_X}\oplus \ell_Z|_{q_Z}$, and  $q_X,q_Z$ are the $X,Z$-sector qubits corresponding to $q$.
    In this case, we see that one and only one of the following can be true
    \begin{itemize}
        \item[($X$)] $q_X\in \ell^{\star}$ and $q_Z \in \sigma^D s$
        \item[($Z$)] $q_Z\in \ell^{\star}$ and $q_X\in \sigma^D s$
    \end{itemize}
    Note that $\ell^{\star},\sigma^D s$ must commute. Hence, there exists an even number of $q$ such that $\ell^{\star},\sigma^D s$ overlap (anti-commute) at $q$, and thus can be arbitrarily paired up.
    Given a pair $qq'$ in the overlap, determine whether $q$ is in scenario $(X)$ or $(Z)$, and similarly for $q'$, which we denote by $(X'),(Z')$.
    If in scenario $(\alpha)\cap(\beta')$, then add an edge between vertices in $\cV_\alpha\sqcup\cV_\beta$ corresponding to $q_\alpha\lr q_\beta'$.
    In any case, denote the edge as $\|qq';s\ket$, and let $\cE_S$ denote the collection of $\|qq';s\ket$ over all pairs and overlapping checks $s$.
    
    If $\ell_X^{\star},\ell_Z^{\star}$ regarded as subsets of $\cQ^D$ has nontrivial overlap $\ell_X^{\star}\wedge \ell_Z^{\star} \ne \varnothing$, then add an edge connecting $|q_X\ket\in \cV_X$ and $|q_Z\ket\in \cV_Z$ where $q_X\lr q_Z$ correspond to the same qubit in the overlap $\ell_X^{\star}\wedge \ell_Z^{\star}$.
    Let $\cE_{X\wedge Z}$ denote the collection of such edges.

    Further add arbitrary edges $\cE_{\rm{extra}}$ while keeping the max degree fixed so that the resulting (multi-)graph $\cG=(\cE,\cV)$ is connected and expanding with $\Omega(1)$ Cheeger constant\footnote{This can be achieved by adding random edges, or using more sophisticated methods~\cite{lubotzky1988ramanujan,hoory2006expander,alon2008elementary}.}.
    Then $\cG$ is the \textbf{measurement (multi-) graph} associated with $\ell^{\star}$, with associated graph complex $G=E\to V$.
\end{definition}

Note that 
\begin{align}
    \cV&= \cV_X \sqcup  \cV_Z \\
    \cE&= \cE_S\sqcup  \cE_{X\wedge Z} \sqcup  \cE_{\rm{extra}}
\end{align}
Also note that the max degree $\Delta(G) \le \qubit+1$ where $\qubit$ is the total qubit degree of data code $D$.
The remainder of the construction is then exactly as non-CSS logical measurement on CSS data code, i.e., starting from Theorem \ref{thm:parsimonious-cone} and culminating in Theorem \ref{thm:surgery} (except that the data code $D$ is now non-CSS qLDPC).

\section{Non-CSS Layer Codes}
\label{sec:layer}

In this section, we show that using the non-CSS embedding formalism in Theorem \ref{thm:symplectic-embedding}, the CSS Layer Codes in \cite{williamson2023layer} can be generalized nearly immediately to non-CSS (i.e., the input code can be non-CSS qLDPC), by a simple observation of the geometry of the layers and defects.
Specifically, the check layers should be the union of $X$- and $Z$-check layers, as elaborated in Definition \ref{def:raw-check-layer}.
The generalization also applies to the most recent breakthrough of 4D and 5D CSS Layer Codes \cite{yuan20264d}, though for the sake of simplicity, we shall restrict our attention to the conventional 3D Layer Codes.
To guide the reader, we shall consider the example input code with checkes $XIZ,ZZX$ on three qubits, as shown in Fig. \ref{fig:layers} and \ref{fig:interactions}.

\subsection{Input Code}
\label{sec:layer-input}
Similar to the setup in \cite{yuan2026unified}, let the non-CSS input code $B$ be written as a symplectic complex
\begin{equation}
    B = S \xrightarrow{\sigma} P\xrightarrow{\hat{\sigma}}\bar{S}
\end{equation}
where $S,\bar{S}$ are isometric with checks $s\in \cS$ corresponding to syndromes $\bar{s}\in \bar{\cS}$, and $P=Q_X\oplus Q_Z$ so that $X,Z$-sector qubits $q_X\in \cQ_X, q_Z\in \cQ_Z$ are in one-to-one correspondence with the physical qubits $q\in \cQ$.
Let $\weight,\qubit$ denote the maximum check weight and total qubit degree of $B$, respectively.

Arrange the qubits $|\cQ|=n$ along a 1D grid $[n]$ with axes $\hat{q}$ in a fixed manner.
Since $Q_X, Q_Z$ are the $X,Z$-sectors of qubits $Q$, we shall use $q_X,q_Z$ to denote the corresponding basis elements with coordinate $q\in [n]$.
Let $R_Q$ denote the corresponding repetition code on $|\cQ|=n$ vertices, which one may view as arranged along the $\hat{q}$ direction.
Let $\hat{x},\hat{z}$ denote the transverse directions $\perp \hat{q}$, and let $n_S=|\cS|$ denote the number of parity checks.
Let $R_X,R_Z$ denote repetition codes on $n_S$ vertices, which one may view as arranged along the $\hat{x},\hat{z}$ direction, respectively.

Moreover, as we shall soon see in Section \ref{sec:defects-3D}, to decongest the defect lines, we also need additional space $\weight \qubit n_S$, corresponding to an upper bound on the number of overlapping pairs $ss'$ of checks $s\sim s'$.
Let $N=n+\weight \qubit n_S=O(\weight \qubit)n$ so that the 3D grid is given by
\begin{equation}
    [n_S]\times [N] \times [n_S] \in \hat{x}\times \hat{q}\times \hat{z}
\end{equation}
In particular, let $j(ss')\in (n,N]$ denote the coordinate corresponding to overlapping pair $ss'$. Note that only the first $n$ coordinates correspond to qubits $q\in \cQ$.

\subsection{Qubit Layers}
\label{sec:layer-qubit}
Similar to the conventional CSS Layer Codes, the qubit layers are defined simple as
\begin{definition}[Qubit Layers]
    As depicted in Fig. \ref{fig:qubit-layers}, define the qubit layer as the 3-term cell complex
    \begin{equation}
        R_X\otimes q\otimes R_Z^{\top}
    \end{equation}
    where $q\in \cQ=[n]$.
    By Example \ref{ex:CSS-symplectic}, the qubit layer is a CSS code, which can be written as a symplectic complex $D(q)$.
\end{definition}

\subsection{Check Layers}
\label{sec:layer-check}
\begin{figure}[ht]
\subfloat[\label{fig:qubit-layers}]{%
    \centering
    \includegraphics[width=0.2\columnwidth]{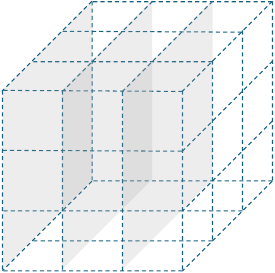}
}
\quad
\subfloat[\label{fig:check-1-layer}]{%
    \centering
    \includegraphics[width=0.2\columnwidth]{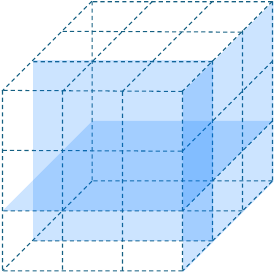}
}
\quad
\subfloat[\label{fig:check-2-layer}]{%
    \centering
    \includegraphics[width=0.2\columnwidth]{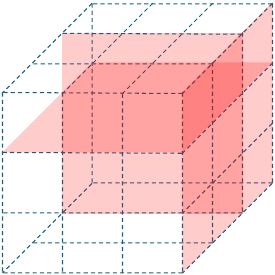}
}
\caption{Layers. Consider the example input code with checks $XIZ,ZZX$. (a) Qubit layers, one for each qubit, (b), (c) denote the check layer corresponding to $XIZ,ZZX$, respectively.
}
\label{fig:layers}
\end{figure}
Recall that in the conventional CSS Layer Codes, each $X,Z$-check layer is simply a 2D plane with fixed $\hat{x},\hat{z}$ coordinate.
A straightforward extension to non-CSS code is to regard each $s$ check layer as the \textit{union} of some $X,Z$-type check layers.
\begin{definition}[$XZ$ Graph]
    \label{def:XZ-graph}
    Arrange the checks $\cS$ along the diagonal of the $\hat{x}\times \hat{z}$ plane in an arbitrary but fixed manner, so that each $s\in \cS$ has a fixed coordinate $\in [n_S]$ so that $(s,s)\in [n_S]^2$ is along the diagonal.
    Define the \textbf{XZ graph} of check $s$ as the collection of rows and columns in the $\hat{x}\times \hat{z}$ plane passing through $(s,s)$, i.e.,
    \begin{equation}
        \XZ(s) = (s\times [n_S]) \cup ([n_S]\times s)
    \end{equation}
    where we shall treat $\XZ(s)$ as a graph complex with the corresponding edges and vertices.
    Note that the XZ graphs $\XZ(s),s\in \cS$ are edge decongested in $\hat{x}\times \hat{z}$, i.e., no edge belongs to more than one XZ graph.
\end{definition}

\begin{definition}[Raw Check Layer]
    \label{def:raw-check-layer}
    Let $R_Q$ be the repetition code on $[N]$ vertices and define the \textbf{(raw) check layer} of check $s\in \cS$ as the 3-term cell complex
    \begin{equation}
        A_{\mathrm{raw}}(s) =\XZ(s) \otimes R_Q
    \end{equation}
    and thus correspond to union of $X,Z$-check layers in conventional CSS Layer Codes \cite{williamson2023layer}.
\end{definition}

\begin{remark}[Raw Check Layer Congestion]
    \label{rem:raw-check-layer-congestion}
    Note that $\XZ(s),s\in \cS$ are edge decongested, but may intersect at vertices.
    However, due to the degree of the 2D grid $\hat{x}\times \hat{z}$, each vertex can belong to at most two check layers.
    Hence, the raw check layers $A_{\mathrm{raw}}(s),s\in \cS$ have edge congestion as follows
    \begin{itemize}
        \item 1-congestion for edges in the $\hat{x}\times \hat{z}$ plane
        \item 2-congestion for edges in the $\hat{q}$ axis
    \end{itemize}
\end{remark}

As depicted in Fig. \ref{fig:check-check-inter}, due to the overlap of distinct (raw) check layers,  additional modifications are necessary so that the defects can be paired up locally in the 3D grid.
\begin{definition}
    \label{def:check-layer}
    As depicted in Fig. \ref{fig:check-1-layer} and \ref{fig:check-2-layer}, define the \textbf{check layer} of check $s\in \cS$ as the union
    \begin{equation}
        A(s) = A_{\mathrm{raw}}(s) + \bigoplus_{s'\sim s} A(s';s)
    \end{equation}
    where the summation of $s'$ is over all checks which overlap in the support with $s$ relative to the input code $B$, and
    \begin{equation}
        A(s';s)= R_X\otimes j(ss')\otimes R_Z
    \end{equation}
    For notation, label the vertices in $A(s)$ as $|xqz;s\ket$ where $xqz$ are coordinates in the 3D grid. Similarly, we can label edges along the $\hat{x}$ direction as $|x^{\pm} qz;s\ket$ where $x^\pm =x\pm 1/2$ are half-integers and similarly, for $|xq^{\pm} z;s\ket$ and $|xqz^{\pm};s\ket$.

\end{definition}

\begin{remark}[Check Layer Congestion]
    \label{rem:check-layer-congestion}
    Note that the addition of layers $A(s',s)$ only affects the congestion for edges within the $\hat{x}\times \hat{z}$ at $\hat{q}$ coordinate in $(n,N]$. In particular, the check layers $A(s),s\in \cS$ have 2 edge congestion.
\end{remark}
\subsection{Attaching Check and Qubit Layers}
\label{sec:layer-attach}

As suggested by Theorem \ref{thm:symplectic-embedding}, define
\begin{align}
    \label{eq:CS-3D}
    A &= \bigoplus_{s\in \cS^B} A(s) \\
    D &= \bigoplus_{q\in \cQ^B} D(q)
    \label{eq:CP-3D}
\end{align}
Note that collections $A(s),s\in \cS$ and $D(q),q\in \cQ$ have finite congestion and thus embedded in the 3D grid.
Hence, it remains to show that the maps $g,h,p$ in Theorem \ref{thm:symplectic-embedding} are local in 3D.
We copy the diagram here for ease of reference.
\begin{equation}
    \begin{tikzpicture}[baseline]
    \matrix(a)[matrix of math nodes, nodes in empty cells, nodes={minimum size=25pt},
    row sep=2em, column sep=2em,
    text height=1.25ex, text depth=0.25ex]
    {&& X^{A}  & Q^{A}_X & \bar{Z}^{A}\\
    & S^{D}  & P^{D}  & \bar{S}^{D} &\\
    Z^{A} & Q^{A}_Z & \bar{X}^{A} &&\\};
    \path[->,font=\scriptsize]
    (a-1-3) edge node[above]{$\delta^A$}  (a-1-4)
    (a-1-4) edge node[above]{$\delta^A$}  (a-1-5)
    (a-2-2) edge node[above]{$\sigma^{D}$}  (a-2-3)
    (a-2-3) edge node[above]{$\hat{\sigma}^{D}$}  (a-2-4)
    (a-3-1) edge node[above]{$\partial^A$}  (a-3-2)
    (a-3-2) edge node[above]{$\partial^A$}  (a-3-3);
    \path[->,font=\scriptsize]
    (a-1-3) edge node[right]{$g$}  (a-2-3)
    (a-1-4) edge node[right]{$h$}  (a-2-4)
    (a-2-2) edge node[right]{$\hat{h}$}  (a-3-2)
    (a-2-3) edge node[right]{$\hat{g}$}  (a-3-3);
    \path[->,dashed,black!30!green,font=\scriptsize]
    (a-1-3) edge[bend right=80] node[left]{$p$} (a-3-2)
    (a-1-4) edge[bend left=80] node[right]{$p^{\top}$} (a-3-3);
    \end{tikzpicture}
\end{equation}

\begin{figure}[ht]
\subfloat[\label{fig:qubit-check-1-inter}]{%
    \centering
    \includegraphics[width=0.2\columnwidth]{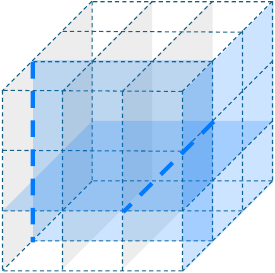}
}
\quad
\subfloat[\label{fig:qubit-check-2-inter}]{%
    \centering
    \includegraphics[width=0.2\columnwidth]{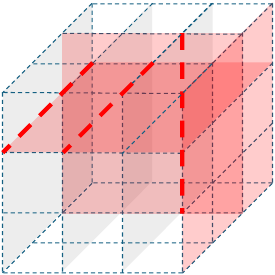}
}\quad
\subfloat[\label{fig:check-check-inter}]{%
    \centering
    \includegraphics[width=0.2\columnwidth]{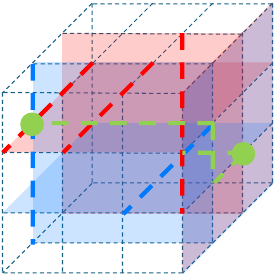}
}
\caption{Interactions. Continue with the example input code with checks $XIZ,ZZX$ in Fig. \ref{fig:layers} (a), (b) denote $XIZ$ and $ZZX$ check layers interacting with the qubit layers, respectively. (c) denotes the defect (green) between the $XIZ$ and $ZZX$ check layers.
}
\label{fig:interactions}
\end{figure}

\subsubsection{Gluing Maps}

As depicted by Fig. \ref{fig:qubit-check-1-inter} and \ref{fig:qubit-check-2-inter}, we define the gluing map $g$ via
\begin{align}
    \label{eq:g-3D}
    |xqz;s\ket^A &\mapsto  |xqz\ket^D_X 1\{q_X \sim s=x\} \\
    &\quad +|xqz\ket^D_Z 1\{q_Z \sim s=z\}
\end{align}
(and zero otherwise), where $\sim$ is with respect to the input code $B$, e.g., $q_X\sim s$ if $q_X\in \sigma s$ and $q_X$ are the $X$- qubit in $\cQ_X$ corresponding to $q$ (the case is similar for $q_Z$).
Further define the gluing map $h$ via
\begin{align}
    \label{eq:h-3D}
    |xqz^+;s\ket^{A} &\mapsto |xqz^+\ket^{D}_X 1\{q_X\sim s=x\} \\
    |x^+ qz;s\ket^{A} &\mapsto |x^+q z\ket^{D}_Z 1\{q_Z\sim s=z\}
\end{align}
(and zero otherwise), where $x^{\pm} =x\pm 1/2$ are used to denote the edges along $\hat{x}$ and similarly for $z^\pm$.
Note that $g,\hat{h}$ map cells to those with the same coordinate, and thus are local in 3D.
It's then straightforward to check that

\begin{lemma}[Chain Map]
    \label{lem:chain-3D}
    The gluing maps $(g,h)$ in Eq. \eqref{eq:g-3D} and \eqref{eq:h-3D} form a chain map.
\end{lemma}
\begin{proof}
    Note that
    \begin{align}
        \hspace{-6pt}|xqz;s\ket^{A} &\xmapsto{g} |xqz\ket_X^{D} 1\{q_X\sim s=x\}\\
        &\quad + |xqz\ket^{D}_Z 1\{q_Z\sim s=z\}  \\
        &\xmapsto{\hat{\sigma}^{D}}(|xqz^+\ket_X^{D} +|xqz^-\ket_X^{D})  1\{q_X\sim s=x\}\\
        &\quad + (|x^+ qz\ket^{D}_Z + |x^- qz\ket^{D}_Z) 1\{q_Z\sim s=z\}
    \end{align}
    Also note that
    \begin{align}
        \hspace{-6pt}|xqz;s\ket^{A} &\xmapsto{\delta^A} \left(|xqz^+;s\ket^{\top} + |xqz^-;s\ket^{\top}\right)1\{s=x\} \\
        &\quad + \left(|x^+qz;s\ket^{\top} + |x^-qz;s\ket^{\top}\right)1\{s=z\} \\
        &\xmapsto{h} \left(|x^+qz\ket_X^D + |x^-qz\ket_X^D\right)1\{q_X\sim s=x\} \\
        &\quad + \left(|x^+qz\ket_Z^D + |x^-qz\ket_Z^D\right)1\{q_Z \sim s=z\}
    \end{align}
    Hence, $(g,h)$ is a chain map.
\end{proof}

\subsubsection{Defect Maps}
\label{sec:defects-3D}
Note that $\hat{g}g$ is derived in Eq. \eqref{eq:compat-g-start}-\eqref{eq:compat-g-end}, and as depicted by the green dots in Fig. \ref{fig:check-check-inter}, indicate the overlap between checks.
Hence, similar to the CSS Layer codes, the goal is to find defect lines which pair up the overlaps.
Concretely, the defect map $p$ has one job: whenever two checks $s,s'$ act with clashing Pauli type on a shared qubit -- one contributes an $X$ Pauli there, the other a $Z$ Pauli -- Theorem \ref{thm:symplectic-embedding} requires this clash to be absorbed as $\hat{g}g=\partial^A p+p^{\top}\delta^A$, a chain-homotopy condition. Building such a $p$ amounts to drawing a short wire in the 3D grid between every pair of clashing qubits, as shown by the green line in Fig. \ref{fig:check-check-inter}, dressed so that the wire's two endpoints are exactly what survives the telescoping sum $\partial^A p+p^{\top}\delta^A$. The rest of this subsection is the construction of these wires.

For every pair of checks $s,s'$ of the input code $A$ with overlapping support, let us choose an arbitrary order.
For example, for CSS input codes, the implicit choice was to always choose the order from $X$-check $x$ to $Z$-check $z$ given overlapping pair $x\sim z$, while the order of checks of the same type, $xx'$ or $zz'$, was neglected since they do not induce defect lines \cite{williamson2023layer}. Since they commute $\Lambda(\sigma s,\hat{\sigma}s)=0$, we see that
\begin{equation}
    \sum_{q\in \cQ} 1\{q_X\sim s,q_Z\sim s'\} = \sum_{q\in \cQ} 1\{q_Z\sim s,q_X\sim s'\}
\end{equation}
where $q_X,q_Z$ are the $X,Z$-qubits corresponding to $q$.
In particular, the left side denotes the qubits at which $s$ acts as $X$ or $Y$ Pauli and $s'$ acts as a $Z$ or $Y$ Pauli, whose collection we denote as $s\rh{\wedge} s' \subseteq \cQ$.
Similarly, the right side denotes those at which $s$ acts as $Z$ or $Y$ Pauli and $s'$ acts as a $X$ or $Y$ Pauli, and thus we denote the collections of qubits as $s\lh{\wedge} s' \subseteq \cQ$.
Read the arrow as the direction of the clash: $s\rh{\wedge}s'$ points from $s$ (contributing the $X$-half) to $s'$ (contributing the $Z$-half), and $s\lh{\wedge}s'$ is the same clash read the other way.
The commutation relations is thus equivalent to
\begin{equation}
    |s\rh{\wedge} s'|=|s\lh{\wedge} s'|
\end{equation}

For the running example $s=XIZ,s'=ZZX$: at qubit 1, $s$ acts $X$ and $s'$ acts $Z$, a forward clash; qubit 2 is silent for both; at qubit 3, $s$ acts $Z$ and $s'$ acts $X$, a backward clash. Hence $s\rh{\wedge}s'=\{1\}$ and $s\lh{\wedge}s'=\{3\}$, both of size one -- this pair already sits in the odd case treated below, which is exactly why it was chosen as the running example.

Let us now separately consider the case where $|s\rh{\wedge}s'|=|s\lh{\wedge} s'|$ is even or odd.
In the simpler case where both are even\footnote{This occurs, for example, when the input code is CSS}, the qubits in $s\rh{\wedge} s'$ can be ordered based on the total ordering on $\cQ=[n]$, so that
\begin{equation}
    s\rh{\wedge}s' = (q_{1} < \cdots  <q_{2m})
\end{equation}
A (defect) path can then be drawn between pairs of qubits, i.e., $\gamma_i(s s') = [q_{2i-1},q_{2i}]$, which corresponds to the \textbf{defect path} in the 3D grid as depicted by the green line in Fig. \ref{fig:check-check-inter}, and explicitly, defined as
\begin{equation}
    \Gamma_i(s s') = s\times \gamma_i(s s') \times s', \quad i\in [m]
\end{equation}
We further parametrize the path $[0,\tau_i]\ni t\mapsto \gamma_i(t;s s')$ so that the path is from $q_{2i-1} \to q_{2i}$, and similarly, for $\Gamma_i(t;s s')$.
Note that half-integers $\Gamma_i(t^+;ss')$ denote the edges of the path.
We then perform the same construction of the defect paths for $s\lh{\wedge}s'$.
In particular,
\begin{equation}
    s\lh{\wedge}s'=(\bar{q}_{1} <\cdots <\bar{q}_{2m})
\end{equation}
and $\bar{\gamma}_{i}(ss')=[\bar{q}_{2i-1},\bar{q}_{2i}]$ and\footnote{Note that the role of $s,s'$ are reversed}
\begin{equation}
    \bar{\Gamma}_{i}(ss') = s'\times \bar{\gamma}_{i}(ss')\times s, \quad i\in [m]
\end{equation}
and $[0,\bar{\tau}_i]\ni t\mapsto \bar{\gamma}_{i}(t;ss')$ is parametrized so that the path is from $\bar{q}_{2i-1} \to \bar{q}_{2i}$, and similarly for $\bar{\Gamma}_{i}(t;ss')$

Let us now consider the sightly more difficult case where both are odd (which does not occur for CSS codes).
Similar to the even case, let
\begin{align}
    s\rh{\wedge}s' &= (q_1<\cdots <q_{2m}<q_0) \\
    s\lh{\wedge}s' &= (\bar{q}_{1} <\cdots < \bar{q}_{2m}<\bar{q}_0)
\end{align}
Define the defect paths $\Gamma_i,\bar{\Gamma}_i$ as in the even case for $i\ne 0$, so that we may restrict our attention to $i=0$.

The reason $i=0$ needs special treatment is geometric: $q_0$ is a forward-clash qubit, so it lies on the line $(x,z)=(s,s')$, while $\bar{q}_0$ is a backward-clash qubit, lying on the different line $(x,z)=(s',s)$. Since a single grid move changes only one coordinate, no straight path connects two points on different $(x,z)$-lines -- unlike the pairs $i\ne 0$, whose two endpoints share the same crossing direction and hence the same line. This is exactly what the extra coordinate $j(ss')$ is for: at $q=j(ss')$, the check layer spans every $(x,z)$ pair at once, so the path can travel in along one line, pivot freely at this coordinate, and travel back out along the other. In the running example, this connects $\bar{q}_0=$ qubit 3 (on the line $(x,z)=(s',s)$) to $q_0=$ qubit 1 (on the line $(x,z)=(s,s')$) via the pivot at $j(ss')$, exactly the defect line depicted in Fig. \ref{fig:check-check-inter}.

Define the time-parametrized defect path $\Gamma_0(t;ss')$ defined as follows
\begin{equation}
    \label{eq:non-CSS-defect-3D}
    s'\bar{q}_0 s\to s' j(s's) s \to s j(ss') s \to s j(ss') s'\to  s  q_0 s'
\end{equation}
where at each step, only one coordinate changes and thus is simply a straight line.
Note that this is the only time where the additional coordinates $(n,N]$ along the $\hat{q}$ axis are utilized.

Further define $I_{i}(\rh{s s}')$ to be the interval $[0,\tau_i)$ if $s\to s'$ is consistent with the order between pair $s,s'$ and $=(0,\tau_i]$ if inconsistent. Similarly, define $\sgn \rh{s s}'=+$ if $s\to s'$ is consistent, and $-$ if inconsistent.
We can thus define the defect map $p$ as follows
\begin{align}
    \label{eq:p-3D}
    |xqz;s\ket^{A} &\mapsto \sum_{s'\sim s} \sum_{i}\sum_{t\in I_i(\rh{ss}')} 1\{xqz=\Gamma_i(t;ss')\} \\
    &\quad\quad \times |\Gamma_i(t^{\sgn \rh{s s}'};ss');s'\ket^A_Z
\end{align}
where the summation is over all $s'$ such that $s,s'$ have overlapping support (denoted as $s'\sim s$), and we treat $\bar{\Gamma}_i=\Gamma_{-i}$ so that the summation over $i$ is such that $|i|\le m$, and the notation emphasizes $|\cdots \ket_Z^A \in Q_Z^A$.
Note that the defect map $p$ maps a cell with coordinate $xqz=\Gamma_i(t;ss')$ to that with coordinate of a neighboring edge $\Gamma_i(t^{\pm};ss')$ in the 3D grid, and thus is local.

\begin{lemma}[Compatibility]
    \label{lem:compatible-3D}
    The gluing maps $g,\hat{g}$ are compatible with the defect map $p$, i.e.,
    \begin{equation}
        \hat{g}g= \partial^A p+p^{\top}\delta^A
    \end{equation}
\end{lemma}
\begin{proof}
    Note that $\hat{g}g$ maps the following
    \begin{align}
        \label{eq:compat-g-start}
        |xqz;s\ket^{A} &\xmapsto{g} |xqz\ket^D_X 1\{q_X\sim s=x\} \\
        &\quad +|xqz\ket^D_Z 1\{q_Z\sim s=z\}\\
        &\xmapsto{\lambda^{D}} |xqz\ket^D_Z 1\{q_X\sim s=x\} \\
        &\quad +|xqz\ket^D_X 1\{q_Z\sim s=z\}\\
        &\xmapsto{g^{\top}} \sum_{s'\in \cS} |xqz;\bar{s}'\ket^A \\
        &\quad\times (1\{q_X\sim s=x,q_Z\sim s'=z\}  \\
        &\quad\quad + 1\{q_Z\sim s=z, q_X\sim s'=x\} )
        \label{eq:compat-g-end}
    \end{align}
    where the $|\cdots; \bar{s}'\ket^A$ notation emphasizes that the element is in the syndrome $\bar{X}^A$.
    Let us consider the defect maps. Specifically, note that $\partial^A p$ maps
    \begin{align}
        |xqz;s\ket^{A} &\xmapsto{p} \sum_{s'\sim s} \sum_{i,t\in I_i(\rh{s s}')} 1\{xqz=\Gamma_i(t;ss')\} \\
                       &\quad\quad \times |\Gamma_i(t^{\sgn{\rh{s s}'}};ss');s'\ket_Z^A \\
                       &\xmapsto{\partial^A} \sum_{s'\sim s} \sum_{i,t\in I_i(\rh{s s}')} 1\{xqz=\Gamma_i(t;ss')\} \times \\
                       &\hspace{-15pt}(|\Gamma_i(t;ss');\bar{s}'\ket^A +|\Gamma_i(t+{\sgn{\rh{s s}'}};ss');\bar{s}'\ket^A)
    \end{align}
    Similarly, note that $p^{\top}\delta^A$ maps
    \begin{align}
        |xqz;s\ket^A &\xmapsto{\delta^A} \sum_{e \sim_s xqz} |e;s\ket^{A}_X\\
                       &\xmapsto{p^{\top}} \sum_{s'\sim s} \sum_{i,t\in I_i(\lh{s s}')}  |\Gamma_i(t;ss');\bar{s}'\ket^{A}  \\
                       &\quad\quad \times  \sum_{e\sim xqz} 1\{e=\Gamma_i(t^{\sgn{\lh{s s}'}};ss')\}
    \end{align}
    where $e\sim_s xqz$ denote edges in cell complex $A(s)$ (contained in the 3D grid) adjacent to coordinate $xqz$, i.e., $e=x^\pm qz,xq^\pm z,xqz^\pm $.
    Note that by utilizing the fact that $\lh{s s}'$ is the opposite direction of $\rh{s s}'$, we can rewrite the map $p^{\top}\delta^A$ as
    \begin{align}
        \hspace{-5pt}
        |xqz;s\ket^{A} &\mapsto \sum_{s'\sim s} \sum_{i,t\in I_i(\rh{s s}')}  |\Gamma_i(t+\sgn \rh{s s}';ss');\bar{s}'\ket^{A}  \\
                       &\quad\quad \times  \sum_{e\sim xqz} 1\{e=\Gamma_i(t^{\sgn{\rh{s s}'}};ss')\}
    \end{align}
    Fix $s'\sim s$, and assume that $\rh{s s}'$ is consistent with the chosen direction.
    If $xqz$ is at the boundary of some defect line, say $xqz=\Gamma_i(0;ss')$. Then we see that $\partial^A p$ maps
    \begin{equation}
        |xqz;s\ket^{A} \mapsto |\Gamma_i(0;ss');\bar{s}'\ket^A + |\Gamma_i(1;ss');\bar{s}'\ket^A +\cdots
    \end{equation}
    where $\cdots$ denotes cells in other $\ne s'$ check layers, while $p^{\top}\delta^A$ maps
    \begin{equation}
        |xqz;s\ket^{A} \mapsto |\Gamma_i(1;ss');\bar{s}'\ket^A +\cdots
    \end{equation}
    Hence, the sum $\partial^A  p + p^{\top}\delta^A $ maps
    \begin{equation}
        |xqz;s\ket^{A} \mapsto |xqz;\bar{s}'\ket^A +\cdots
    \end{equation}
    Similarly, if $xqz=\Gamma_i(\tau_i;ss')$, then we see that $\partial^A p$ maps
    \begin{equation}
        |xqz;s\ket^{A} \mapsto 0 +\cdots
    \end{equation}
    while $p^{\top}\delta^A$ maps
    \begin{equation}
        |xqz;s\ket^{A} \mapsto |\Gamma_i(\tau_i;ss');\bar{s}'\ket^A +\cdots
    \end{equation}
    Hence, the sum $\partial^A  p + p^{\top}\delta^A $ maps
    \begin{equation}
        |xqz;s\ket^{A} \mapsto |xqz;\bar{s}'\ket^A +\cdots
    \end{equation}
    The other cases can be similarly verified to show that $$\partial^A  p + p^{\top}\delta^A $$ maps
    \begin{equation}
        |xqz;s\ket^A \mapsto \sum_{s'\sim s} |xqz;\bar{s}'\ket^A 1\{xqz \in \partial\Gamma(ss')\}
    \end{equation}
    where we regard $\Gamma(ss')$ as the union of all defect lines $\Gamma_i(ss')$.
    Since the boundary $\partial \Gamma(ss')$ is exactly the overlap between $s,s'$ checks, we see that the statement follows.
\end{proof}

\subsection{Main Results}

\begin{theorem}[non-CSS Layer Codes]
    \label{thm:layer-3D}
    Let $A,D$ be defined as in Eq. \eqref{eq:CS-3D}-\eqref{eq:CP-3D}.
    Let gluing maps $g,\hat{h}$ be defined as in Eq. \eqref{eq:g-3D}-\eqref{eq:h-3D}, and defect map $p$ in Eq. \eqref{eq:p-3D}.
    Then the output constructed via Theorem \ref{thm:symplectic-embedding}, referred as the \textbf{3D non-CSS Layer Code}, $C$ is a well-defined symplectic complex embedded in the 3D grid $[n_S]\times [O(\weight\qubit)n]\times [n_S]$ where each edge hosts at most $3$ qubits, and with column code equal to input code $=B$ so that
    \begin{equation}
        H_1(C) \cong H_1(B)
    \end{equation}
    And maximum weight 9 and total qubit degree 8, and
    \begin{equation}
        d(C) =\Omega\left(\frac{n_S}{\weight \qubit} \right) d(B)
    \end{equation}
    where $d(C),d(B)$ are the code distance of the output $C$ and input $B$ codes, respectively.
\end{theorem}


\begin{proof}
    The proof is similar to \cite{yuan2026unified,yuan2026quantum,yuan20264d} where the key difference is that for non-CSS codes, we utilize the main result in Theorem \ref{thm:symplectic-embedding} (combined with Lemma \ref{lem:chain-3D} and \ref{lem:compatible-3D}).
    The lower bound on the code distance $d(C)$ follows similarly as in \cite{yuan2026quantum,yuan20264d}, where we utilize relative expansion and the Cleaning Lemma in \cite{yuan2026unified}, and thus omitted.

    Note that the number of qubits each edge hosts follows from Remark \ref{rem:check-layer-congestion}, and the fact the qubit layers are edge decongested.
    It may be worth mentioning that the weights of the non-CSS Layer Code is depicted by the following diagram
    \begin{equation}
    \label{eq:weight-3D}
    \begin{tikzpicture}[baseline]
    \matrix(a)[matrix of math nodes, nodes in empty cells, nodes={minimum size=25pt},
    row sep=2em, column sep=2em,
    text height=1.25ex, text depth=0.25ex]
    {&& X^{A}  & Q_X^{A} & \bar{Z}^{A}\\
    & S^{D}  & P^{D}  & \bar{S}^{D} &\\
    Z^{A} & Q_Z^{A} & \bar{X}^{A} &&\\};
    \path[-left to,red, font=\scriptsize,transform canvas={yshift=0.2ex}]
    (a-1-3) edge node[above]{$6$}  (a-1-4)
    (a-2-2) edge node[above]{$4$}  (a-2-3)
    (a-3-1) edge node[above]{$4$}  (a-3-2);
    \path[-left to,blue,font=\scriptsize,transform canvas={yshift=0.2ex}]
    (a-1-4) edge node[above]{$4$}  (a-1-5)
    (a-1-4) edge node[right]{$1$}  (a-2-4)
    (a-1-4) edge[bend left=80, dashed] node[right]{$1$} (a-3-3)
    (a-3-2) edge node[above]{$2$}  (a-3-3);
    \path[-left to,green!50!black,font=\scriptsize,transform canvas={yshift=0.2ex}]
    (a-2-3) edge node[above]{$2$}  (a-2-4)
    (a-2-3) edge node[right]{$1$}  (a-3-3);
    \path[-left to,red,font=\scriptsize,transform canvas={xshift=0.2ex}]
    (a-1-3) edge node[right]{$1$}  (a-2-3)
    (a-2-2) edge node[right]{$1$}  (a-3-2);
    \path[-left to,red,font=\scriptsize]
    (a-1-3) edge[bend right=80, dashed] node[right]{$2$} (a-3-2);

    \end{tikzpicture}
    \end{equation}
    where the number labels the maximum column weight of the linear map corresponding to the arrow.
    In particular, the maximum check weight of the 3D non-CSS Layer Code is 9, which is computed by the red arrows.

    Note that the total qubit degree is not quite the maximum column weight of the syndrome map.
    However, as shown in Theorem \ref{thm:symplectic-embedding}, the $X,Z$ sector qubits given by
    \begin{equation}
        Q_\alpha^C = Q^A_{\alpha} \oplus Q^D_{\alpha}, \quad \alpha=X,Z
    \end{equation}
    Hence, the total qubit degree should be upper bounded by the sum of all blue arrows, and double the sum of green arrows.
    In particular, the total qubit degree is $\le 8$.
\end{proof}

\section{Non-CSS Weight Reduction with Layer Codes}
\label{sec:weight}

In this section, we show that using the non-CSS embedding formalism in Theorem \ref{thm:symplectic-embedding}, the CSS weight reduction procedure \cite{yuan2026quantum} can be generalized nearly immediately to the non-CSS scenario (i.e., the input code can be non-CSS\footnote{Unlike the Euclidean Layer Codes, the (CSS and non-CSS) weight reduction procedure permits Non-qLDPC codes}).
Specifically, the CSS weight reduction procedure in \cite{yuan2026quantum} generalizes the CSS Layer Code construction by removing the Euclidean geometry constraint, while keeping topological defects (interactions between layers) sparse.
A similar intuition is applied to the non-CSS Layer codes in the previous section \ref{sec:layer} to obtain the corresponding non-CSS weight reduction procedure.
In particular, due to the parallels between CSS weight reduction \cite{yuan2026quantum} and CSS Layer Codes \cite{williamson2023layer}, we shall be relatively brief with the non-CSS weight reduction since the non-CSS Layer Codes were derived in detail in Section \ref{sec:layer}.

\subsection{Input Code}

Adopt the notation in Section \ref{sec:layer-input} so that the non-CSS input code $B$ is written as a symplectic complex
\begin{equation}
    B = S \xrightarrow{\sigma} P\xrightarrow{\hat{\sigma}}\bar{S}
\end{equation}
where checks $s\in \cS$ and qubits $q\in \cQ$ with $X,Z$-sectors $q_X,q_Z$.
Let $\weight,\qubit$ denote the maximum check weight and total qubit degree of $B$, respectively.

\begin{definition}[Induced Check Graph]
    \label{def:induced-check-graph}
    Let $\cG_S = (\cS,\cE_S)$ denote the graph with vertices $\cS$ and edges $s_1s_2\in \cE_{Q}$ if (1) $\supp s_1 \cap \supp s_2 \ne \varnothing$ or (2) there exists check $s$ which \textit{nontivially overlaps} (as in Definition \ref{def:measurement-graph-non-CSS}) with both $s_1,s_2$.
    Let $\chi_{S}$ be such that $\cG_S$ is $\chi_S$ vertex colorable so that 
    \begin{equation}
        \chi_S=O(\weight^2 \qubit^2)
    \end{equation}
    The coloring scheme induces a mapping $\eta_S:\cS\to [\chi_S]$.
    We may omit the subscript, when the meaning is clear from context. 
\end{definition}
Let $R_X,R_Z$ denote repetition codes on $\chi_S$ vertices, which we shall view as arranged in the $\hat{x},\hat{z}$ axes, respectively.

\begin{definition}[Induced Qubit Graph]
    \label{def:induced-qubit-graph}
    Let $\cG_Q = (\cQ,\cE_Q)$ denote the graph with vertices $\cQ$ and edges $q_1 q_2\in \cE_{Q}$ if there exists check $s\in \cS$ such that both $q_1,q_2 \in \supp s$ are within the support of $s$. 
    Let $\chi_{Q}$ be such that $\cG_Q$ is $\chi_Q$ vertex colorable so that 
    \begin{equation}
        \chi_Q=O(\weight \qubit)
    \end{equation}
    The coloring scheme induces a mapping $\eta_Q:\cQ\to [\chi_Q]$.
    We may omit the subscript, when the meaning is clear from context. 
\end{definition}

As described in Section \ref{sec:layer-input} for the non-CSS Layer Codes, additional space (which was unnecessary for CSS codes) is necessary to decongest the defect lines for general non-CSS codes. In particular, define
\begin{definition}[Induced Check-Check Graph]
    \label{def:induced-check-check-graph}
    Let $\cG_{SS} = (\cS,\cE_{SS})$ denote the graph with vertices $\cS$ and edges $s_1 s_2\in \cE_{Q}$ if $\supp s_1 \cap \supp s_2 \ne \varnothing$.
    Let $\chi_{SS}$ be such that $\cG_{SS}$ is $\chi_{SS}$ \textit{edge} colorable so that 
    \begin{equation}
        \chi_{SS}=O(\weight \qubit)
    \end{equation}
    The coloring scheme induces a mapping $\eta_{SS}:\cE_S\to [\chi_{SS}]$.
    We may omit the subscript, when the meaning is clear from context. Note that $\cG_{SS}$ is simply $\cG_{S}$ with edges given by requirement (1) in Definition \ref{def:induced-check-graph}.
\end{definition}
Let $R_Q$ denote the repetition code on $\tilde{\chi}_Q$ vertices where
\begin{equation}
    \tilde{\chi}_Q=\chi _Q+\chi_{SS} =O(\weight \qubit)
\end{equation}
With slight abuse of notation, let $\eta_{SS}$ map to the latter $\chi_{SS}$ colors in $[\tilde{\chi}_Q]$, i.e., $\in (\chi_Q,\tilde{\chi}_Q]$.



\subsection{Qubit Layers}
\label{sec:weight-qubit}
Similar to Section \ref{sec:layer-qubit}, the qubit layers are defined simply as
\begin{definition}[Qubit Layers]
    \label{def:weight-qubit-layer}
    Define the qubit layer as the 3-term cell complex
    \begin{equation}
        R_X\otimes q\otimes R_Z^{\top}
    \end{equation}
    where $q\in \cQ$.
    By Example \ref{ex:CSS-symplectic}, the qubit layer is a CSS code, which can be written as a symplectic complex $D(q)$.
\end{definition}

\subsection{Check Layers}

Similar to Section \ref{sec:layer-check}, we define the (\textit{raw}) check layers as follows.

\begin{definition}[$XZ$ Graph]
    Consider a plane with axes $\hat{x}\times \hat{z}$ and of linear dimension $\chi_S$, i.e., $[\chi_S]^2$.
    Define the \textbf{XZ graph} of check $s$ as the collection of rows and columns in the $\hat{x}\times \hat{z}$ plane passing through $(\eta(s),\eta(s))\in [\chi_S]^2$, i.e.,
    \begin{equation}
        \XZ(s) = (\eta(s) \times [n_S]) \cup ([n_S]\times \eta(s))
    \end{equation}
    where we shall treat $\XZ(s)$ as a graph complex with the corresponding edges and vertices.
    Note that $\XZ(s)$ can be viewed as the union of $R_X,R_Z$ joining at the $\eta(s)$ vertex.
\end{definition}

\begin{definition}[Raw Check Layer]
    Define the \textbf{(raw) check layer} of check $s\in \cS$ as the 3-term cell complex
    \begin{equation}
        A_{\mathrm{raw}}(s) =\XZ(s) \otimes R_Q
    \end{equation}
    and thus correspond to union of $X,Z$-check layers in CSS weight reduction \cite{yuan2026quantum}.
\end{definition}

\begin{definition}[Check Layer]
    \label{def:weight-check-layer}
    Define the \textbf{check layer} of check $s\in \cS$ as the union 
    \begin{equation}
        A(s) = A_{\mathrm{raw}}(s) + \bigoplus_{s'\sim s} A(s';s)
    \end{equation}
    where the summation of $s'$ is over all checks which nontrivially overlap with $s$ relative to the input code $B$, and 
    \begin{equation}
        A(s';s)= R_X\otimes \eta(ss')\otimes R_Z
    \end{equation}
    For notation, label the vertices in $A(s)$ as $|ijk;s\ket$ where $ijk$ are coordinates along the axes $\hat{x}\times \hat{q}\times \hat{z}$. 
    Similarly, we can label edges along the $\hat{x}$ direction as $|i^{\pm} jk;s\ket$ where $i^\pm =i\pm 1/2$ are half-integers and similarly, for $|ij^{\pm} k;s\ket$ and $|ijk^{\pm};s\ket$.
\end{definition}

\subsection{Attaching Check and Qubit Layers}
Similar to Section \ref{sec:layer-attach}, define
\begin{align}
    \label{eq:weight-ancilla}
    A &= \bigoplus_{s\in \cS} A(s) \\
    D &= \bigoplus_{q\in \cQ} D(q)
    \label{eq:weight-data}
\end{align}
Hence, it remains to show that the maps $g,h,p$ in Theorem \ref{thm:symplectic-embedding} can be defined in a manner that is sufficiently sparse so that the output code $C$ is weight and degree reduced.

\subsubsection{Gluing Maps}

We thus define the gluing map $g$ via
\begin{align}
    \label{eq:g-weight}
    |i,\eta(q),k;s\ket^A &\mapsto  |iqk\ket^D_X 1\{q_X \sim s, \eta(s)=i\} \\
    &\quad +|iqk\ket^D_Z 1\{q_Z \sim s,\eta(s)=k\}
\end{align}
(and zero otherwise), where $\sim$ is with respect to the input code $B$, e.g., $q_X\sim s$ if $q_X\in \sigma s$ and $q_X$ are the $X$- qubit in $\cQ_X$ corresponding to $q$ (the case is similar for $q_Z$).
Note that the map is well-defined since by Definition \ref{def:induced-qubit-graph}, the coloring map $\eta$ is injective on the support of any check $s$.
Further define the gluing map $h$ via
\begin{align}
    |i,\eta(q),k^+;s\ket^{A} &\mapsto |iqk^+\ket^{D}_X 1\{q_X\sim s,\eta(s)=i\} \\
    |i^+,\eta(q),k;s\ket^{A} &\mapsto |i^+q k\ket^{D}_Z 1\{q_Z\sim s,\eta(s)=k\}
    \label{eq:h-weight}
\end{align}
(and zero otherwise), where $i^{\pm} =i\pm 1/2$ are used to denote the edges along $\hat{x}$ and similarly for $k^\pm$.
It's then straightforward to check the following, whose proof is omitted since it's nearly the same as that of Lemma \ref{lem:chain-3D}.
\begin{lemma}[Chain Map]
    \label{lem:chain-weight}
    The gluing maps $(g,h)$ in Eq. \eqref{eq:g-weight}-\eqref{eq:h-weight} form a chain map.
\end{lemma}

\subsubsection{Defect Maps}
\label{sec:weight-defects}

The construction of the defects is similar to that in Section \ref{sec:defects-3D}, where we utilize the additional $\chi_{SS}$ planes to route the defect lines.
Here, we provide a sketch instead of repeating the details.
Let $s\rh{\wedge} s', s\lh{\wedge} s' \subseteq \cQ$ be as before so that by the commutation relations, $|s\rh{\wedge}s'|=|s\lh{\wedge}s'|$.
By Definition \ref{def:induced-qubit-graph}, the color mapping $\eta_Q$ is injective on $s\rh{\wedge}s'$ and on $s\lh{\wedge}s'$, and thus qubits can be ordered $\prec$ with respect to the values of $\eta_Q$.
In particular, if $|s\rh{\wedge}s'|=|s\lh{\wedge}s'|$ are both odd, then
\begin{align}
    s\rh{\wedge}s' &= (q_1 \prec\cdots \prec q_{2m}\prec q_0) \\
    s\lh{\wedge}s' &= (\bar{q}_{1} \prec \cdots \prec  \bar{q}_{2m}\prec \bar{q}_0)
\end{align}
We then construct the \textbf{defect paths},
\begin{align}
    \gamma_i(ss')&=[\eta(q_{2i-1}),\eta(q_{2i})], \quad i\in[m] \\
    \Gamma_i(ss')&=\eta(s)\times \gamma_i(ss')\times \eta(s')
\end{align}
and
\begin{align}
    \bar{\gamma}_i(ss')&=[\eta(\bar{q}_{2i-1}),\eta(\bar{q}_{2i})], \quad i\in [m] \\
    \bar{\Gamma}_{i}(ss') &= \eta(s')\times \bar{\gamma}_{i}(ss')\times \eta(s)
\end{align}
and similar to Eq. \eqref{eq:non-CSS-defect-3D}, define the special route $\Gamma_0(ss')$ as follows (where the Euclidean coordinates are replaced by the corresponding color maps)
\begin{equation}
    \eta(s')\eta(\bar{q}_0) \eta(s)\to \cdots \to  \eta(s)  \eta(q_0) \eta(s')
\end{equation}
Add a time parametrization to the defect paths similar to Section \ref{sec:defects-3D}, e.g., $[0,\tau_i]\mapsto\gamma_i(t;ss')$ routes $\eta(q_{2i-1})\to \eta(q_{2i})$.
The case where $|s\rh{\wedge}s'|=|s\lh{\wedge}s'|$ are both even is similarly derived.

Further define $I_{i}(\rh{s s}')$ to be the interval $[0,\tau_i)$ if $s\to s'$ is consistent with the order between pair $s,s'$ and $=(0,\tau_i]$ if inconsistent. Similarly, define $\sgn \rh{s s}'=+$ if $s\to s'$ is consistent, and $-$ if inconsistent. 
We can thus define the defect map $p$ as follows
\begin{align}
    \label{eq:p-weight}
    |ijk;s\ket^{A} &\mapsto \sum_{s'\sim s} \sum_{i}\sum_{t\in I_i(\rh{ss}')} 1\{ijk=\Gamma_i(t;ss')\} \\
    &\quad\quad \times |\Gamma_i(t^{\sgn \rh{s s}'};ss');s'\ket^A_Z
\end{align}
where the summation is over all $s'$ such that $s,s'$ have overlapping support (denoted as $s'\sim s$), and we treat $\bar{\Gamma}_i=\Gamma_{-i}$ so that the summation over $i$ is such that $|i|\le m$, and the notation emphasizes $|\cdots \ket_Z^A \in Q_Z^A$.

It's then straightforward to check that (whose proof we omit since it's nearly identical to that of Lemma \ref{lem:compatible-3D})
\begin{lemma}[Compatibility]
    The gluing maps $g,\hat{g}$ are compatible with the defect map $p$, i.e.,
    \begin{equation}
        \hat{g}g= \partial^A p+p^{\top}\delta^A
    \end{equation}
\end{lemma}

\subsection{Main Results}

\begin{theorem}[Non-CSS Weight Reduction]
    \label{thm:layer-weight}
    Let $A,D$ be defined as in Eq. \eqref{eq:weight-ancilla}-\eqref{eq:weight-data} (along with Definitions \ref{def:weight-check-layer} and \ref{def:weight-qubit-layer}).
    Let gluing maps $g,h$ be defined as in Eq. \eqref{eq:g-weight}-\eqref{eq:h-weight}, and defect map $p$ in Eq. \eqref{eq:p-weight}.
    Then the output $C$ constructed via Theorem \ref{thm:symplectic-embedding} is a well-defined symplectic complex with qubits
    \begin{equation}
        |\cQ^C| = O(\weight^4 \qubit^4) |\cQ|
    \end{equation}
    and
    \begin{equation}
        H_1(C) \cong H_1(B)
    \end{equation}
    And maximum weight 9 and total qubit degree 8, and
    \begin{equation}
        d(C) \ge \min(\chi_S,2\tilde{\chi}_Q) \frac{1}{\weight} d(B)
    \end{equation}
    where $d(C),d(B)$ are the code distance of the output $C$ and input $B$ codes, respectively.
\end{theorem}

\begin{remark}
    As discussed previously, it's typical to expect
    \begin{align}
        \chi_{S} &= \Theta(\weight^2 \qubit^2) \\
        \tilde{\chi}_{Q} &= \Theta(\weight\qubit) 
    \end{align}
    Hence, for large weights and degrees, we have $\tilde{\chi}_{Q}\ll \chi_{S}$ and thus the code distance is
    \begin{equation}
        d(C) \ge \Theta(\qubit) d(B)
    \end{equation}
    Note that this lower bound in code distance is not ideal due to $\tilde{\chi}_{Q}\ll \chi_{S}$.
    
    However, similar to the CSS weight reduction procedure \cite{yuan2026quantum}, this lower bound may be an artifact of our proof.
    Moreover, similar to Remark 1 of \cite{yuan2026quantum}, we can arbitrarily increase $\tilde{\chi}_Q$ so that $\tilde{\chi}_Q,\chi_S=\Theta(\weight^2 \qubit^2)$ without increasing the asymptotic scaling of the output code size.
    In this case, the code distance has a larger lower bound
    \begin{equation}
        d(C) \ge \Theta(\weight \qubit^2) d(B)
    \end{equation}
\end{remark}

\begin{proof}
    This proof is nearly identical to Theorem \ref{thm:layer-3D} and Theorem III.2-III.3 in SI of \cite{yuan2026quantum}, and thus omitted.
\end{proof}
\appendix

\setcounter{definition}{0}
\renewcommand{\thedefinition}{\Alph{section}\arabic{definition}}
\renewcommand{\theHdefinition}{appendix.\arabic{definition}}
\def\theHtheorem{\theHdefinition}
\def\theHlemma{\theHdefinition}
\def\theHclaim{\theHdefinition}
\def\theHproblem{\theHdefinition}
\def\theHproposition{\theHdefinition}
\def\theHfact{\theHdefinition}
\def\theHcorollary{\theHdefinition}
\def\theHremark{\theHdefinition}
\def\theHobservation{\theHdefinition}
\def\theHnotation{\theHdefinition}
\def\theHexample{\theHdefinition}
\def\theHexamples{\theHdefinition}
\def\theHquestion{\theHdefinition}
\renewcommand{\thetheorem}{\thedefinition}
\renewcommand{\thelemma}{\thedefinition}
\renewcommand{\theclaim}{\thedefinition}
\renewcommand{\theproblem}{\thedefinition}
\renewcommand{\theproposition}{\thedefinition}
\renewcommand{\thefact}{\thedefinition}
\renewcommand{\thecorollary}{\thedefinition}
\renewcommand{\theremark}{\thedefinition}
\renewcommand{\theobservation}{\thedefinition}
\renewcommand{\thenotation}{\thedefinition}
\renewcommand{\theexample}{\thedefinition}
\renewcommand{\theexamples}{\thedefinition}
\renewcommand{\thequestion}{\thedefinition}

\section{CSS Measurement}
\label{sec:CSS-measure}

In this section, we derive the fault complex \cite{hillmann2025single} for CSS measurement -- specifically, for syndrome extraction of CSS codes in Section \ref{sec:syndrome-fault-complex} and CSS logical measurement of CSS codes in Section \ref{sec:logical-fault-complex}. This approach echoes the spacetime code of \cite{williamson2026low}.

Note that to rewrite the error correction procedure as a fault complex, the measurement information, which exists at the level of Pauli groups, must be incorporated in the language of complexes.
This is elaborated in the following subsection.

\subsection{Preliminaries}
\label{sec:CSS-prelim}
It's well known that after abelianization, a CSS code can be treated as a length-2 complex.
However, technically, abelianization does not capture the information of phases.
For example, a CSS code with checks $XXX,ZZI$ and another with $-XXX,ZZI$ would be described by the same complex.
For most scenarios, this subtlety is not important, since properties such as the code distance or logical subspace dimension do not depend on the phase information.
However, when performing measurement, the phases are important, since the measurement outcome $(-1)^s$ can be regarded as the phase in front of the stabilizer checks.
Hence, this section provides some preliminaries on how such information can be captured after abelianization, i.e., in complexes.


\begin{definition}[CSS Code]
    Let $\cP$ denote the Pauli group with phases and $\cP_X,\cP_Z$ be $X,Z$-type Pauli subgroups.
    A CSS code is a commuting Pauli subgroup $-I\notin \cS \subseteq \cP$ such that $\cS$ is generated by subgroups $\cS_X,\cS_Z$ where $\cS_\alpha\subseteq \cP_{\alpha},\alpha=X,Z$.
\end{definition}

\begin{definition}[Abelianization]
    Let $\Sigma_\alpha \subseteq \F_2^{n}$ be subspaces and $\mu_{\bar{\alpha}}\in \Sigma_\alpha^*$ for $\alpha=X,Z$ and $\bar{\alpha}$ is the other choice. Then the quantum state $\psi$ is \textbf{stabilized} with \textbf{measurement} $\mu_X,\mu_Z$ if
    \begin{equation}
        \rho(s) \psi = (-1)^{\mu_{\bar{\alpha}} s} \psi, \quad s\in \Sigma_\alpha
    \end{equation}
    where $\rho(s)$ is the $X,Z$-type operator obtained from $s$, e.g., $ZZI$ from $110$.
    If $\mu_X=\mu_Z=0$, then the measurements are trivial and omitted.
\end{definition}

\begin{lemma}[Equivalence]
    \label{lem:equivalence}
    There exists a canonical bijective map from collection of subgroups $\cS_X\subseteq \cP_X$ with $-I\notin \cS_X$ to the tuple $(\Sigma_X,\mu_Z)$ where $\Sigma_X$ is a subspace of $\F_2^n$ and $\mu_Z \in \Sigma_X^*$ is in the dual space.
    The case is similar for $Z$-type Paulis.
\end{lemma}
\begin{proof}
    Given subgroup $\cS_X\subseteq \cP_X$ with $-I\notin \cS_X$, let $\Sigma_X \subseteq \F_2^n$ denote its abelianization.
    Note that given $s\in \Sigma_X$, there must exist a unique Pauli in $\cS_X$, which we denote by $\rho(s)$. Indeed, if there are more than one, then the product of the two must be a phase $\ne I$, and thus $-I\in \cS_X$ so that we reach a contradiction.
    Let $\mu_Z:\Sigma_X\to \F_2$ denote the phase of the unique Pauli $\rho(s) \in \cS_X$ with abelianization $s\in \Sigma_X$.
    Linearity follows from the fact that abelianization is a homomorphism.
    Hence, this defines a map.
    The converse is similarly derived and thus we obtain a bijective map.
\end{proof}
\begin{remark}
    The previous observation tells us that a CSS code is specified uniquely by its abelianization $\Sigma_\alpha$ and measurement maps $\mu_{\bar{\alpha}}$.
\end{remark}
\begin{corollary}
    Let $\cS$ be a CSS code with corresponding abelianizations $\Sigma_\alpha$ and linear maps $\mu_{\bar{\alpha}}$. Then state $\psi$ is stabilized by $\cS\Leftrightarrow \psi$ is stabilized by $\Sigma_\alpha$ with measurements $\mu_{\bar{\alpha}}$.
\end{corollary}

\begin{remark}[Orientation]
    \label{rem:CSS-orientation}
    Throughout this section (and Section \ref{sec:low-lvl-CSS-measurement}) we treat the CSS case on its own and orient the CSS complex as
    \begin{equation}
        A = X\xrightarrow{\partial} Q \xrightarrow{\partial} Z
    \end{equation}
    i.e.\ $\partial$ carries the $X$-checks and $\delta=\partial^{\top}$ carries the $Z$-checks.
    This is the transpose of the convention used in the main text and in the non-CSS sections \ref{sec:non-CSS-measure} and \ref{sec:low-lvl-non-CSS-measure}, where $\partial^A:Z^A\to Q_Z^A$ and $\delta^A:X^A\to Q_X^A$.
    The two settings are developed independently and never share a diagram, so no translation is required; the reader need only keep in mind that $\partial^A$ acts on $X^A$ here and on $Z^A$ there.
\end{remark}

\begin{remark}[Distances]
    \label{rem:CSS-distances}
    For a CSS code $D$ we write $d_X(D)$ ($d_Z(D)$) for the minimum weight of a nontrivial $X$- ($Z$-) type logical.
    For a fault complex $F=F_2\to F_1\to F_0$ -- which is not in general a spatial code -- we write $d_1(F)$ for its $1$-systolic distance, and, when $F$ is the fault complex for $X$- ($Z$-) type errors, we also write $d_X(F)\equiv d_1(F)$ ($d_Z(F)\equiv d_1(F)$) to record the error type.
    Accumulated (time-summed) quantities are written in fraktur, e.g.\ $\mathfrak{e}_X(Q^D)=\sum_t e_X(Q^D)|_t$, to distinguish them from the vectorized $\be_X(Q^D)=\sum_t e_X(Q^D)|_t\otimes|t\ket$.
\end{remark}

\begin{lemma}[Physical]
    \label{lem:canonical-iso}
    Let $A=X\to Q\to Z$ be a complex (co)differential $\delta =\partial^\top$. Then there exists a canonical isomorphism $[a_Z] \mapsto \mu_Z:Q/\ker \delta \to (\im \partial)^*$ defined by
    \begin{equation}
        \mu_{Z}(\cdots )=\bra a_Z|\cdots \ket
    \end{equation}
    where $a_Z$ is any representation of equivalence class $[a_Z]$, i.e., modulo ($Z$-type) logical.
\end{lemma}

\begin{remark}[Physical]
    \label{rem:phys-err}
    From a physical perspective, $a_Z$ can be regarded as a physical $Z$-Pauli acting on a state stabilized by $A$ so that the resulting state has measurement $\mu_Z=\bra a_Z|$.
    Hence, we say that $\psi$ is stabilized by $A$ with ($Z$-type) \textbf{physical (measurements)} $a_Z$ \textbf{(modulo logical)\footnote{We will often omit the ``mod logical" portion for simplicity}}, instead of measurements $\mu_Z \in (\im \partial)^*$.
    In particular, this formulation also includes the possibility that a (presumably small) unknown physical error occurs on the given state.
    The case is similar for $X$-type.
\end{remark}

\begin{remark}[Valid Syndrome]
    Note that $\delta:Q/\ker \delta \to \im \delta$ is an isomorphism.
    Since $X$ is the space of ($Z$-type) \textbf{syndrome}, we refer to $\im \delta$ as the space of ($Z$-type) \textbf{valid syndromes}.
    The case is similar for $X$-type valid syndromes.
\end{remark}
\begin{proof}
    Since the basis on $X$ defines a non-degenerate bilinear form, by Riesz representation, given any $\mu_{Z}\in (\im \partial)^*$, there exists a unique $\alpha \in X$ such that
    \begin{equation}
        \mu_Z (\partial x)=\bra \alpha |x\ket
    \end{equation}
    Note that if $x\in \ker \partial$, then $\mu_Z (\partial x)=0$ and thus $\alpha \in (\ker \partial)^\perp =\im \delta$.
    Conversely, given any $\partial x \in \im \partial$, we define the map $\partial x\mapsto \bra \alpha|x\ket$.
    This is well-defined since if $\partial x =\partial x'$, then $x+x'\in \ker \partial$ and thus $\bra \alpha|x\ket=\bra \alpha|x'\ket$.
    Since $\mu_Z \mapsto \alpha$ is linear, we obtain an isomorphism $(\im \partial)^*\to \im \delta$.
    Note that $\im \delta \cong Q/\ker \delta$ via the isomorphism $\delta: Q/\ker \delta \to \im \delta$. The statement then follows from combining the two isomorphism.
\end{proof}

\begin{example}[Initialization]
    \label{ex:initialization}
    Suppose that qubits $Q$ are initialized transversally in the $X$-basis, so that the state is stabilized by
    \begin{align}
        X(Q) \xrightarrow{\rm{id}} Q
    \end{align}
    where $X(Q)$ is a copy of $Q$, with unknown physical $e_{\bullet}$.
    Note that $e_X$ can always be chosen $=0$ since the $X$-logicals (and $X$-stabilizers) are arbitrary, and thus only the $Z$-type physical error requires to be specified.
\end{example}

Utilizing the canonical isomorphism in Lemma \ref{lem:canonical-iso}, we can formulate subsequent measurements on a given state as follows.

\begin{lemma}[Overcomplete Measurement]
    \label{lem:overcomplete-physicals}
    Let $\psi$ be stabilized by a CSS code
    \begin{equation}
        A=X^A \to Q^A \to Z^A
    \end{equation}
    with physical $a_{\bullet}$.
    Measure $X$-type Paulis based on $\partial^B:X^B \to Q^A$ with physical $b_Z$. Then
    \begin{equation}
        \label{eq:valid-measurement}
        a_Z+b_Z \in \ker \delta^A +\ker \delta^B
    \end{equation}
    where $\delta^\alpha = (\partial^\alpha)^{\top}$.
    Moreover, if $\alpha_Z\in \ker \delta^A$ and $\beta_Z \in \ker \delta^B$ are such that $c_Z\equiv a_Z +\alpha_Z =b_Z +\beta_Z$, then the output state is stabilized by
    \begin{itemize}
        \item $X$-Paulis in $\im \partial^A +\im \partial^B$ with physical $c_Z$.
        \item $Z$-Paulis in $\im \delta^A \cap \ker \delta^B$ with physical $a_X$
    \end{itemize}
\end{lemma}
\begin{remark}[Valid Measurement]
    Note that the measurement may be \textit{overcomplete} since the initial state $\psi$ is already stabilized by $A$, and thus measurements on $\im \partial^A \cap \im \partial^B$ must be consistent with both $a$ and $b$. This is equivalent to
    Eq. \eqref{eq:valid-measurement}, as thus referred as the \textbf{valid measurement} condition.
\end{remark}
\begin{remark}
    Note that when perform an $X$-Pauli measurement, the $X$-type physical error $a_X$ does not change; instead, only the $Z$-type stabilizers change (which may change the mod logical part) and thus we often omit it in writing.
\end{remark}

\begin{remark}[Transversal Measurement]
    \label{rem:transversal-measurement}
    Note that if $Q^B \subseteq Q^A$ denotes a subspace with basis $\cQ^B \subseteq \cQ^A$ and $Q^B$ which are then measured in the $X$-basis transversally, then the readout $(Q^B)^*$ can be identified with some physical $b_Z$ since $(Q^B)^* \cong Q^B$ gives the canonical isomorphism in Lemma \ref{lem:canonical-iso}.
\end{remark}

\subsection{Fault Complexes}
\label{sec:fault}
In this section, we derive the spacetime fault complexes for $d$ rounds of measurement, for both syndrome extraction (in Algorithm \ref{alg:syndrome-extract}) and logical measurement (in Algorithm \ref{alg:logical-measure}).
Note that the detailed algorithms are postponed, since the relevant information is the measurement obtained each round, along with the possibility that physical errors can accumulate (with small probability) between rounds of measurement.

In particular, to prove a threshold, we shall utilizing the following fact about complexes.
Note that the theorem applies to any complex, and not necessarily those that are interpreted as spatial CSS codes.
Hence, the statement also applies to spacetime fault complex, as we shall elaborate in the following subsections.

\begin{lemma}[General Threshold]
    \label{lem:general-threshold}
    Let
    \begin{equation}
        F=F_2\to F_1\to F_0
    \end{equation}
    denote an arbitrary sparse (qLDPC) complex with basis $\cF$ and differential $\partial^F$.
    Let errors $e\in F_1$ occur with probability $p$, i.e., $\Pr[e \supseteq f] \le p^{|f|}$ for any $f\subseteq \cF_1$.
    Given readout $\partial^F e$, let $\partial^F e\mapsto \hat{e}$ denote the min-weight correction.
    Then there exists threshold $p_c$ such that if $p<p_c$, then the probability that $e+\hat{e}$ is a nontrivial logical of $F$ is upper bounded by
    \begin{equation}
        \Pr[e+\hat{e}\notin \im \partial^F] \le O(1) \dim F_1 \left( \frac{p}{p_c}\right)^{d_1(F)/2}
    \end{equation}
    where $d_1(F)$ is the 1-systolic distance of $F$.
\end{lemma}
\begin{proof}
    Follows from Gottesman \cite{gottesman2013fault}.
\end{proof}
\subsubsection{Syndrome Extraction}
\label{sec:syndrome-fault-complex}

Since syndrome extraction is symmetric in $X$ and $Z$-type errors, we shall restrict our attention to $X$-type errors. The case of $Z$-type errors is similarly derived.
Given a quantum state stabilized by a data CSS code $D=X^D\to Q^D\to Z^D$, input the state into Algorithm \ref{alg:syndrome-extract} and use the output as the input for subsequent syndrome extractions of Algorithm \ref{alg:syndrome-extract}.
Perform the syndrome extraction $T+1$ times (where we assume that that last measurement is perfect without measurement errors), and assume that physical errors can accumulate on $Q^D$ before each extraction, so that we can regard physical errors $e_X(Q^D)|_t$ occuring at integer time steps $t\in [0,T]$ and measurement errors $e_{X}(Z^D)|_{t^+}$ occurring at half-integer time steps $t^+$ where $t\in [0,T)$, where the subscript $X$ is omitted henceforth since we will only deal with $X$-type errors.

In particular, after $t+1$ rounds of subsequent extraction, we have readout
\begin{equation}
    \partial^D \left(\sum_{s\le t} e(Q^D)|_s\right) +e(Z^D)|_{t^+}
\end{equation}
where the summation denotes the accumulated physical error, and we assumed that $e(Z^D)|_{0^-},e(Z^D)|_{T^+}=0$.
Equivalently, we can compute the difference between readouts $t^+$ and and $t^-$ to obtain readouts
\begin{equation}
    \label{eq:syndrome-readout}
    \partial^D e(Q^D)|_t+e(Z^D)|_{t^+}+e(Z^D)|_{t^-}, \quad t\in [0,T]
\end{equation}
where we assume that $e(Z^D)|_{0^-},e(Z^D)|_{T^+}=0$.
This corresponds to the diagram
\begin{equation}
    \begin{tikzpicture}[baseline]
    \matrix(a)[matrix of math nodes, nodes in empty cells, nodes={minimum size=25pt},
    row sep=1.5em, column sep=1.5em,
    text height=1.25ex, text depth=0.25ex]
    { &   & Z^DE \\
      & Q^DV & Z^DV\\};
    \path[->,font=\scriptsize]
    (a-2-2) edge node[above]{$\partial^D$} (a-2-3)
    (a-1-3) edge node[right]{$\partial^R$} (a-2-3);
    \end{tikzpicture}
\end{equation}
where $R=E\to V$ denote the repetition code on $T+1$ vertices $|t\ket,t\in [0,T]$ with differential $\partial^R$. It's then clear that the physical errors occupy space $QV$ and the measurement errors occupy $Z^DE$, where we have omitted the tensor product $\otimes$ for simplicity, so that the direct sum $Q^DV\oplus Z^DE$ indicates the \textbf{fault locations (data and syndrome qubits)}.
The readouts in Eq. \eqref{eq:syndrome-readout} occupy $Z^D V$ and are determined by the sum of the differentials. In particular, $Z^D V$ is the space of \textbf{fault syndromes}, differential $\partial^F = \partial^R+\partial^D$ is the \textbf{syndrome map} and $\im \partial^F$ denotes the space of \textbf{valid fault syndromes}.

To create a fault complex, we must also introduce \textbf{fault stabilizers} to indicate whether accumulated errors are trivial logicals. Writing
$\be_X(Q^D) = \sum_t e_X(Q^D)|_t \otimes | t \ket \in Q^D V$ the vectorized version of $e_X(Q^D)$, we have the following lemma.
\begin{lemma}[Fault Stabilizers]
    \label{lem:fault-stabilizers}
    If $\be_X(Q^D)\in Q^DV$, then the accumulated error
    \begin{equation}
        \mathfrak{e}_X(Q^D) = \sum_{t=0}^T e_X(Q^D)|_{t}
    \end{equation}
    is a trivial logical, i.e., $\in \im \partial^D \Leftrightarrow \be_X(Q^D)$ is in the image of the map given by the red arrows as follows.
    \begin{equation}
        \label{eq:fault-stabilizers}
        \begin{tikzpicture}[baseline]
        \matrix(a)[matrix of math nodes, nodes in empty cells, nodes={minimum size=25pt},
        row sep=1.5em, column sep=1.5em,
        text height=1.25ex, text depth=0.25ex]
        { & Q^D E & Z^D E \\
         X^D V & Q^D V & Z^D V\\};
        \path[->,font=\scriptsize]
        (a-2-2) edge node[above]{$\partial^D$} (a-2-3)
        (a-1-3) edge node[right]{$\partial^R$} (a-2-3);
        \path[->,red,font=\scriptsize]
        (a-2-1) edge node[above]{$\partial^D$} (a-2-2)
        (a-1-2) edge node[right]{$\partial^R$} (a-2-2);
        \end{tikzpicture}
    \end{equation}
\end{lemma}
\begin{remark}
    Physically, the map $Q^D E \to Q^D V$ denotes physical errors that come in pairs and thus do not contribute to the accumulated error since they cancel each other out. We refer to $Q^D E\oplus X^D V$ as the \textbf{fault stabilizers}.
\end{remark}
\begin{proof}
    This direction $\Leftarrow$ is trivial and thus we focus on the $\Rightarrow$ direction.
    Consider the collection $\mathfrak{T}=\mathfrak{T}(e(Q^D))$ of times $t$ such that $e(Q^D)|_{t}$ is not a stabilizer, i.e., $\notin \im \partial^D$.
    Suppose that $\mathfrak{T}$ is non-empty; otherwise $e(Q^D)|_{t} \in \im \partial^D$ for every $t$ and the claim is immediate.
    Let $\mathfrak{T}(e)=(\tau_0<\cdots <\tau_\ell)$ be ordered.
    Note that the accumulated error is a stabilizer $\mathfrak{e} \in \im \partial^D$ and thus
    \begin{equation}
        \sum_{i=0}^\ell e(Q^D)|_{\tau_i} \in \im \partial^D
    \end{equation}
    Let $s_0(Q^D)\in Q^D E$ be defined as
    \begin{equation}
        s_0(Q^D)|_{t^+} = e(Q^D)|_{\tau_0} 1\{\tau_0 \le t <\tau_1\}
    \end{equation}
    Then we see that
    \begin{align}
        (e(Q^D)+&\partial^R s_0(Q^D))|_{t} \\
        &=
        \begin{cases}
            e(Q^D)|_t & t\ne \tau_0,\tau_1 \\
            0 & t=\tau_0 \\
            e(Q^D)|_{\tau_0} +e(Q^D)|_{\tau_1} &t=\tau_1
        \end{cases}
    \end{align}
    In particular, the error at $\tau_0$ is \textit{pushed} to time-step $\tau_1$.
    Repeat this process to obtain $s(Q^D)\in Q^DE$ such that
    \begin{align}
        (e(Q^D)+&\partial^R s(Q^D))|_{t} \\
        &=
        \begin{cases}
            e(Q^D)|_t & t\notin \mathfrak{T}(e(Q^D)) \\
            0 & t\in \mathfrak{T}(e(Q^D))\backslash\tau_{\ell} \\
            \sum_{i=0}^{\ell} e(Q^D)|_{\tau_i} &t =\tau_{\ell}
        \end{cases}
    \end{align}
    It's then clear $e(Q^D)+\partial^R s(Q^D) \in \im \partial^D$ and thus the statement follows.
\end{proof}
Let us further complete the diagram in Eq. \eqref{eq:fault-stabilizers} to form a complex.
\begin{lemma}
    \label{lem:fault-complex}
    Consider the map indicated by the red arrow.
    \begin{equation}
        \begin{tikzpicture}[baseline]
        \matrix(a)[matrix of math nodes, nodes in empty cells, nodes={minimum size=25pt},
        row sep=1.5em, column sep=1.5em,
        text height=1.25ex, text depth=0.25ex]
        { & Q^D E & Z^D E \\
         X^D V & Q^D V & Z^D V\\};
        \path[->,font=\scriptsize]
        (a-2-2) edge node[above]{$\partial^D$} (a-2-3)
        (a-1-3) edge node[right]{$\partial^R$} (a-2-3)
        (a-2-1) edge node[above]{$\partial^D$} (a-2-2)
        (a-1-2) edge node[right]{$\partial^R$} (a-2-2);
        \path[->,red,font=\scriptsize]
        (a-1-2) edge (a-1-3);
        \end{tikzpicture}
    \end{equation}
    Then the fault stabilizers has trivial fault syndrome $\Leftrightarrow$ the map is $\partial^D$.
\end{lemma}
\begin{proof}
    The direction $\Leftarrow$ is trivial and thus we consider the direction $\Rightarrow$.
    Specifically, let the red arrow have linear map $\partial^{\mathrm{tbd}}$.
    Let $s(Q^D)\in Q^D E$ and consider its image $(\partial^{\mathrm{tbd}}+\partial^R)s(Q^D)$.
    If the fault stabilizer has trivial fault syndrome, then
    \begin{align}
        \partial^D \partial^R s(Q^D) &= \partial^R \partial^{\mathrm{tbd}} s(Q^D)\\
        \partial^R (\partial^D +\partial^{\mathrm{tbd}}) s(Q^D) &=0
    \end{align}
    where we utilized the fact that $\partial^D,\partial^R$ commute. Since $\ker \partial^R =0$, we see that $\partial^{\mathrm{tbd}}=\partial^D$.
\end{proof}

With these lemmas in place, the appropriate fault complex naturally follows.
\begin{definition} \label{def:fault-complex-syndrome}
    The \textbf{fault complex} $F=F_2\to F_1\to F_0$ for $X$-type errors in syndrome extraction after $T+1$ rounds of measurements for the code $D$ is then given by
\begin{equation}
    \label{eq:fault-complex}
    \begin{tikzpicture}[baseline]
    \matrix(a)[matrix of math nodes, nodes in empty cells, nodes={minimum size=25pt},
    row sep=1.5em, column sep=1.5em,
    text height=1.25ex, text depth=0.25ex]
    { & Q^D E & Z^D E \\
     X^D V & Q^D V & Z^D V\\};
    \path[->,font=\scriptsize]
    (a-2-2) edge node[above]{$\partial^D$} (a-2-3)
    (a-1-3) edge node[right]{$\partial^R$} (a-2-3)
    (a-2-1) edge node[above]{$\partial^D$} (a-2-2)
    (a-1-2) edge node[right]{$\partial^R$} (a-2-2)
    (a-1-2) edge node[above]{$\partial^D$} (a-1-3);
    \end{tikzpicture}
\end{equation}
\end{definition}

One may notice the that Definition \ref{def:fault-complex-syndrome} corresponds to the tensor product $D\otimes R$, where $X^D E$ doesn't play a role.
The fault complex then provides a natural threshold.

\begin{remark}
    On the correspondence with the construction of \cite{williamson2026low}; the components $Q^DV,Z^D E$ correspond to respectively space and time faults, $Z^D V$ corresponds to detectors, while $X^DV, Q^DE$ map to respectively space and time stabilizers.
\end{remark}


Much like the distance of a code, the distance of the fault complex then allows us to guarantee resilience against both data and measurement errors.

\begin{lemma}[Recovery] \label{lem:recovery-lemma-syndrome}
    Let $F$ the fault complex for a code $D$, and $e = e_X(Q^D) \oplus e_X(Z^D) \in F_1$ a set of faults satisfying
    \begin{equation}
        |e| < \frac{d_X(F)}{2}
    \end{equation}
    Then, there exists $\hat{e} \in F_1$ such that
    \begin{equation} \label{eq:recovery-syndrome}
        \sum_t e_X(Q^D)|_t+ \hat{e}_X(Q^D)|_t \in \im \partial ^D
    \end{equation}
\end{lemma}
\begin{proof}
    Pick $\hat{e}$ the minimum weight correction with $\partial^F e = \partial^F \hat{e}$. From the usual code theoretic argument, we have $e + \hat{e} \in \im \partial ^F$; and Lemma \ref{lem:fault-stabilizers} gives Eq. \eqref{eq:recovery-syndrome}.
\end{proof}

The function of the fault complex, relative to the usual code complex, is to find the correction $\hat{\mathfrak{e}}_X(Q^D)=\sum_t \hat{e}_X(Q^D)|_t$; which is typically unattainable for a simple code complex in the presence of measurement errors. While the recovery operation is entirely contained on the data qubits, the fault complex finds $\hat{\mathfrak{e}}_X(Q^D)$ by extending its phase space to include syndrome qubits. One can also remark that the above lemma is optimal, in the sense that the tight characterization given by Lemma \ref{lem:fault-stabilizers} guarantees the existence of an uncorrectable error of size $d_X(F)/2$.

Beyond the adversarial setting, we also obtain a threshold theorem in the stochastic setting when the measurements are repeated a number $T \geq d_X(D)$ times.

\begin{theorem}[Threshold]
    \label{thm:syndrome-threshold}
    Suppose that $T\ge d_X(D)$ where $d_X(D)$ is the $X$-type code distance of the qLDPC CSS data code.
    Let (physical and measurement) $X$-type errors
    \begin{equation}
        \be_X(F_1)=\be_X (Q^D)\oplus \be_X(Z^D)
    \end{equation}
    occur with rate $p$.
    Suppose that readouts $\partial^F \be_X(F_1)$ and min-weight correction $\partial^F \be_X(F_1 )\mapsto \hat{\be}_X(F_1)$ relative to the fault complex $F$ are provided.
    Then there exists threshold $p_c$ such that if $p<p_c$, then the accumulated physical error of $\be_X(Q^D)+\hat{\be}_X(Q^D)$ is a nontrivial logical with probability
    \begin{equation}
        \le O(1) |\cQ^D| d_X(D) \left(\frac{p}{p_c}\right)^{d_X(D)/2}
    \end{equation}
\end{theorem}
\begin{proof}
    Since $T\ge d_X(D)$, we see that the fault complex has 1-systolic distance $\ge d_X(D)$, which follows from the tensor product $D\otimes R$.
    By Lemma \ref{lem:general-threshold}, it's then sufficient to show that the accumulated physical error of $\be_X(Q^D)+\hat{\be}_X(Q^D)$ is nontrivial $\notin \im \partial^D$ only if $\be_X(F_1)+\hat{\be}_X(F_1) \notin \im \partial^F$.
    Indeed, by Lemma \ref{lem:fault-stabilizers}, we see that $\be_X(Q^D)+\hat{\be}_X(Q^D) \notin \im \pi \partial^F$ where $\pi$ denotes the projection onto $Q^D V$.
    Hence, $\be_X(F_1) +\hat{\be}_X(F_1)\notin \im \partial^F$ and thus the statement follows.
\end{proof}

\subsubsection{Logical Measurement}
\label{sec:logical-fault-complex}

In this subsection, we similarly derive the fault complex of logical measurement in Algorithm \ref{alg:logical-measure}.
However, there are a couple of subtleties that must be clarified to compare with that of syndrome extraction.
For example, Algorithm \ref{alg:logical-measure} has three readouts instead of two in Algorithm \ref{alg:syndrome-extract}, so that the $X$- and $Z$-type readouts are asymmetric and thus must be treated differently.
Also note that any physical error on ancilla qubits $Q^A$ are discarded after each round of measurement, and thus the physical meaning of \textit{accumulated} error on $Q^A$ must be clarified.

\paragraph{$X$-Pauli errors}

We start by considering the $X$-type errors in logical measurement in Algorithm \ref{alg:logical-measure}.
Suppose that we are given a quantum state stabilized by a data CSS code $D$.
Input the state with possible unknown errors $e(Q)|_0$ into Algorithm \ref{alg:logical-measure}, and write $Q = Q^A \oplus Q^D, Z = Z^A \oplus Z^D$, so that we obtain relevant readouts
\begin{align}
    \partial e_X(Q)|_0+e_X(Z)|_{0^+} &\in Z \\
    \label{eq:ancilla-readout-0}
    \partial^A x^A|_{0} + e_X(Q^A)|_{0} &\in Q^A
\end{align}
where the output state has $X$-type physical $e_X(Q^D)|_0+gx^A|_{0}$.
Use the output as the input for another round of logical measurement to obtain relevant readouts
\begin{align}
    \partial (e_X(Q^D)|_0 +gx^A|_0 + e_X(Q)|_1)+e(Z)|_{1^+} &\in Z \\
    \partial^A x^A|_{1} + e_X(Q^A)|_{1} &\in Q^A
\end{align}
Let $r_0^A:=\partial^A x^A|_0+e_X(Q^A)|_0$ denote the measured readout in Eq. \eqref{eq:ancilla-readout-0}.
Then the qubit contribution to the syndrome of the next round is explicitly
\begin{align}
    &\partial
    \begin{pmatrix}
        e_X(Q^A)|_1\\
        e_X(Q^D)|_0+gx^A|_0+e_X(Q^D)|_1
    \end{pmatrix}\notag\\
    &\quad=
    \partial
    \begin{pmatrix}
        \partial^A x^A|_0+e_X(Q^A)|_1\\
        e_X(Q^D)|_0+e_X(Q^D)|_1
    \end{pmatrix}\notag\\
    &\quad=
    \partial\left[
    \begin{pmatrix}
        e_X(Q^A)|_0+e_X(Q^A)|_1\\
        e_X(Q^D)|_0+e_X(Q^D)|_1
    \end{pmatrix}
    +
    \begin{pmatrix}
        r_0^A\\0
    \end{pmatrix}
    \right]\notag\\
    &\quad=\partial(e_X(Q)|_0+e_X(Q)|_1)
    +\partial
    \begin{pmatrix}
        r_0^A\\0
    \end{pmatrix},
\end{align}
Hence, by utilizing the readout in Eq. \eqref{eq:ancilla-readout-0} when discarding the ancilla qubits, we can regard the readout as that error caused by the accumulated error $e_X(Q)|_0+e_X(Q)|_1$ even though errors $e_X(Q^A)|_0$ on the ancilla were discarded.
Specifically while the observed syndrome $s$ is
\begin{align}
     s_1 = \partial(e_X(Q)|_0+e_X(Q)|_1)
    +
    \begin{pmatrix}
        \partial r_0^A\\0
    \end{pmatrix} &\in Z \\
    r_1^A &\in Q^A
\end{align}
the quantity $r_0^A$ is known from the destructive readout of $Q^A$ in the previous round, thus the \emph{effective} syndrome information instead is:
\begin{align}
    s^{\text{eff}}_1 = s_1 + \partial r_0^A = \partial(e_X(Q)|_0+e_X(Q)|_1)
     &\in Z \\
    r_1^A &\in Q^A
\end{align}
Repeat this procedure $T+1$ times (where the last logical measurement is assumed to be perfect without measurement errors).
Then by computing the difference between readouts $t^+$ and $t^-$, we obtain information
\begin{align}
    \label{eq:logical-readout}
    \partial e_X(Q)|_t+e_X(Z)|_{t^+} +e(Z)|_{t^-}, \quad t\in [0,T]  \\
    \label{eq:ancilla-readout}
    \partial^A x^A|_t + e_X(Q^A)|_t\in Q^A, \quad t\in [0,T]
\end{align}
where $e_X(Z)|_{0^-},e_X(Z)|_{T^+}=0$ by assumption, and the output state has accumulated physical errors
\begin{equation}
    \mathfrak{e}_{\mathrm{phys}}  = \sum_{t\le T} e_X(Q^D)|_t +g\sum_{t\le T} x^A|_t
\end{equation}

Using Eq. \eqref{eq:logical-readout}, we can then obtain a fault complex $F$ similar to Eq. \eqref{eq:fault-complex}.

\begin{definition}
    The \textbf{fault complex} $F=F_2\to F_1\to F_0$ for $X$-type errors for measuring $g \ker \partial ^A$ on $D$ using $A$, after $T+1$ rounds of measurements, is given by $\cone(g) \otimes R$, where $\cone(g)$ is the (height-1) cone corresponding to the following complex:
    \begin{equation}
    \begin{tikzpicture}[baseline]
    \matrix(a)[matrix of math nodes, nodes in empty cells, nodes={minimum size=25pt},
    row sep=1.5em, column sep=1.5em,
    text height=1.25ex, text depth=0.25ex]
    { & X^{A} & Q^{A} & Z^{A} \\
     X^{D} & Q^{D} & Z^{D} &\\};
    \path[->,font=\scriptsize]
    (a-1-2) edge (a-1-3)
    (a-1-3) edge (a-1-4)
    (a-2-1) edge (a-2-2)
    (a-2-2) edge (a-2-3)
    (a-1-2) edge node[right]{$g$} (a-2-2)
    (a-1-3) edge node[right]{$g$} (a-2-3);
    \end{tikzpicture}
\end{equation}
\end{definition}

\begin{lemma}[Recovery]
    \label{lem:recovery-logical-msmt}
    Let $F$ the fault complex for a pair $D, A$, and $\be_X(F_1) = \be_X(Q) \oplus \be_X(Z) \in F_1$ a set of faults satisfying
    \begin{equation}
        |\be_X(F_1)| < \frac{d_X(F)}{2}
    \end{equation}
    Then, there exists $\hat{\be}_X(F_1) \in F_1$ such that
    \begin{equation} \label{eq:recovery-logical-X}
        \sum_{t\le T} e_X(Q)|_t+ \hat{e}_X(Q)|_t \in \im \partial^D + g \ker \partial^A
    \end{equation}
\end{lemma}
\begin{proof}
    We write
    \begin{align}
        \mathfrak{e}_X(Q^D) = \sum_{t\le T} e_X(Q^D)|_t, \quad \mathfrak{e}_X(Q^A) &= \sum_{t\le T} e_X(Q^A)|_t, \quad  \mathfrak{x}^A = \sum_{t\le T} x^A|_t, \quad \mathfrak{r}^A = \sum_{t \leq T} r_t^A\\ \mathfrak{e}_{\mathrm{phys}} &= \mathfrak{e}_X(Q^D) + g \mathfrak{x}^A
    \end{align}
    where the fraktur letters denote accumulated (time-summed) quantities as in Remark \ref{rem:CSS-distances}.
    From Lemma \ref{lem:recovery-lemma-syndrome}, the fault complex gives a correction $\hat{\mathfrak{e}}_X(Q) = \hat{\mathfrak{e}}_X(Q^D) \oplus \hat{\mathfrak{e}}_X(Q^A) \in Q^D \oplus Q^A$, such that:
    \begin{align}
        \mathfrak{e}_X(Q^D) + \hat{\mathfrak{e}}_X(Q^D) &= \partial^D s^D + g s^A \\
        \mathfrak{e}_X(Q^A) + \hat{\mathfrak{e}}_X(Q^A) &= \partial^A s^A
    \end{align}
    for some $s^A \in X^A, s^D \in X^D$.
    As the fault complex is oblivious to $x^A$, the residual error remains non trivial:
    \begin{equation}
        \mathfrak{e}_{\mathrm{phys}} + \hat{\mathfrak{e}}_X(Q^D) = \partial^D s^D + g s^A + g \mathfrak{x}^A
    \end{equation}
    Note that much like in \cite{williamson2026low}, the residual error on the data subsystem is the image of an ancilla stabilizer $s^A + \mathfrak{x}^A$ through the map $g$. This image is then corrected by finding a pre-image, or correction \footnote{In \cite{williamson2026low} the correction obtained is exactly $s^A + \mathfrak{x}^A$, because the graph ancilla has $H_1(D) \cong Z_1(D)$. In our case we recover it up to some element that acts trivially on the output state.}, for $\partial^A (s^A + \mathfrak{x}^A)$. For this we leverage the correction $\hat{\mathfrak{e}}_X(Q^A)$ we obtained for the ancillary system, applied to $\mathfrak{r}^A$:
    \begin{equation}
        \mathfrak{r}^A + \hat{\mathfrak{e}}_X(Q^A)= \partial^A \mathfrak{x}^A + \mathfrak{e}_X(Q^A) + \hat{\mathfrak{e}}_X(Q^A) = \partial^A \mathfrak{x}^A + \partial^A s^A = \partial^A \hat{s}^A
    \end{equation}
    for some choice of $\hat{s}^A$.
    We then have
    \begin{equation}
        \mathfrak{e}_{\mathrm{phys}} + \hat{\mathfrak{e}}_X(Q^D) + g\hat{s}^A = \partial^D s^D + g (s^A + \mathfrak{x}^A + \hat{s}^A) \in \im \partial^D + g \ker \partial^A
    \end{equation}
    We see that the accumulated data error belongs to $\im \partial^D + g\ker \partial^A$, which is simply an $X$-type stabilizer, plus some $g\ker \partial^A$ term. That term corresponds to the subspace of logicals that is measured, which stabilizes the post-measurement state, and thus the remaining error acts trivially on that state.
\end{proof}

We can then obtain a threshold similar to Theorem \ref{thm:syndrome-threshold}.
\begin{theorem}[$X$-Type Logical Threshold]
    \label{thm:logical-threshold-X}
    Suppose that $T\ge d_X(C)$ where $d_X(C)$ is the $X$-type code distance of the (height-1) cone $C$. Let $X$-type fault errors
    \begin{equation}
        \be_X(F_1) =\be_X(Q) \oplus \be_X(Z)
    \end{equation}
    occur with rate $p$.
    Suppose that readouts $\partial^F \be_X(F_1)$ and min-weight correction $\partial^F \be_X(F_1)\mapsto \hat{\be}_X(F_1)$ relative to the fault complex $F$ are provided.
    Then there exists threshold $p_c$ such that if $p<p_c$, then the accumulated physical error of $\be_X(Q)+\hat{\be}_X(Q)$ is a nontrivial logical with probability
    \begin{equation}
        \le O(1) |\cQ| d_X(C) \left(\frac{p}{p_c}\right)^{d_X(C)/2}
    \end{equation}
    Note that the ancilla gadget $A$ was chosen so that the code $C$ has distance $d_X(C) = \Omega(1) d_X(D)$.
\end{theorem}


\paragraph{$Z$-Pauli errors}

By Theorem \ref{thm:logical-measure} or Algorithm \ref{alg:logical-measure},
\begin{equation}
    \tilde{e}_Z(Q) = \tilde{e}_Z(Q^D) \oplus \tilde{e}_Z(Q^A)
\end{equation}
where $\tilde{e}_Z(Q^D) =e_Z(Q^D) +\ell_Z(Q^D)$ for some $Z$ logical $\ell_Z(Q^D)$ of $D$ (which contains information of the logical measurement of the targeted $\ell_X^{\star}$), so that $e_Z(Q^D)$ can be regarded as small physical errors appearing at integer time steps and $\ell_Z(Q^D)$ is the measurement information of the targeted $X$-type logical $\ell^{\star}_X$.
However, $\tilde{e}_Z(Q^A)$ is unknown and thus its weight cannot be well-controlled (which is a result of discarding the ancilla at the end of each round of measurement).
Note that the $X$-type errors on the ancilla were still relevant since the ancilla qubits were measured out in the $Z$-basis and thus can be determined from the readout.

Therefore, the only relevant information of the readout $\delta \tilde{e}_Z(Q) +e_Z(X)$ for $Z$-type errors is given by that $\in X^D$, i.e.,
\begin{equation} \label{eq:z-err-syndrome}
    \delta^D e_Z(Q^D) + e_Z(X^D) \in X^D
\end{equation}
where both $e_Z(Q^D),e_Z(X^D)$ are small errors, and we used the fact that $\delta^D \tilde{e}_Z(Q^D) = \delta^D e_Z(Q^D)$.
Hence, the fault complex for $Z$-type errors can be constructed similar to that in syndrome extraction and given by $D \otimes R^{\top}$.

\begin{definition}
    The \textbf{fault complex} $F=F_2\to F_1\to F_0$ for $Z$-type errors for measuring $g \ker \partial ^A$ on $D$ using $A$, after $T+1$ rounds of measurements, is given by the construction of Definition \ref{def:fault-complex-syndrome} applied to the following (height-1) cone:
    \begin{equation}
    \begin{tikzpicture}[baseline]
    \matrix(a)[matrix of math nodes, nodes in empty cells, nodes={minimum size=25pt},
    row sep=1.5em, column sep=1.5em,
    text height=1.25ex, text depth=0.25ex]
    { & X^{A} & Q^{A} & Z^{A} \\
     X^{D} & Q^{D} & Z^{D} &\\};
    \path[->,font=\scriptsize]
    (a-1-2) edge (a-1-3)
    (a-1-3) edge (a-1-4)
    (a-2-1) edge (a-2-2)
    (a-2-2) edge (a-2-3)
    (a-1-2) edge node[right]{$g$} (a-2-2)
    (a-1-3) edge node[right]{$g$} (a-2-3);
    \end{tikzpicture}
\end{equation}
\end{definition}

\begin{lemma}[Recovery]
    Let $F$ the fault complex for a pair $D, A$, and $e = e_Z(Q^D) \oplus e_Z(X^D) \in F_1$ a set of faults satisfying
    \begin{equation}
        |e| < \frac{d_Z(F)}{2}
    \end{equation}
    Then, there exists $\hat{e} \in F_1$ such that
    \begin{equation} \label{eq:recovery-logical-Z}
        \mathfrak{e}_Z(Q^D)+ \hat{\mathfrak{e}}_Z(Q^D) \in \im \delta^D , \quad \bra g^{\top} \ell_Z(Q^D)|\cX^A\ket = \bra m^A + g^{\top} \hat{\mathfrak{e}}_Z(Q^D) | \cX^A\ket
    \end{equation}
    Where $m^A$ is the readout of the last $X^A$ measurements, $\cX^A\in X^A$ is the sum over all $X$-checks (vertices) of the ancilla, $\mathfrak{e}_Z(Q^D) = \sum_{t\le T} e_Z(Q^D)|_t$, and $\hat{\mathfrak{e}}_Z(Q^D) = \sum_{t\le T} \hat{e}_Z(Q^D)|_t $.
\end{lemma}
\begin{proof}
    We have $\mathfrak{e}_Z(Q^D)+ \hat{\mathfrak{e}}_Z(Q^D) \in \im \delta^D$ directly from Lemma \ref{lem:recovery-lemma-syndrome}. Then, by Algorithm \ref{alg:logical-measure}, the $T+1$-th measurement (assumed to be void of measurement errors) has readout
    \begin{equation}
        \label{eq:logical-readout-remark}
        m^A  =g^{\top} \left( \mathfrak{e}_Z(Q^D) +\ell_Z(Q^D) \right)+ \alpha, \quad \alpha \in \im \delta^A
    \end{equation}
    with $\alpha$ encoding the effect of $e_Z(Q^A)$ on the syndrome.
    Now note that
    \begin{equation}
        \bra m^A| \cX^{A}\ket = \bra g^{\top} \left( \mathfrak{e}_Z(Q^D) +\ell_Z(Q^D) \right)+ \alpha| \cX^{A}\ket = \bra g^{\top} \left( \mathfrak{e}_Z(Q^D) +\ell_Z(Q^D) \right)| \cX^{A}\ket
    \end{equation}
    Since $\cX^{A}\in \ker \partial^A$, we see that the $\F_2$ measurement is invariant under a change of $\alpha$ within $\im \delta^A$; and therefore no matter how large the error $e_Z(Q^A)$ on the ancilla qubits is, it does not materially affect the measurement outcome.
    Therefore, using the corrections $\hat{e}_Z(Q^D)|_{t}$, we can obtain the desired logical information:
    \begin{equation}
        \bra g^{\top} \ell_Z(Q^D)|\cX^A\ket = \bra g^{\top} \left( \mathfrak{e}_Z(Q^D) + \hat{\mathfrak{e}}_Z(Q^D) +\ell_Z(Q^D) \right)| \cX^{A}\ket  = \bra  m^A +g^{\top} \hat{\mathfrak{e}}_Z(Q^D)  | \cX^{A}\ket
    \end{equation}
\end{proof}

The threshold then follows from Lemma \ref{lem:general-threshold}, i.e.,

\begin{theorem}[$Z$-Type Logical Threshold]
    \label{thm:logical-threshold-Z}
    Suppose that $T\ge d_Z(D)$ where $d_Z(D)$ is the $Z$-type code distance of the data code $D$. Let $Z$-type fault errors
    \begin{equation}
        \be_Z(F_1) =\be_Z(Q^D) \oplus \be_Z(X^D)
    \end{equation}
    occur with probability $p$.
    Suppose that readouts $\delta^F \be_Z(F_1)$ and min-weight correction $\delta^F \be_Z(F_1)\mapsto \hat{\be}_Z(F_1)$ relative to the fault complex $F=D\otimes R^{\top}$ are provided.
    Then there exists threshold $p_c$ such that if $p<p_c$, then the accumulated physical error of $\be_Z(Q^D)+\hat{\be}_Z(Q^D)$ is nontrivial with probability
    \begin{equation}
        \le O(1) |\cQ^D| d_Z(D) \left(\frac{p}{p_c}\right)^{d_Z(D)/2}
    \end{equation}
\end{theorem}

\begin{remark}
    We note that the fault-complexes we obtain are based on a circuit model, and are distinct from those derived for measurement-based computation \cite{hillmann2025single}. In particular we obtain two different complexes, one for each of $X$ and $Z$ type errors.
\end{remark}
\section{Non-CSS Measurement}
\label{sec:non-CSS-measure}

Similar to Section \ref{sec:CSS-measure}, in this section, we derive the fault complex \cite{hillmann2025single} from first principles for non-CSS measurement using the language of symplectic complexes.

\subsection{Preliminaries}
In this section, we generalize the formalism of Section \ref{sec:CSS-prelim} to non-CSS codes.

\begin{definition}[Non-CSS Code]
    Let $\cP$ denote the Pauli group with phases.
    A \textbf{non-CSS code} is a commuting Pauli subgroup $-I\notin \cS \subseteq \cP$.
    Note despite the name, CSS codes are non-CSS codes.
    Let $\cS^{\perp}$ denote the collection of elements in $\cP$ which commute with $\cS$, i.e., the \textbf{logicals} of $\cS$, so that $\cS^{\perp}\backslash \cS$ are the \textbf{nontrivial} logicals.
\end{definition}

\begin{definition}[Abelianization]
    Let $P=\F_2^{2n}$ be the \textbf{Pauli space} on $n$ qubits, where $P$ is equipped with the standard \textbf{symplectic form} $\Lambda:P\times P\to \F_2$,
    \begin{equation}
        \Lambda(a,b)=a^{\top}\lambda b, \quad \lambda =
        \begin{pmatrix}
            0 & I\\
            I & 0
        \end{pmatrix}
    \end{equation}
    We say that $a,b\in P$ are \textbf{commuting} if $\Lambda(a,b)=0$.
    Let $\dutchcal{a} \mapsto a: \cP \to P$ be the usual \textbf{abelianization} (homomorphism) map so that $X,Y,Z$ Pauli on the $i$-th qubit maps to $10,11,01$ on bit positions $i,i+n$, respectively.
\end{definition}

\begin{remark}
    Note that the image of a non-CSS code $\cS$ is a commuting subspace $\Sigma$. In fact, abelianization $\cS \to \Sigma$ defines an isomorphism when restricted to $\cS$ and its image.
    The question is then given a commuting subspace, how can all the non-CSS codes that abelianize to the same subspace be specified.
    This will correspond to the phase information as follows, similar to Lemma \ref{lem:equivalence} for CSS codes, though for general non-CSS, the equivalence will be non-canonical, i.e., depends on a fixed stabilizer group.
    Note that this nature does not affect measurement, since we always start with a fixed stabilizer group, and obtain the corresponding phase information.
\end{remark}

\begin{lemma}[Equivalence]
    \label{lem:equivalence-non-CSS}
    Let $\cS$ be a fixed stabilizer group with abelianized commuting subspace $\Sigma\subseteq P$.
    Then there exists a one-to-one correspondence between stabilizer groups $\cS_\mu$ with abelianization $=\Sigma$ and $\F_2$ linear maps $\mu\in \Sigma^*$.
    In particular, we say that $\cS_{\mu}$ is \textbf{equivalent} to $\cS$ up to phase information $\mu\in \Sigma^*$.
\end{lemma}
\begin{proof}
    Note that abelianization $\cS \to \Sigma$ is an isomorphism when restricted to $\cS$ (i.e., injective) since $-I\notin \cS$. Let $\rho:\Sigma \to \cS\subseteq \cP$ denote the inverse map so that if $\dutchcal{s}\in \cS$ has abelianization $s$, then $\dutchcal{s}=\rho(s)$.
    Also note that $\rho(s)^2=\rho(0)=I$ and thus $\rho(s)$ is Hermitian.

    Now let $\mu:\Sigma \to \F_2$ be a linear map, and define $\cS_{\mu}$ as the collection of Pauli group elements of the form
    \begin{equation}
        (-1)^{\mu s} \rho(s), \quad s\in \Sigma
    \end{equation}
    Then it's clear that $\cS_{\mu}$ is a subgroup which does not contain $-I$ and that has abelianized subspace $=\Sigma$.

    Conversely, let $\cS_{\mu}$ be a stabilizer group with abelianized subspace $=\Sigma$.
    By abelianization, we see that there exists map $\mu:\Sigma\to \F_2$ such that
    \begin{equation}
        \dutchcal{s} = (-1)^{\mu s} \rho(s)
    \end{equation}
    where $\dutchcal{s}\mapsto s$ after abelianization. Indeed, since $\dutchcal{s}$ and $\rho(s)$ have the same abelianization, they must be equal up to a phase. If that phase is $=\pm i$, then $(\dutchcal{s} \rho(s))^2= \dutchcal{s}^2=-I$ where we utilized the fact that $\rho(s)^2=I$ and thus we reach a contradiction since $-I\notin \cS_{\mu}$.

    Now it remains to show that $\mu$ is linear. Specifically, note that
    \begin{equation}
        \dutchcal{s}_i = (-1)^{\mu s_i}\rho(s_i), \quad \dutchcal{s}_i\in \cS_\mu
    \end{equation}
    Then
    \begin{align}
        \dutchcal{s}_1 \dutchcal{s}_2 =(-1)^{\mu s_1 +\mu s_2} \rho(s_1+s_2)
    \end{align}
    Note that $s=s_1+s_2$ is the abelianization of $\dutchcal{s}_1 \dutchcal{s}_2$ and thus $\mu s=\mu s_1+\mu s_2$ since $\rho(s)^2=I$.
    The statement then follows.
\end{proof}

\begin{corollary}
    Let $\cS$ be a non-CSS code and $\cS_{\mu}$ be an equivalent up to phase $\mu\in \Sigma^*$. Then state $\psi$ is stabilized by $\cS_{\mu}\Leftrightarrow \psi$ is stabilized by $\cS$ with \textbf{measurement} $\mu$, i.e.,
    \begin{equation}
        \dutchcal{s} \psi = (-1)^{\mu s} \psi, \quad \forall \dutchcal{s}\in \cS
    \end{equation}
    where $\dutchcal{s}\mapsto s$ by abelianization.
\end{corollary}


\begin{lemma}[Physical Measurement]
    \label{lem:canonical-iso-non-CSS}
    Let $A=S\to P\to S^*$ be a symplectic complex with stabilizer map $\sigma$. Then there exists a canonical isomorphism $[a] \mapsto \mu:P/\ker \hat{\sigma} \to (\im \sigma)^*$ defined by
    \begin{equation}
        \mu(\cdots )=\Lambda(a,\cdots)
    \end{equation}
    where $a$ is any representation of equivalence class $[a]$, i.e., modulo logical.
\end{lemma}

\begin{proof}
    It's clear that $P/\ker \hat{\sigma}\to (\im \sigma)^*$ via $[a]\mapsto \Lambda(a,\cdots)$ is well-defined linear and injective, and thus it remains to show that the map is surjective.
    Since $\mu\in (\im \sigma)^*$, by simple basis extension, we can extend $\mu$ to a linear map $\tilde{\mu}$ on $P$.
    Since the symplectic form $\Lambda$ is non-degenerate, there exists unique $a$ such that $\tilde{\mu}(\cdots )=\Lambda(a,\cdots)$. In particular, we see that $\mu(\sigma s) = \Lambda(a,\sigma s) =(\hat{\sigma}a)(s)$ and thus the statement follows.

\end{proof}

\begin{remark}[Physical Errors]
    Similar to Remark \ref{rem:phys-err} for CSS codes, $a$ can be regarded as a physical Pauli (modulo logicals). In particular, if $\cS$ is a CSS code with measurement $a_{\bullet}$, then $a=a_X\oplus a_Z$ modulo logicals.
\end{remark}

\begin{remark}[Valid Syndrome]
    Note that $\hat{\sigma}:P/\ker \hat{\sigma} \to \im \hat{\sigma}$ is an isomorphism.
    Since $S^*$ is the space of \textbf{syndromes}, we refer to $\im \hat{\sigma}$ as the space of \textbf{valid syndromes}.
\end{remark}



Similar to Lemma \ref{lem:overcomplete-physicals}, the non-CSS measurement can be given below.
\begin{lemma}[Overcomplete Measurement]
    \label{lem:overcomplete-non-CSS-physicals}
    Let $\psi$ be stabilized by a non-CSS code $\cS^A$ with symplectic complex
    \begin{equation}
        A=S^A \to P^A \to S^{A*}
    \end{equation}
    with physical measurements $a$ modulo logicals.
    Measure Paulis in commuting subgroup $\cS^B$ based on map $\sigma^B:S^B\to P^A$ with physical measurement $b$ modulo logical. Then the measurement is valid iff
    \begin{equation}
        a+b \in \ker \hat{\sigma}^A +\ker \hat{\sigma}^B
    \end{equation}
    Moreover, if $\alpha\in \ker \hat{\sigma}^A$ and $\beta \in \ker \hat{\sigma}^B$ are such that $c\equiv a +\alpha =b +\beta$, then the output state is stabilized by Paulis in $\bra \cS^A ,\cS^B\ket \cap (\cS^B)^\perp$ with abelianization
    \begin{equation}
        (\im \sigma^A +\im \sigma^B) \cap \ker \hat{\sigma}^B
    \end{equation}
    and physical measurement $c$ mod logicals.
\end{lemma}

\begin{remark}[CSS Measurement]
    Recall that if we perform a CSS, say, $X$-type, measurement so that $\im \sigma^B \subseteq Q_X^A$, then $b$ can always be chosen in the $Z$-sector, i.e., $b\in Q_Z^A$, since it's modulo logicals, and thus the $X$-type physical error $a_X$ (mod logicals) does not change.
\end{remark}

\subsection{Non-CSS Fault Complexes}
\label{sec:non-CSS-fault}
In this section, we derive the spacetime fault complexes for $d$ rounds of measurement, for both syndrome extraction in Algorithm \ref{alg:non-CSS-syndrome-extract} and logical measurement in Algorithm \ref{alg:non-CSS-logical-measure}.
Note that the detailed algorithms are postponed, since the relevant information is the measurement obtained each round, along with the possibility that physical errors can accumulate (with small probability) between rounds of measurement.

Also note that despite dealing with non-CSS codes which are represented as symplectic complexes, the spacetime fault complex is not necessarily symplectic, and thus we will utilize the same Lemma \ref{lem:general-threshold}.


\subsubsection{Syndrome Extraction}
\label{sec:non-CSS-syndrome-fault-complex}

Given a quantum state stabilized by a (non-CSS) data code $D=S^D\to P^D\to \bar{S}^D$, input the state into Algorithm \ref{alg:non-CSS-syndrome-extract} and use the output as the input for subsequent syndrome extractions of Algorithm \ref{alg:non-CSS-syndrome-extract}.
Perform the syndrome extraction $T+1$ times (where we assume that that last measurement is perfect without measurement errors), and assume that physical errors can accumulate on $P^D$ before each extraction, so that we can regard physical errors $e(P^D)|_t$ occuring at integer time steps $t\in [0,T]$ and measurement errors $e(\bar{S}^D)|_{t^+}\in \bar{S}^D$ occurring at half-integer time steps $t^+$ where $t\in [0,T)$.

In particular, after $t+1$ rounds of subsequent extraction, we have readout
\begin{equation}
    \hat{\sigma}^D \left(\sum_{s\le t} e(P^D)|_s\right) +e(\bar{S}^D)|_{t^+}
\end{equation}
where the summation denotes the accumulated physical error, and we assumed that $e(\bar{S}^D)|_{0^-},e(\bar{S}^D)|_{T^+}=0$.
Equivalently, we can compute the difference between readouts $t^+$ and and $t^-$ to obtain readouts
\begin{equation}
    \label{eq:non-CSS-syndrome-readout}
    \hat{\sigma}^D e(P^D)|_t+e(\bar{S}^D)|_{t^+}+e(\bar{S}^D)|_{t^-}, \quad t\in [0,T]
\end{equation}
This corresponds to the diagram
\begin{equation}
    \begin{tikzpicture}[baseline]
    \matrix(a)[matrix of math nodes, nodes in empty cells, nodes={minimum size=25pt},
    row sep=1.5em, column sep=1.5em,
    text height=1.25ex, text depth=0.25ex]
    { &   & \bar{S}^D E \\
      & P^D V & \bar{S}^D V\\};
    \path[->,font=\scriptsize]
    (a-2-2) edge node[above]{$\hat{\sigma}$} (a-2-3)
    (a-1-3) edge node[right]{$\partial^R$} (a-2-3);
    \end{tikzpicture}
\end{equation}
where $R=E\to V$ denote the repetition code on $T+1$ vertices $|t\ket,t\in [0,T]$ with differential $\partial^R$. It's then clear that the physical errors occupy space $P^D V$ and the measurement errors occupy $\bar{S}^D E$, where we have omitted the tensor product $\otimes$ for simplicity, so that the direct sum $P^D V\oplus \bar{S}^D E$ indicates the \textbf{fault locations (data and syndrome qubits)}.
The readouts in Eq. \eqref{eq:non-CSS-syndrome-readout} occupy $\bar{S}^D V$ and are determined by the sum of the differentials.
In particular, $\bar{S}^D V$ is the space of \textbf{fault syndromes}, differential $\partial^F = \partial^R+\hat{\sigma}^D$ is the \textbf{syndrome map} and $\im \partial^F$ denotes the space of \textbf{valid fault syndromes}.

To create a fault complex, we must also introduce \textbf{fault stabilizers} to indicate whether accumulated errors are trivial logicals. Specifically, writing $\be(P^D)=\sum_t e(P^D)|_t\otimes |t\ket \in P^D V$ for the vectorized error, we have the following.
\begin{lemma}[Fault Stabilizers]
    \label{lem:non-CSS-fault-stabilizers}
    If $\be(P^D)\in P^D V$, then the accumulated error
    \begin{equation}
        \mathfrak{e}(P^D) = \sum_{t=0}^T e(P^D)|_{t}
    \end{equation}
    is a trivial logical, i.e., $\in \im \sigma^D \Leftrightarrow \be(P^D)$ is in the image of the map given by the red arrows as follows.
    \begin{equation}
        \label{eq:non-CSS-fault-stabilizers}
        \begin{tikzpicture}[baseline]
        \matrix(a)[matrix of math nodes, nodes in empty cells, nodes={minimum size=25pt},
        row sep=1.5em, column sep=1.5em,
        text height=1.25ex, text depth=0.25ex]
        { & P^D E & \bar{S}^D E \\
         S^D V & P^D V & \bar{S}^D V\\};
        \path[->,font=\scriptsize]
        (a-2-2) edge node[above]{$\hat{\sigma}^D$} (a-2-3)
        (a-1-3) edge node[right]{$\partial^R$} (a-2-3);
        \path[->,red,font=\scriptsize]
        (a-2-1) edge node[above]{$\sigma^D$} (a-2-2)
        (a-1-2) edge node[right]{$\partial^R$} (a-2-2);
        \end{tikzpicture}
    \end{equation}
\end{lemma}
\begin{remark}
    Physically, the map $P^D E \to P^D V$ denotes physical errors that come in pairs and thus do not contribute to the accumulated error since they cancel each other out. We refer to $P^D E\oplus S^D V$ as the \textbf{fault stabilizers}.
\end{remark}
\begin{proof}
    The proof is exactly the same as Lemma \ref{lem:fault-stabilizers} and thus omitted.
\end{proof}

Let us further complete the diagram in Eq. \eqref{eq:non-CSS-fault-stabilizers} to form a complex.
\begin{lemma}
    Consider the map indicated by the red arrow.
    \begin{equation}
        \begin{tikzpicture}[baseline]
        \matrix(a)[matrix of math nodes, nodes in empty cells, nodes={minimum size=25pt},
        row sep=1.5em, column sep=1.5em,
        text height=1.25ex, text depth=0.25ex]
        { & P^D E & \bar{S}^D E \\
         S^D V & P^D V & \bar{S}^D V\\};
        \path[->,font=\scriptsize]
        (a-2-2) edge node[above]{$\hat{\sigma}^D$} (a-2-3)
        (a-1-3) edge node[right]{$\partial^R$} (a-2-3)
        (a-2-1) edge node[above]{$\sigma^D$} (a-2-2)
        (a-1-2) edge node[right]{$\partial^R$} (a-2-2);
        \path[->,red,font=\scriptsize]
        (a-1-2) edge (a-1-3);
        \end{tikzpicture}
    \end{equation}
    Then the fault stabilizers has trivial fault syndrome $\Leftrightarrow$ the map is $\hat{\sigma}^D$.
\end{lemma}
\begin{proof}
    Proof exactly the same as Lemma \ref{lem:fault-complex} and thus omitted.
\end{proof}

The \textbf{fault complex} $F=F_2\to F_1\to F_0$ for syndrome extraction is then given by
\begin{equation}
    \label{eq:noncss-fault-complex}
    \begin{tikzpicture}[baseline]
    \matrix(a)[matrix of math nodes, nodes in empty cells, nodes={minimum size=25pt},
    row sep=1.5em, column sep=1.5em,
    text height=1.25ex, text depth=0.25ex]
    { & P^D E & \bar{S}^D E \\
         S^D V & P^D V & \bar{S}^D V\\};
    \path[->,font=\scriptsize]
    (a-2-2) edge node[above]{$\hat{\sigma}^D$} (a-2-3)
    (a-1-3) edge node[right]{$\partial^R$} (a-2-3)
    (a-2-1) edge node[above]{$\sigma^D$} (a-2-2)
    (a-1-2) edge node[right]{$\partial^R$} (a-2-2)
    (a-1-2) edge node[above]{$\hat{\sigma}^D$} (a-1-3);
    \end{tikzpicture}
\end{equation}
which one may notice corresponds to the tensor product $D\otimes R$, where $S^D E$ doesn't play a role.
The fault complex then provides a natural threshold.


\begin{theorem}[Threshold]
    \label{thm:non-CSS-syndrome-threshold}
    Suppose that $T\ge d(D)$ where $d(D)$ is the code distance of the qLDPC (non-CSS) data code $D$.
    Let (physical and measurement) errors
    \begin{equation}
        \be(F_1)= \be(P^D)\oplus \be(\bar{S}^D)
    \end{equation}
    occur with probability $p$.
    Suppose that readouts $\partial^F \be(F_1)$ and min-weight correction $\partial^F \be(F_1 )\mapsto \hat{\be}(F_1)$ relative to the fault complex $F$ are provided.
    Then there exists threshold $p_c$ such that if $p<p_c$, then the accumulated physical error of $\be(P^D)+\hat{\be}(P^D)$ is nontrivial with probability
    \begin{equation}
        \le O(1) |\cQ^D| d(D) \left(\frac{p}{p_c}\right)^{d(D)/2}
    \end{equation}
\end{theorem}
\begin{proof}
    Exactly the same as Theorem \ref{thm:syndrome-threshold} and thus omitted.
\end{proof}

\subsubsection{Logical Measurement}
\label{sec:non-CSS-logical-fault-complex}

In this subsection, we derive the fault complex of logical measurement in Algorithm \ref{alg:non-CSS-logical-measure}, which by Remark \ref{rem:recovery-and-threshold}, will be similar to that in Section \ref{sec:logical-fault-complex}.
Specifically, suppose that we are given a quantum state stabilized by a data (non-CSS) code $D$ written as a symplectic complex
\begin{equation}
    D= S^D \xrightarrow{\sigma^D} P^D \xrightarrow{\hat{\sigma}^D} \bar{S}^D
\end{equation}
(equipped with the usual objects), and a logical $\ell^{\star}(P^D)$ of $D$.
Input the state with possible unknown errors $e(P^D)|_0$ into Algorithm \ref{alg:non-CSS-logical-measure} so that we obtain readouts (where $\hat{\sigma},P,\bar{S}$ refer to the height-2 cone $C$ of Theorem \ref{thm:symplectic-embedding} assembled from $A$ and $D$)
\begin{align}
    \hat{\sigma} \tilde{e}(P)|_0 +e(\bar{S})|_{0^+} \in \bar{S} \\
    \label{eq:non-CSS-ancilla-readout-0}
    \delta^A x^{A}|_0 +e(Q_X^{A})|_0\in Q_X^{A}
\end{align}
where
\begin{equation}
    \tilde{e}(P) =
    \begin{pmatrix}
        e(Q_X^A) \\
        e(P^D) \\
        \tilde{e}(Q_Z^A)
    \end{pmatrix}
\end{equation}
and the output state has error $e(P^D)|_0+gx^A|_{0}$.
Note that similar to CSS logical measurement in Section \ref{sec:logical-fault-complex}, we have no information (and no control over the size) regarding the $Z$-type errors on ancilla qubits $\tilde{e}(Q_Z^A)$ since the ancilla qubits are measured out in the $Z$-basis at the end of every round of measurement.
Hence, the only relevant information is instead
\begin{align}
    \label{eq:non-CSS-relevant-readout}
    \begin{pmatrix}
        \delta^A & 0\\
        \hat{h} & \hat{\sigma}^D
    \end{pmatrix}
    \left.
    \begin{pmatrix}
        e(Q_X^A) \\
        e(P^D)
    \end{pmatrix}\right|_{0}
    +
    \left.
    \begin{pmatrix}
        e(\bar{Z}^A) \\
        e(\bar{S}^D)
    \end{pmatrix}\right|_{0^+}
    &\in
    \begin{pmatrix}
        \bar{Z}^A \\
        \bar{S}^D
    \end{pmatrix} \\
    \delta^A x^A|_0+ e(Q_X^A)|_0 &\in Q_X^A
\end{align}
In particular, note that Eq. \eqref{eq:non-CSS-relevant-readout} denotes the differential of the following cone complex (not symplectic anymore), which we denote as $G=G_2 \to G_1 \to G_0$ with differential $\partial^{G}$, so that $G_2=X^A\oplus S^D$, $G_1=Q_X^A\oplus P^D$ and $G_0=\bar{Z}^A\oplus \bar{S}^D$.
\begin{equation}
    \begin{tikzpicture}[baseline]
        \matrix(a)[matrix of math nodes, nodes in empty cells, nodes={minimum size=25pt},
        row sep=1.5em, column sep=1.5em,
        text height=1.25ex, text depth=0.25ex]
        { & X^A & Q_X^A & \bar{Z}^A\\
         S^D & P^D & \bar{S}^D &\\};
        \path[->,font=\scriptsize]
        (a-1-2) edge node[above]{$\delta^A$} (a-1-3)
        (a-1-3) edge node[above]{$\delta^A$} (a-1-4)
        (a-2-2) edge node[above]{$\hat{\sigma}^D$} (a-2-3)
        (a-1-2) edge node[right]{$g$} (a-2-2)
        (a-1-3) edge node[right]{$h$} (a-2-3)
        (a-2-1) edge node[above]{$\sigma^D$} (a-2-2);
    \end{tikzpicture}
\end{equation}
Note that $G$ has 1-systolic distance $d_1(G)=\Omega(1)d(D)$ due to the relative expansion of $g$.
Hence, we can rewrite the relevant readouts as
\begin{align}
    \label{eq:non-CSS-relevant-readout-rewrite}
    \partial^G e(G_1)|_{0} +e(G_0)|_{0^+} &\in G_0 \\
    \delta^A x^A|_{0} + e(Q_X^A)|_0 &\in Q_X^A
\end{align}

Similar to conventional CSS logical measurement, we use the output as the input for another round of logical measurement to obtain relevant readouts
\begin{align}
    \partial^G (e(P^D)|_0 +gx^A|_0 + e(G_1)|_1)+e(G_0)|_{1^+} &\in G_0 \\
    \delta^A x^A|_{1} + e(Q_X^A)|_{1} &\in Q_X^A
\end{align}
Note that
\begin{align}
    \partial^G (e(P^D)|_0 +&gx^A|_0 + e(G_1)|_1) \\
    &= \partial^G
    \begin{pmatrix}
        \delta^A x^A|_0 +e(Q^A_X)|_1\\
        e(P^D)|_0+e(P^D)|_1
    \end{pmatrix}\\
    &= \partial^G (e(G_1)|_0+e(G_1)|_1) \\
    &\quad\quad + \partial^G \text{Eq. } \eqref{eq:non-CSS-ancilla-readout-0}
\end{align}
Hence, by utilizing the readout in Eq. \eqref{eq:non-CSS-ancilla-readout-0} when discarding the ancilla qubits, we can regard the readout as that error caused by the accumulated error $e(G_1)|_0+e(G_1)|_1$ even though errors $e(Q^A_X)|_0$ on the ancilla were discarded.
Specifically, we have information
\begin{align}
    \partial^G (e(G_1)|_0 + e(G_1)|_1)+e(G_0)|_{1^+} &\in G_0 \\
    \delta^A x^A|_{1} + e(Q_X^A)|_{1} &\in Q_X^A
\end{align}

Repeat this procedure $T+1$ times (where the last logical measurement is assumed to be perfect without measurement errors).
Then by computing the difference between readouts $t^+$ and $t^-$, we obtain information
\begin{align}
    \label{eq:non-CSS-logical-readout}
    \partial^G e(G_1)|_t+e(G_0)|_{t^+} +e(G_0)|_{t^-}, \quad t\in [0,T]  \\
    \label{eq:non-CSS-ancilla-readout}
    \delta^A x^A|_t + e(Q_X^A)|_t\in Q_X^A, \quad t\in [0,T]
\end{align}
where measurement errors $e(G_0)|_{0^-},e(G_0)|_{T^+}=0$ by assumption, and the output state has accumulated physical errors
\begin{equation}
    \sum_{t\le T} e(P^D)|_t +g\sum_{t\le T} x^A|_t
\end{equation}

Using Eq. \eqref{eq:non-CSS-logical-readout}, we can then obtain a fault complex $F=G\otimes R$ similar to Eq. \eqref{eq:fault-complex} with $\partial^D \mapsto \partial^G$ replaced with the differential of the deformed code $G$, and $R$ is the repetition code on $T+1\ge d_1(G)=\Omega(1)d(D)$ vertices.
We can then obtain a threshold similar to Theorem \ref{thm:syndrome-threshold}.
\begin{theorem}[Non-CSS Logical Threshold]
    \label{thm:non-CSS-logical-threshold}
    Suppose that $T+1\ge d_1(G)$ where $d_1(G)$ is the 1-systolic distance of the deformed code $G$. Let fault errors
    \begin{equation}
        \be(F_1) =\be(G_1) \oplus \be(G_0)
    \end{equation}
    occur with probability $p$.
    Suppose that readouts $\partial^F \be(F_1)$ and min-weight correction $\partial^F \be(F_1)\mapsto \hat{\be}(F_1)$ relative to the fault complex $F=G\otimes R$ are provided.
    Then there exists threshold $p_c$ such that if $p<p_c$, then the accumulated physical error of $\be(P^D)+\hat{\be}(P^D)$ is nontrivial with probability
    \begin{equation}
        \le O(1) \dim G_1 \, d_1(G) \left(\frac{p}{p_c}\right)^{d_1(G)/2}
    \end{equation}
    Note that the ancilla gadget $A$ was chosen so that the merged code $G$ has distance $d_1(G) = \Omega(1) d(D)$.
\end{theorem}

\begin{remark}[Recovery]
    Note, however, that the accumulated physical error on the data code is given by
    \begin{equation}
        \sum_{t\le T} (e(P^D)|_t+ gx^A|_t)
    \end{equation}
    and thus Theorem \ref{thm:non-CSS-logical-threshold} only guarantees that, with high probability,
    \begin{align}
        \sum_{t\le T} (e(Q_X^A)|_t +\hat{e}(Q_X^A)|_t) &= \delta^A s^A \\
        \sum_{t\le T} (e(P^D)|_t +\hat{e}(P^D)|_t) &= g s^A +\sigma^D s^D
    \end{align}
    for some $s=s^A \oplus s^D\in G_2 = X^A\oplus S^D$.
    Hence, we must utilize the additional readouts in Eq. \eqref{eq:non-CSS-ancilla-readout} to apply a final recovery procedure.
    Specifically, sum over all $t\in [0,T]$ in Eq. \eqref{eq:non-CSS-ancilla-readout}, along with the correction $\hat{e}(Q_X^A)$ to obtain information
    \begin{equation}
        \delta^A\left(\sum_{t\le T} x^A|_t  +s^A\right) \in \im \delta^A
    \end{equation}
    Choose any $\hat{s}^A \in X^A$ such that
    \begin{equation}
        \delta^A\left(\sum_{t\le T} x^A|_t  +s^A +\hat{s}^A\right)=0
    \end{equation}
    And apply the Pauli onto data qubits corresponding to $g\hat{s}^A$.
    Then, in addition to correction $\hat{e}(P^D)$, we see that, with high probability, the accumulated data error belongs to $\im \sigma^D + g\ker \delta^A$, which is simply a stabilizer of $D$ plus the measured logical, and thus the accumulated error is corrected.
\end{remark}
\section{Implementation}
\label{sec:implement}

For the sake of completeness, we provide a detailed derivation of the implementation of measurement, where subsections \ref{sec:low-lvl-CSS-measurement} and \ref{sec:low-lvl-non-CSS-measure} treat the special CSS and more general non-CSS case separately.
In particular, we show that measurement errors arise due to physical errors on check qubits.

\subsection{Low Level CSS Measurement}
\label{sec:low-lvl-CSS-measurement}

We adopt the convention that given a $\F_2$ space $X$, the space is always equipped with basis $\cX$. We also adopt the notation that $Q(X)$ is an (isometric) copy of $X$ with corresponding basis $\cQ(X)$.

\begin{lemma}[Low Level]
\label{lem:low-lvl-measurement}
Algorithm \ref{alg:low-lvl-measurement} is correct.
\end{lemma}

\begin{remark}[Physical $\lr$ Measurement Error]
    \label{rem:phys-meas-errs}
    Note that at the low-level implementation, the $Z$-type physical error $e_Z(X^B)$ on discarded qubits converts to a measurement error in the readout.
\end{remark}

\begin{proof}
    By Lemma \ref{lem:overcomplete-physicals}, the valid measurement condition is equivalent to
    \begin{equation}
        e_Z + s_Z \in \ker\delta +Q(X^B)^\perp
    \end{equation}
    where $s_Z(X^B) \in Q(X^B)$ is the readout and $e_Z =e_Z(X^B) \oplus e_Z(Q^A)$ where $e_Z(X^B) \in Q(X^B)$ and $e_Z(Q^A) \in Q^A$. Note that if $\ell_Z\in \ker \delta$, we have $\ell_Z(Q^A) \in \ker \delta^A$ and $\ell_Z(X^B) = \delta^B \ell_Z(Q^A)$. Hence,
    \begin{equation}
        s_Z = \delta^B \ell_Z(Q^A) + e_Z(X^B)
    \end{equation}
    Again, by Lemma \ref{lem:overcomplete-physicals}, the residual $Z$-type physical error, modulo logical, is then $e_Z(Q^A) +\ell_Z(Q^A)$, since qubits $Q(X^B)$ are discarded.
    The $X$-type residual physical errors are straightforward.
\end{proof}

\begin{lemma}[High Level Measurement]
    Algorithm \ref{alg:high-lvl-measurement} is correct.
\end{lemma}

\begin{algorithm}
\caption{\textsc{Low-Level Measurement}}
\label{alg:low-lvl-measurement}
\begin{algorithmic}[1]
    \Require State stabilized by complex with (co)differential $\partial$ ($\delta$)
    \begin{equation}
        \label{eq:low-lvl-measurement}
        \begin{tikzpicture}[baseline]
        \matrix(a)[matrix of math nodes, nodes in empty cells, nodes={minimum size=25pt},
        row sep=1.5em, column sep=1.5em,
        text height=1.25ex, text depth=0.25ex]
        {&X^B & Q(X^B)\\
         X^A &Q^A & Z^A\\};
        \path[->,font=\scriptsize]
        (a-1-2) edge node[above]{$\rm{id}$}  (a-1-3)
        (a-2-1) edge node[above]{$\partial^A$}  (a-2-2)
        (a-2-2) edge node[above]{$\partial^A$}  (a-2-3)
        (a-1-2) edge node[right]{$\partial^B$}  (a-2-2)
        (a-1-3) edge node[right]{$g=\partial^A \partial^B$}  (a-2-3);
        \end{tikzpicture}
    \end{equation}
    (where the bottom row is also a complex $A$) with (unknown) physical $e_{\bullet}(Q^A)\in Q^A,e_{\bullet}(X^B) \in Q(X^B)$.
    \Ensure Readout
    \begin{equation}
        \delta^B \ell_Z(Q^A) +e_Z(X^B) \in X^B
    \end{equation}
    for some logical $\ell_Z(Q^A) \in \ker \delta^A$ and output state stabilized by
    \begin{itemize}
        \item $X$ Paulis in $\im \partial^A +\im \partial^B$ with physical $e_Z(Q^A)+\ell_Z(Q^A)$
        \item $Z$ Paulis in $\im \delta^A \cap (\im \partial^B)^\perp$ with physical $e_X(Q^A)$
    \end{itemize}
    \State Measure out $Q(X^B)$ in the $X$-basis with physical measurement $s\in Q(X^B) \cong X^B$
\end{algorithmic}
\end{algorithm}

\begin{algorithm}
\caption{\textsc{High-Level Measurement}}
\label{alg:high-lvl-measurement}
\begin{algorithmic}[1]
    \Require State stabilized by CSS code
    \begin{equation}
        A= X^A \xrightarrow{\partial^A} Q^A \xrightarrow{\partial^A} Z^A
    \end{equation}
    with (unknown) physical $e_{\bullet}(Q^A)$.
    Circuit to perform $X$-type measurement $\partial^B:X^B \to Q^A$.
    \Ensure Valid syndrome
    \begin{equation}
        \delta^B (e_Z(Q^A) +\ell_Z(Q^A)) \in X^B
    \end{equation}
    for some logical $\ell_Z(Q^A) \in \ker \delta^A$ and output stabilized by
    \begin{itemize}
        \item $X$ Paulis in $\im \partial^A +\im \partial^B$ with physical $e_Z(Q^A)+\ell_Z(Q^A)$
        \item $Z$ Paulis in $\im \delta^A \cap (\im \partial^B)^\perp$ with physical $e_X(Q^A)$
    \end{itemize}
    \State Perform $X$-type measurement $\partial^B:X^B \to Q^A$ with valid syndrome $\in X^B$
\end{algorithmic}
\end{algorithm}

\begin{algorithm}
\caption{\textsc{Low-Level with CNOTs}}
\label{alg:low-lvl-wCNOTs}
\begin{algorithmic}[1]
    \Require State stabilized by CSS code
    \begin{equation}
        A= X^A \xrightarrow{\partial^A} Q^A \xrightarrow{\partial^A} Z^A
    \end{equation}
    with (unknown) physical $e_{\bullet}(Q^A)$.
    \Ensure Measure $X$-Paulis in $\partial^B:X^B \to Q^A$ with readout
    \begin{equation}
        \delta^B (e_Z(Q^A) +\ell_Z(Q^A)) +e_Z(X^B) \in X^B
    \end{equation}
     where $e_Z(X^B)$ is an unknown initialization error, and $\ell_Z(Q^A) \in \ker \delta^A$. The output state is stabilized by
    \begin{itemize}
        \item $X$ Paulis in $\im \partial^A +\im \partial^B$ with physical $e_Z(Q^A)+\ell_Z(Q^A)$
        \item $Z$ Paulis in $\im \delta^A \cap (\im \partial^B)^\perp$ with physical $e_X(Q^A)$
    \end{itemize}
    \State Initialize check qubits $Q(X^B)$ transversally in $X$-basis with (unknown) physicals $e_Z(X^B)$
    \Comment{Example \ref{ex:initialization}}
    \State Apply CNOT circuits from $Q(X^B)\to Q^A$ based on $\partial^B$
    \State Measure out $Q(X^B)$ in the $X$-basis with physical $\in Q(X^B) \cong X^B$
\end{algorithmic}
\end{algorithm}

\begin{lemma}[Relation]
    Algorithm \ref{alg:low-lvl-wCNOTs} is correct.
\end{lemma}
\begin{proof}
    After initialization (line 1), the state is stabilized by
    \begin{equation}
        \begin{tikzpicture}[baseline]
        \matrix(a)[matrix of math nodes, nodes in empty cells, nodes={minimum size=25pt},
        row sep=1.5em, column sep=1.5em,
        text height=1.25ex, text depth=0.25ex]
        {&X^B & Q(X^B)\\
         X^A &Q^A & Z^A\\};
        \path[->,font=\scriptsize]
        (a-1-2) edge node[above]{$\rm{id}$}  (a-1-3)
        (a-2-1) edge node[above]{$\partial^A$}  (a-2-2)
        (a-2-2) edge node[above]{$\partial^A$}  (a-2-3);
        \end{tikzpicture}
    \end{equation}
    By applying a CNOT, the state is stabilized by
    \begin{equation}
        \begin{tikzpicture}[baseline]
        \matrix(a)[matrix of math nodes, nodes in empty cells, nodes={minimum size=25pt},
        row sep=1.5em, column sep=1.5em,
        text height=1.25ex, text depth=0.25ex]
        {&X^B & Q(X^B)\\
         X^A &Q^A & Z^A\\};
        \path[->,font=\scriptsize]
        (a-1-2) edge node[above]{$\rm{id}$}  (a-1-3)
        (a-2-1) edge node[above]{$\partial^A$}  (a-2-2)
        (a-2-2) edge node[above]{$\partial^A$}  (a-2-3)
        (a-1-2) edge node[right]{$\partial^B$}  (a-2-2)
        (a-1-3) edge node[right]{$g=\partial^A \partial^B$}  (a-2-3);
        \path[->, >=stealth, dashed, font=\scriptsize]
        (a-1-3) edge (a-2-2);
        \end{tikzpicture}
    \end{equation}
    where the dashed arrow denotes the CNOT circuit implementing $\partial^B$.
    Note that CNOT maps control $c$ and target qubit Paulis as follows
    \begin{align}
        Z_c &\mapsto Z_c\\
        Z_t &\mapsto Z_c Z_t \\
        X_c &\mapsto X_c X_t \\
        X_t &\mapsto X_t
    \end{align}
    Hence, after line 2, the state has $Z$-type physical $e_Z(Q^A)\in Q^A$ and $e_Z(X^B) +\delta^B e_Z(Q^A) \in Q(X^B)$.
    The statement then follows from Lemma \ref{lem:low-lvl-measurement}.
\end{proof}
\begin{remark}
    Note that the only difference between Algorithm \ref{alg:high-lvl-measurement} and \ref{alg:low-lvl-wCNOTs} is the measurement error as described in Remark \ref{rem:phys-meas-errs}.
\end{remark}


\subsubsection{Syndrome Extraction}
\label{sec:syndrome-extraction}
In this section, we shall apply the low level implementation of a measurement in the context of syndrome extraction of a data code.
At a high-level, we shall utilize the following corollary, which implies that if an input state has physical errors, those error can be detected via subsequential measurement without modifying the stabilizer code.

\begin{corollary}[Subsequent Measurement]
    Suppose that  $X^B$ is a copy of $X^A$ with $\partial^A = \partial^B$.
    Then Algorithm \ref{alg:low-lvl-wCNOTs} has readout $\delta^A e_Z(Q^A) +e_Z(X^B)$ with output state stabilized by $A$ with unknown physical $e_{\bullet}(Q^A)$.
\end{corollary}

Let us further elaborate the syndrome extraction as follows.

\begin{definition}[Code Gadget]
    \label{def:code-gadget}
    Given a CSS code $D=X^D\to Q^D\to Z^D$ with basis $\cX^D,\cQ^D,\cZ^D$ and (co)differential $\partial^D$ ($\delta^D$), the corresponding \textbf{code gadget} is the collection of
    \begin{itemize}
        \item Qubits $\cQ^D$ and check qubits $\cQ(X^D),\cQ(Z^D)$
        \item CNOT circuits such that a CNOT from check qubits $Q(X)$ to data qubits $Q$ can be implemented in $O(1)$ circuit depth based on the differential $\partial:Q(X)\cong X\to Q$.
        Similarly, a CNOT  from data qubits $Q$ to check qubits $Q(Z)$ can be implement in $O(1)$ depth based on $\partial:X\to Q(Z)\cong Z$.
    \end{itemize}
    We shall always assume that qubits can be initialized transversally relative to the $X$ or $Z$ basis with unknown but presumably small physical errors, and can be measured transversally perfectly.
\end{definition}

The following algorithm then addresses the initialization and subsequent measurement of a syndrome extraction process.
\begin{algorithm}
\caption{\textsc{Syndrome Extraction}}
\label{alg:syndrome-extract}
\begin{algorithmic}[1]
    \Require Data code $D$ gadget.
    \Ensure Readout
    \begin{align}
        \delta^D \tilde{e}_Z(Q^D) +e_Z(X^D) \in X^D \\
        \partial^D \tilde{e}_X(Q^D)+ e_X(Z^D) \in Z^D
    \end{align}
    for some $\tilde{e}_Z(Q^D),\tilde{e}_X(Q^D),e_Z(X^D),e_X(Z^D)$, and output state stabilized by $D$ with physical error $\tilde{e}_{\bullet}(Q^D)$ mod logicals
    \State If input state is not given, initialize $Q^D$ in $Z$-basis with unknown physical $\tilde{e}_{X}(Q^D)$
    \State Initialize $Q(X^D),Q(Z^D)$ in the $X,Z$ basis, with unknown physical $e_{Z}(X^D),e_{X}(Z^D)$, respectively
    \State Apply $O(1)$ depth CNOT circuit for $\partial$
    \State Measure out $Q(X^D)$ transversally in the $X$-basis with readout
    \begin{equation}
        \delta^D \tilde{e}_Z(Q^D) +e_Z(X^D)
    \end{equation}
    for some $\tilde{e}_Z(Q^D)$ so that output state has $Z$-type physical measurement $\tilde{e}_Z(Q^D)$ mod logicals
    \State Measure out $Q(Z^D)$ transversally in the $Z$-basis with readout
    \begin{equation}
        \partial^D \tilde{e}_X(Q^D) + e_X(Z^D)
    \end{equation}
    so that the output state has $X$-type physical measurement $e_X(Q^D)$ mod logicals
\end{algorithmic}
\end{algorithm}


\begin{lemma}[Syndrome Extraction]
    \label{lem:syndrome-extraction}
    Algorithm \ref{alg:syndrome-extract} is correct.
    Moreover, if an input state stabilized by $D$ with physical $e_{\bullet}(Q^D)$ (e.g., subsequent measurement) is given, then the output is as given by Algorithm \ref{alg:syndrome-extract} with $\tilde{e}_{\bullet}(Q^D) = e_{\bullet}(Q^D)$.
\end{lemma}
\begin{remark}[Fault Tolerance with Threshold]
    The physical errors can then be corrected via the general error correction procedure \cite{gottesman2013fault}.
    Specifically, after $T\ge d$ rounds of subsequent syndrome extractions of Algorithm \ref{alg:syndrome-extract} where $d$ is the code distance of $D$, one can build a spacetime fault complex.
    A threshold can then be obtained by applying min-weight correction relative to the fault complex.
    See Section \ref{sec:syndrome-fault-complex} for details.
\end{remark}
\begin{proof}
    The proofs follows from Lemma \ref{lem:low-lvl-measurement}.
\end{proof}


\subsubsection{Logical Measurement}
\label{sec:logical-measure}
In this section, we shall consider the low level implementation of a logical measurement on a data code
\begin{equation}
    D=X^D \to Q^D\to Z^D
\end{equation}
Conventionally, an ancilla $A$ is constructed which is then \textit{merged} with the data code $D$ to form a deformed code.
The checks of the deformed code are then measured in a manner similar to syndrome extraction in Section \ref{sec:syndrome-extraction}, so that the logical measurement can be extracted using the general error correction procedure \cite{gottesman2013fault} after $d(D)$ rounds of subsequential measurement.

Let us first consider the ancilla gadget.
Currently, most ancillas for logical measurements follow the same underlying procedure, which first constructs a measurement graph defined as follow.
\begin{definition}[Measurement Graph]
    Given, say an $X$-type, logical operator $\ell^{\star}_X\in Q^D$ of data code $D$, the corresponding \textbf{measurement (multi) graph} $\cG=(\cE,\cV)$ is constructed with vertices $|q\ket \in \cV$ that are in one-to-one correspondence of qubits in $q\in \ell^{\star}_X\subseteq\cQ^{D}$.
    For every $z$ check of $D$, the support of $z$ must have an even overlap with $\ell^{\star}_X$ and thus can be paired up arbitrarily.
    Define an edge $\|q\bar{q};z \ket\in \cE $ for each pair of overlapping qubits $q,\bar{q}$ and each overlapping check $z$.
\end{definition}

Given the measurement graph of logical $\ell^{\star}_X$ with associated graph complex $G=E\to V$, a corresponding cell (co)complex
\begin{equation}
    \label{eq:ancilla}
    A=X^A \to Q^A \to Z^A
\end{equation}
is then constructed with $X$-checks, qubits and $Z$-checks corresponding to the vertices, edges and faces, such that $A$ is connected as a graph with only trivial (contractible) cycles, qLDPC, has low overhead \cite{yuan2026parsimonious,williamson2026low}, and contains the graph complex $G$ as a subgraph so that a natural projection (chain) map is depicted by the diagram below
\begin{equation}
    \label{eq:subgraph}
    \begin{tikzpicture}[baseline]
    \matrix(a)[matrix of math nodes, nodes in empty cells, nodes={minimum size=25pt},
    row sep=1.5em, column sep=1.5em,
    text height=1.25ex, text depth=0.25ex]
    {X^{A} &Q^A & Z^{A}\\
     V & E & \\};
    \path[->,font=\scriptsize]
    (a-1-1) edge (a-1-2)
    (a-1-2) edge (a-1-3)
    (a-2-1) edge (a-2-2)
    (a-1-1) edge node[right]{$\pi$} (a-2-1)
    (a-1-2) edge node[right]{$\pi$} (a-2-2);
    \end{tikzpicture}
\end{equation}
We refer to $A$ as a \textbf{(logical) ancilla} with respect to the measurement graph of logical $\ell^{\star}_X$.

Further note that by construction, the measurement graph induces a natural inclusion (chain) map $\iota$ depicted as follows.
\begin{equation}
    \label{eq:inclusion}
    \begin{tikzpicture}[baseline]
    \matrix(a)[matrix of math nodes, nodes in empty cells, nodes={minimum size=25pt},
    row sep=1.5em, column sep=1.5em,
    text height=1.25ex, text depth=0.25ex]
    { & V & E \\
     X^{D} & Q^{D} & Z^{D}\\};
    \path[->,font=\scriptsize]
    (a-1-2) edge (a-1-3)
    (a-2-1) edge (a-2-2)
    (a-2-2) edge (a-2-3)
    (a-1-2) edge node[right]{$\iota$} (a-2-2)
    (a-1-3) edge node[right]{$\iota$} (a-2-3);
    \end{tikzpicture}
\end{equation}
such that $\iota$ maps
\begin{align}
    |q\ket &\mapsto q\\
    \|q\bar{q};z\ket &\mapsto z
\end{align}
Hence, the composition $g=\iota \pi$ defines a chain map as depicted below
\begin{equation}
    \label{eq:cone}
    \begin{tikzpicture}[baseline]
    \matrix(a)[matrix of math nodes, nodes in empty cells, nodes={minimum size=25pt},
    row sep=1.5em, column sep=1.5em,
    text height=1.25ex, text depth=0.25ex]
    { & X^{A} & Q^{A} & Z^{A} \\
     X^{D} & Q^{D} & Z^{D} &\\};
    \path[->,font=\scriptsize]
    (a-1-2) edge (a-1-3)
    (a-1-3) edge (a-1-4)
    (a-2-1) edge (a-2-2)
    (a-2-2) edge (a-2-3)
    (a-1-2) edge node[right]{$g$} (a-2-2)
    (a-1-3) edge node[right]{$g$} (a-2-3);
    \end{tikzpicture}
\end{equation}
The corresponding $\cone(g:A\to D)$ is the \textbf{deformed code}, which we label as
\begin{equation}
    \cone(g) = X \to Q\to Z
\end{equation}
equipped with (co)differential $\partial$ ($\delta =\partial^{\top}$).
We refer to $g$ as the \textbf{(relative) connectivity} from the ancilla to the data code. Note that despite the construction of $A$, the map $g$ only depends on its restriction to the measurement graph, and thus $g$ defines the relative connectivity from the measurement graph to the data code.

\begin{definition}[(Relative) Connectivity Gadget]
    The \textbf{relative connectivity gadget} is a $O(1)$ depth circuit which implements a CNOT from ancilla check qubits $Q(X^A)$ to data qubits $Q^D$ based on
    \begin{equation}
        g:Q(X^A) \cong X^A \to Q^D
    \end{equation}
    and similarly for $g:Q^A \to Q(Z^D)$.
\end{definition}


\begin{theorem}[Logical Measurement]
    \label{thm:logical-measure}
    Algorithm \ref{alg:logical-measure} is correct.
    In particular, if no errors occur, then the Algorithm provides readout $s\in X$ such that the $X$-type logical $\ell^{\star}_X$ has measurement
    \begin{equation}
        \bra s|\cX^A\ket \in \F_2
    \end{equation}
    where $\cX^A\in X^A$ denotes the sum over all $X$-checks (vertices) of the ancilla.
    Moreover, if the input state is also stabilized by logical $\ell^{\star}_X$ (e.g., subsequent measurement), then we replace $\tilde{e}_Z(Q^D)\to e_Z(Q^D)$ in the output (but $\tilde{e}_Z(Q^A)$ is still unknown since it captures the logical measurement).
\end{theorem}

\begin{proof}
    After line 3, the coupled state is stabilized by the (CSS code) complex (with some errors) depicted by the following diagram
    \begin{equation}
        \begin{tikzpicture}[baseline]
        \matrix(a)[matrix of math nodes, nodes in empty cells, nodes={minimum size=25pt},
        row sep=1.5em, column sep=1.5em,
        text height=1.25ex, text depth=0.25ex]
        { & \red{X^{A}} & \darkgreen{Q^{A}} & \blue{Z^{A}} \\
         \red{X^{D}} & \darkgreen{Q^{D}} & \blue{Z^{D}} &\\};
        \end{tikzpicture}
    \end{equation}
    where, to keep the diagram simple, we have done some simplifications.
    Specifically, red nodes indicate that the corresponding (check) qubits are initialized in the $X$-basis, e.g., \red{$X^D$} indicates that the system has check qubits stabilized by the CSS code $X^D \to Q(X^D)$.
    Blue nodes indicate the corresponding qubits are initialized in the $Z$-basis, e.g., \blue{$Z^A$} indicates $Q(Z^A) \to Z^A$.
    Green node \darkgreen{$Q^{D}$} indicates the qubits in $Q^D$ are stabilized by the data CSS code $D$, while \darkgreen{$Q^A$} indicates that the qubits in $Q^A$ are initialized in the $Z$-basis.
    In all cases, there are corresponding unknown physical errors $e_{\bullet}(X),e_{\bullet}(Q),e_{\bullet}(Z)$(mod logicals).
    Equivalently, the state is stabilized by the complex with the fully drawn out diagram
    \begin{equation}
    \label{eq:logical-code-line3}
    \begin{tikzpicture}[baseline]
    \matrix(a)[matrix of math nodes, nodes in empty cells, nodes={minimum size=20pt},
    row sep=0.2em, column sep=1.5em,
    text height=1.25ex, text depth=0.25ex]
    {&\red{X} & \red{Q(X)}\\ \\
    &\darkgreen{Q^A} & \darkgreen{Z(Q^A)}\\
     \darkgreen{X'^D} &\darkgreen{Q^D} & \darkgreen{Z'^D}\\ \\
     \blue{Q(Z)} & \blue{Z}&\\};
    \path[->,red,font=\scriptsize]
    (a-1-2) edge node[above]{$\rm{id}$}  (a-1-3);
    \path[->,blue,font=\scriptsize]
    (a-6-1) edge node[above]{$\rm{id}$}  (a-6-2);
    \path[->,green!40!black,font=\scriptsize]
    (a-3-2) edge node[above]{$\rm{id}$}  (a-3-3)
    (a-4-1) edge node[above]{$\partial^D$}  (a-4-2)
    (a-4-2) edge node[above]{$\partial^D$}  (a-4-3);
    \draw[green!40!black]
    (a-3-2) edge node[]{$\oplus$}  (a-4-2)
    (a-3-3) edge node[]{$\oplus$}  (a-4-3);
    \end{tikzpicture}
    \end{equation}
    where $X'^D,Z'^D$ denote copies of $X^D,Z^D$ which were stabilized due to the input.

    After line 4, the system is stabilized (with some errors) by
    \begin{equation}
        \begin{tikzpicture}[baseline]
        \matrix(a)[matrix of math nodes, nodes in empty cells, nodes={minimum size=25pt},
        row sep=1.5em, column sep=1.5em,
        text height=1.25ex, text depth=0.25ex]
        { & \blue{X^{A}} & \darkgreen{Q^{A}} & \red{Z^{A}} \\
         \blue{X^{D}} & \darkgreen{Q^{D}} & \red{Z^{D}} &\\};
        \path[->,font=\scriptsize]
        (a-1-2) edge (a-1-3)
        (a-2-1) edge (a-2-2)
        (a-1-2) edge node[right]{$g$} (a-2-2)
        (a-1-3) edge (a-1-4)
        (a-2-2) edge (a-2-3)
        (a-1-3) edge node[right]{$g$} (a-2-3);
        \end{tikzpicture}
    \end{equation}
    where the colored arrows indicate that we have omitted other arrows so that overall sequence is a complex if the check qubits are also shown.
    Equivalently, the full diagram is given by (compare with \cite{yuan2026unified})
    \begin{equation}
    \label{eq:logical-code-line4}
    \begin{tikzpicture}[baseline]
    \matrix(a)[matrix of math nodes, nodes in empty cells, nodes={minimum size=20pt},
    row sep=0.15em, column sep=1.5em,
    text height=1.25ex, text depth=0.25ex]
    {&\red{X} & \red{Q(X)}\\ \\
    &\darkgreen{Q^A} & \darkgreen{Z(Q^A)}\\
     \darkgreen{X'^D} &\darkgreen{Q^D} & \darkgreen{Z'^D}\\ \\
     \blue{Q(Z)} & \blue{Z}&\\};
    \path[->,red,font=\scriptsize]
    (a-1-2) edge node[above]{$\rm{id}$}  (a-1-3);
    \path[->,blue,font=\scriptsize]
    (a-6-1) edge node[above]{$\rm{id}$}  (a-6-2);
    \path[->,green!40!black,font=\scriptsize]
    (a-3-2) edge node[above]{$\rm{id}$}  (a-3-3)
    (a-4-1) edge node[above]{$\partial^D$}  (a-4-2)
    (a-4-2) edge node[above]{$\partial^D$}  (a-4-3);
    \draw[green!40!black]
    (a-3-2) edge node[]{$\oplus$}  (a-4-2)
    (a-3-3) edge node[]{$\oplus$}  (a-4-3);
    \path[->,font=\scriptsize]
    (a-1-2) edge node[right]{$\partial$} (a-3-2)
    (a-1-3) edge node[right]{$\cdots$} (a-3-3)
    (a-4-2) edge node[right]{$\partial$} (a-6-2)
    (a-4-1) edge node[right]{$\cdots$} (a-6-1);
    \end{tikzpicture}
    \end{equation}
    where are vertical maps form a chain map.
    Note that as shown in Lemma \ref{lem:low-lvl-measurement}, the omitted arrows (which resemble $g$ in Eq. \eqref{eq:low-lvl-measurement}) are irrelevant once we measure out the check qubits.
    Note that CNOT maps control $c$ and target qubit Paulis as follows
    \begin{align}
        Z_c &\mapsto Z_c\\
        Z_t &\mapsto Z_c Z_t \\
        X_c &\mapsto X_c X_t \\
        X_t &\mapsto X_t
    \end{align}
    Therefore, the system maps physical errors on qubits given as follows:
    \begin{align}
        Q(X): e_Z(X) &\mapsto e_Z(X) +\delta e_Z(Q) \\
        Q: e_{\bullet} (Q) &\mapsto e_{\bullet}(Q) \\
        Q(Z): e_X(Z) &\mapsto e_X(Z) +\partial e_X(Q)
    \end{align}
    where
    \begin{align}
        e_{\bullet}(Q) = e_{\bullet}(Q^A) \oplus e_{\bullet}(Q^D)
    \end{align}
    where recall that we can take $e_Z(Q^A) =0$ since $Q^A$ is initialized in the $Z$-basis.

    By Lemma \ref{lem:low-lvl-measurement}, measuring out the check qubits $Q(X)=Q(X^A) \oplus Q(X^D)$ will output the readout
    \begin{equation}
        \delta \tilde{\ell}_Z(Q) + e_Z(X) +\delta e_Z(Q) \in X
    \end{equation}
    for some $Z$-type logical $\tilde{\ell}_Z(Q)$ of the green portion in Eq. \eqref{eq:logical-code-line3}.
    In particular, since ancillas \darkgreen{$Q^A$} are initialized in the $Z$-basis of Eq. \eqref{eq:logical-code-line3}, $\tilde{\ell}_Z(Q^A)$ is arbitrary in \darkgreen{$Q^A$}.
    Similarly, since \darkgreen{$Q^D$} are stabilized by the data code $D$ due to the input, we see that $\tilde{\ell}_Z(Q^D)$ is a logical of data code $D$.
    Let
    \begin{equation}
        \tilde{e}_Z(Q) = e_Z(Q) +\tilde{\ell}_Z(Q)
    \end{equation}
    Then we see that line 5 is correct.
    Specifically, the state after line 5 is stabilized by $X$ Paulis of the form $\im \partial$ on $Q$ with physical errors $\tilde{e}_Z(Q)$ (mod logicals), where we have ignored the stabilizers on the $Z$-check qubits $Q(Z)$, since they will be measured out in the next line.
    Similarly, the state after line 5 is stabilized by $Z$ Pauli of the form
    \begin{equation}
        (Q^A \oplus \im \delta^D) \cap (\im \partial)^\perp
    \end{equation}
    with physical errors $e_X(Q)$. Note that if $q^A \oplus \delta^D z^D \in (\im \partial)^\perp =\ker \delta$, then $\delta^A (q^A +g^{\top}z^D)=0$ and thus $q^A \in g^{\top} z^D +\im \delta^A$ where we utilized the fact that the ancilla has no internal logicals so that $\ker \delta^A = \im \delta^A$. Hence,
    \begin{equation}
        (Q^A \oplus \im \delta^D) \cap (\im \partial)^\perp = \im \delta
    \end{equation}

    By Lemma \ref{lem:low-lvl-measurement}, measuring out the check qubits $Q(Z)=Q(Z^A) \oplus Q(Z^D)$ will output the readout
    \begin{equation}
        \partial \ell_X(Q) + e_X(Z) +\partial e_X(Q) \in Z
    \end{equation}
    for some $X$-type logical $\ell_X(Q)$ of the stabilizer code after line 5. However, note that after line 5, the state is stabilize by $Z$-check of the form $\im \delta$, and thus $\partial \ell_X(Q)=0$ so that the readout is
    \begin{equation}
        e_X(Z) +\partial e_X(Q) \in Z
    \end{equation}
    and thus the state has physical $X$-error $e_X(Q) +\ell_X(Q)=e_X(Q)$ mod logicals so that line 6 is correct.
    More specifically, the state after line 6 is stabilized by the deformed code $\cone(g)$ with errors $\tilde{e}_Z(Q),e_X(Q)$ mod logicals.

    By Lemma \ref{lem:low-lvl-measurement}, we see that measuring ancilla qubits $Q^A$ transversally gives readout
    \begin{equation}
        s_X(Q^A) + e_X(Q^A) \in Q^A
    \end{equation}
    with $s_X(Q^A) \in \ell(Q^A)$ where $\ell(Q^A)$ is the projection of some $\ell(Q)\in \ker \partial$ onto $Q^A$. In particular, we see that $\ell(Q^A) \in \ker \partial^A$ and thus $s_X(Q^A) = \partial ^A x^A$ since the ancilla only has trivial logicals.
    Also note that $\partial^D \ell(Q^D) = g\ell(Q^A)=\partial^D gx^A$.
    The output is then stabilized by
    \begin{itemize}
        \item $X$-Paulis of the form $\im \delta^D$ and $\ell^D$ with error $\tilde{e}_Z(Q^D)$ mod logicals
        \item $Z$-Paulis of the form $\im \partial^D$ with error $e_X(Q^D) +\ell(Q^D)=e_X(Q^D) +gx^A$ mod logicals.
    \end{itemize}

    In particular, note that the measurement of $\ell^D$ is given by
    \begin{align}
        \bra \tilde{e}_Z(Q^D)|\ell^D\ket &= \bra \tilde{e}_Z(Q^D)| g\cX^A\ket \\
        &= \bra \delta \tilde{e}_Z(Q)|\cX^A\ket
    \end{align}
    where we utilized the fact that $\partial^A \cX^A =0$. Note that if no errors occurs, then we obtain readout $\delta \tilde{e}_Z(Q)$ and thus can determine the logical measurement exactly.
\end{proof}

\begin{algorithm}
\caption{\textsc{Logical Measurement}}
\label{alg:logical-measure}
\begin{algorithmic}[1]
    \Require An input state stabilized by data code $D$ with unknown physical $e_{\bullet}(Q^D)$.
    Data code gadget $D$. Ancilla code gadget $A$ with respect to $X$-type logical $\ell^{\star}_X$ and relative connectivity gadget $g$.
    \Ensure Readout
    \begin{align}
        \delta \tilde{e}_Z(Q) +e_Z(X) \in X \\
        \partial e_X(Q)+ e_X(Z) \in Z \\
        \partial^A x^A +e_X(Q^A)\in Q^A
    \end{align}
    where
    \begin{align}
        e_X(Z) & = e_X(Z^A) \oplus e_X(Z^D) \\
        e_X(Q) & = e_X(Q^A) \oplus e_X(Q^D) \\
        e_Z(X) & = e_Z(X^A) \oplus e_Z(X^D) \\
        \tilde{e}_Z(Q) &= \tilde{e}_Z(Q^A) \oplus \tilde{e}_Z(Q^D)
    \end{align}
    where non-tilde elements $e$ denote (presumably small) initialization errors, and
    \begin{equation}
        \delta^D \tilde{e}_Z(Q^D) =\delta^D e_Z(Q^D)
    \end{equation}
    The output state is stabilized by data code $D$ and $X$-type logical $\ell^{\star}_X$ with physical $\tilde{e}_Z(Q^D),e_X(Q^D)+gx^A$.
    In particular, $\ell^{\star}_X$ has measurement
    \begin{equation}
        \bra \tilde{e}_Z(Q)|\ell^{\star}_X\ket = \bra \delta \tilde{e}_Z(Q)| \cX^A\ket \in \F_2
    \end{equation}
    where $\cX^A$ is the sum over all $X$-checks (vertices) in $X^A$.
    \State Initialize $Q^{A}$ in the $Z$ basis with unknown error $e_{X}(Q^A)$, and $e_{Z}(Q^A)=0$
    \State Initialize $Q(X)=Q(X^A)\oplus Q(X^D)$ in the $X$ basis with unknown error
    \begin{equation}
        e_Z(X) = e_Z(X^A) \oplus e_Z(X^D)
    \end{equation}
    \State Initialize $Q(Z)=Q(Z^A) \oplus Q(Z^D)$ in the $Z$ basis with unknown error
    \begin{equation}
        e_X(Z) = e_X(Z^A) \oplus e_X(Z^D)
    \end{equation}
    \State Apply $O(1)$ depth CNOT circuit for $\partial^D,\partial^A,g$
    \State Measure out $Q(X)$ transversally in the $X$-basis with readout
    \begin{equation}
        \delta \tilde{e}_Z(Q) +e_Z(X) \in X
    \end{equation}
    for some $\tilde{e}_Z(Q)\in Q$
    such that $\delta^D \tilde{e}_Z(Q^D) = \delta^D e_Z(Q^D)$ so that the output state has physical error $\tilde{e}_Z(Q)$ mod logicals.
    \State Measure out $Q(Z)$ transversally in the $Z$-basis with readout
    \begin{equation}
        \partial e_X(Q)+ e_X(Z) \in Z
    \end{equation}
    so that the output state has physical error $e_X(Q)$ mod logicals.
    \State Measure out $Q^A$ transversally in the $Z$-basis (opposite of logical type $\ell^{\star}_X$) with readout
    \begin{equation}
        \partial^A x^A +e_X(Q^A)\in Q^A
    \end{equation}
    for some $x^A\in X^A$.
    so that the output state is stabilized by
    \begin{itemize}
        \item $X$-Paulis in $\im \partial^D + \ell^{\star}_X$ with physical $\tilde{e}_Z(Q^D)$ mod logicals (not only commute with $X$-Paulis of $D$ but also $\ell^{\star}_X$)
        \item $Z$-Paulis in $\im \delta^D$ with physical $e_X(Q^D) + gx^A$ mod logicals
    \end{itemize}
\end{algorithmic}
\end{algorithm}

\subsection{Low Level Non-CSS Measurement}
\label{sec:low-lvl-non-CSS-measure}

\begin{lemma}[Low Level]
\label{lem:low-lvl-non-CSS-measurement}
Algorithm \ref{alg:low-lvl-non-CSS-measurement} is correct.
\end{lemma}

\begin{proof}
    By Lemma \ref{lem:overcomplete-non-CSS-physicals}, the transversal measurement outcome $s(\bar{X}^A)\in Q(\bar{X}^A)$ (in the $Z$-sector of the ancilla qubits $P^A$) must satisfy
    \begin{equation}
        e(P) + s(\bar{X}^A) \in \ker \hat{\sigma} +\left(Q(X^{A}) \oplus P^D\right)
    \end{equation}
    Hence, there must exists $\ell =0\oplus \ell (P^D) \oplus \ell(\bar{X}^A)\in \ker \hat{\sigma}$ and $x^{A}\in X^{A}$ such that
    \begin{equation}
        \begin{pmatrix}
            e(X^{A}) \\
            e(P^D) \\
            e(\bar{X}^A)
        \end{pmatrix}
        +
        \underbrace{
        \begin{pmatrix}
            0 \\
            \ell(P^D) \\
            \ell(\bar{X}^A)
        \end{pmatrix}
        +
        \sigma x^A
        }_{\in  \ker\hat{\sigma}}
        =
        \begin{pmatrix}
            0 \\
            0 \\
            s
        \end{pmatrix}
        +
        \begin{pmatrix}
            Q(X^A) \\
            P^D \\
            0
        \end{pmatrix}
    \end{equation}
    Hence,
    \begin{equation}
        s=e(\bar{X}^A) +\ell(\bar{X}^A)
    \end{equation}
    where we utilized the fact that the defect map $p=0$.
    As we shall see in Algorithm \ref{alg:low-lvl-non-CSS-wCNOTs}, this is due to measurement of commuting Paulis so that $\hat{g}g=0$.
    Since $\ell$ is a logical, we see that $\ell(\bar{X}^A) = \hat{g} \ell(P^D)$ where $\ell(P^D)\in \ker \hat{\sigma}^D$ is a logical of the data code $D$ and thus the readout is
    \begin{equation}
        e(\bar{X}^A) + \hat{g} \ell(P^D)
    \end{equation}
    Again, by Lemma \ref{lem:overcomplete-non-CSS-physicals}, we see that the output state is stabilized with physical measurement $e(P)+\ell+\sigma x^{A}$, with stabilizer of the form
    \begin{align}
        &\im \sigma \cap (Q(X^A)\oplus P^D) +Q(X^{A}) \\
        &= Q(X^{A}) \oplus \underbrace{(\im \sigma^D \cap \ker \hat{g} + \im g)}_{P^D}
    \end{align}
    Since we discard ancilla qubits $P^A = Q(X^{A} ) \oplus Q(\bar{X}^A)$, we only care about the phase info in $\im \sigma^D \cap \ker \hat{g}$.
    In particular, we see that if $s^D\in S^D$ is such that $\sigma^D s^D \in \im \sigma^D \cap \ker \hat{g}$, then the output state is stabilized by the Pauli element corresponding to $\sigma^D s^D$ with measurement
    \begin{align}
        &\Lambda^D(e(P^D) +\ell(P^D) +gx^{A},\sigma^D s^D) \\
        &=\Lambda^D(e(P^D)+\ell(P^D),\sigma^D s^D) +(\hat{g} \sigma^D s^D)(x^{A})\\
        &= \Lambda^D(e(P^D)+\ell(P^D),\sigma^D s^D)
    \end{align}
    Similarly, note that
    \begin{equation}
        \Lambda^D (g x^{A} ,g\tilde{x}^A) =  0
    \end{equation}
    where we used the fact that $\hat{g} g=0$.
    Hence, the output state has physical measurement $e(P^D) +\ell(P^D)$ mod logicals.
\end{proof}

\begin{lemma}[High Level]
    Algorithm \ref{alg:high-lvl-non-CSS-measurement} is correct.
\end{lemma}

\begin{algorithm}
\caption{\textsc{Low-Level non-CSS Measurement}}
\label{alg:low-lvl-non-CSS-measurement}
\begin{algorithmic}[1]
    \Require State stabilized by symplectic complex (so that the diagram commutes and $\hat{g}g=0$) with stabilizer maps $\sigma$
    \begin{equation}
        \label{eq:low-lvl-non-CSS-measurement}
        \begin{tikzpicture}[baseline]
        \matrix(a)[matrix of math nodes, nodes in empty cells, nodes={minimum size=25pt},
        row sep=1.5em, column sep=1.5em,
        text height=1.25ex, text depth=0.25ex]
        {&X^{A} & Q(X^{A})\\
         S^D &P^D & \bar{S}^{D}\\
         Q(\bar{X}^A) & \bar{X}^A &\\};
        \path[->,font=\scriptsize]
        (a-1-2) edge node[above]{$\rm{id}$}  (a-1-3)
        (a-2-1) edge node[above]{$\sigma^{D}$}  (a-2-2)
        (a-2-2) edge node[above]{$\hat{\sigma}^D$}  (a-2-3)
        (a-1-2) edge node[right]{$g$}  (a-2-2)
        (a-1-3) edge node[right]{$\hat{\sigma}^D g$}  (a-2-3)
        (a-2-1) edge node[right]{$\hat{g}\sigma^D$} (a-3-1)
        (a-2-2) edge node[right]{$\hat{g}$} (a-3-2)
        (a-3-1) edge node[above]{$\mathrm{id}$} (a-3-2);
        \end{tikzpicture}
    \end{equation}
    with (unknown) physical
    \begin{equation}
        e(P)= e(P^D) \oplus e(\bar{X}^A)
    \end{equation}
    mod logicals.
    \Ensure Readout
    \begin{equation}
        \hat{g} \ell(P^D)+e(\bar{X}^A) \in \bar{X}^A
    \end{equation}
    for some logical $\ell(P^D) \in \ker \hat{\sigma}^D$ and output state stabilized by
    \begin{equation}
        \im \sigma^D \cap \ker \hat{g} +\im g
    \end{equation}
    with physical $e(P^D)+\ell(P^D)$ mod logicals.
    \State Measure out $P^A=Q(X^{A})\oplus Q(\bar{X}^A)$ (in the $X$-basis correspond to $Q(X^{A})$) with physical measurement (in the $Z$-sector) $\in Q(\bar{X}^A) \cong \bar{X}^A$
\end{algorithmic}
\end{algorithm}

\begin{algorithm}
\caption{\textsc{High-Level non-CSS Measurement}}
\label{alg:high-lvl-non-CSS-measurement}
\begin{algorithmic}[1]
    \Require State stabilized by non-CSS code with symplectic complex
    \begin{equation}
        D= S^D \xrightarrow{\sigma^D} P^D \xrightarrow{\hat{\sigma}^D} {\bar{S}^{D}}
    \end{equation}
    with (unknown) physical $e(P^D)$ mod logicals.
    Circuit to perform (commuting) measurement $g:X^{A} \to P^D$ so that $\hat{g}g=0$
    \Ensure Valid syndrome
    \begin{equation}
        \hat{g}(e(P^D) +\ell(P^D)) \in \bar{X}^A \cong X^A
    \end{equation}
    for some logical $\ell(P^D) \in \ker \hat{\sigma}^D$ and output stabilized by Paulis of the form
    \begin{equation}
        \im \sigma^D \cap \ker \hat{g} +\im g
    \end{equation}
    \State Perform (commuting) measurement $g:X^{A} \to P^D$ with valid syndrome $\in \bar{X}^A\cong X^{A}$
\end{algorithmic}
\end{algorithm}

\begin{algorithm}
\caption{\textsc{Low-Level non-CSS with CNOTs}}
\label{alg:low-lvl-non-CSS-wCNOTs}
\begin{algorithmic}[1]
    \Require State stabilized by non-CSS code with symplectic complex
    \begin{equation}
        D= S^D \xrightarrow{\sigma^D} P^D \xrightarrow{\hat{\sigma}^D} {\bar{S}^{D}}
    \end{equation}
    with (unknown) physical $e(P^D)$ mod logicals.
    \Ensure Measure commuting Paulis in $g:X^{A} \to P^D$ with readout
    \begin{equation}
        \hat{g}\tilde{e}(P^D) +e(\bar{X}^A) \in \bar{X}^A
    \end{equation}
    where $e(\bar{X}^A)$ is unknown, and $\tilde{e}(P^D)+e(P^D) \in \ker \hat{\sigma}^D$ is some unknown logical of $D$, and the output state is stabilized by Paulis of the form
    \begin{equation}
        \im \sigma^D \cap \ker \hat{g} +\im g
    \end{equation}
    with physical $\tilde{e}(P^D)$ mod logicals.
    \State Initialize ancilla qubits $P^A=Q(X^{A})\oplus Q(\bar{X}^A)$ transversally in $X$-basis with (unknown) physicals in the $Z$-sector $e(\bar{X}^A)\in Q(\bar{X}^A)\cong \bar{X}^A$  mod logicals
    \State Apply CNOT and CZ circuits from ancilla qubits to data qubits based on $g_X,g_Z$ (defined as the projection of $g$ onto the $X$- and $Z$-sector), respectively
    \State Measure out $P^A$ in the $X$-basis (correspond to $Q(X^{A})$) with (physical) measurement in the $Z$-sector $\in Q(\bar{X}^A) \cong \bar{X}^A$
\end{algorithmic}
\end{algorithm}

\begin{lemma}[Relation]
    \label{lem:low-lvl-non-CSS-wCNOTs}
    Algorithm \ref{alg:low-lvl-non-CSS-wCNOTs} is correct.
\end{lemma}
\begin{proof}
    After initialization (line 1), the state is stabilized by
    \begin{equation}
        \begin{tikzpicture}[baseline]
        \matrix(a)[matrix of math nodes, nodes in empty cells, nodes={minimum size=25pt},
        row sep=1.5em, column sep=1.5em,
        text height=1.25ex, text depth=0.25ex]
        {&X^{{A}} & Q(X^{{A}})\\
         S^D &P^D & \bar{S}^{D}\\
         Q(\bar{X}^A) & \bar{X}^A &\\};
        \path[->,font=\scriptsize]
        (a-1-2) edge node[above]{$\rm{id}$}  (a-1-3)
        (a-2-1) edge node[above]{$\sigma^{D}$}  (a-2-2)
        (a-2-2) edge node[above]{$\hat{\sigma}^D$}  (a-2-3)
        (a-3-1) edge node[above]{$\mathrm{id}$} (a-3-2);
        \end{tikzpicture}
    \end{equation}

    Note that CNOT maps control $c$ and target qubit Paulis as follows
    \begin{align}
        X_c &\mapsto X_c X_t \\
        X_t &\mapsto X_t \\
        Z_t &\mapsto Z_c Z_t \\
        Z_c &\mapsto Z_c
    \end{align}
    Hence, if a CNOT circuit is applied from ancilla qubits $Q(X^{A})\cong X^{A}$ to data qubits $Q^D\cong Q_X^D$ based on $g_X:X^{A}\to Q_X^D$, then the stabilizer map $\sigma$ of the height-2 cone
    \begin{equation}
        \mapsto
        \begin{pmatrix}
            \rm{id} &&& \\
            g_X & \rm{id} && \\
            && \rm{id} &\\
            && g_X^{\top} &\rm{id}
        \end{pmatrix} \sigma
    \end{equation}
    where the matrix is relative to
    \begin{equation}
        Q(X^{A})\oplus Q_X^D\oplus Q_Z^D\oplus Q(\bar{X}^A)
    \end{equation}
    Note that the matrix is block-diagonal with respect to the $X,Z$-decomposition, i.e., $X,Z$ Paulis are mapped to $X,Z$ Paulis, respectively.
    This corresponds to a diagram of the height-2 cone with $p=0$.

    Similarly, a CZ maps control $c$ and target qubit Paulis as follows
    \begin{align}
        X_c &\mapsto X_c Z_t \\
        X_t &\mapsto Z_c X_t \\
        Z_t &\mapsto Z_t \\
        Z_c &\mapsto Z_c
    \end{align}
    Hence, if a CZ is applied from ancilla qubits $Q(X^{A})\cong X^{A}$ to data qubits $Q^D\cong Q_Z^D$ based on $g_Z:X^{A} \to Q_Z^D$, then the stabilizer map $\sigma$ of the height-2 cone
    \begin{equation}
        \mapsto
        \begin{pmatrix}
            \rm{id} &&& \\
             & \rm{id} && \\
            g_Z && \rm{id} &\\
            &  g_Z^{\top} &&\rm{id}
        \end{pmatrix} \sigma
    \end{equation}
    where we note that the matrix representation is no longer block-diagonal, i.e., it mixes $X,Z$ Paulis.
    Note that the matrix is lower-triangular and thus corresponds to a similar height-2 cone diagram with $p=0$.

    In particular, note that the application of CNOT and CZ circuits results in matrix
    \begin{equation}
        \begin{pmatrix}
            \rm{id} &&& \\
            g_X & \rm{id} && \\
            g_Z && \rm{id} &\\
            &  g_Z^{\top} & g_X^{\top} &\rm{id}
        \end{pmatrix}
        =
        \begin{pmatrix}
            \rm{id} && \\
            g & \rm{id} & \\
            &  \hat{g} &\rm{id}
        \end{pmatrix}
    \end{equation}
    Hence, physical $e(P^D)\oplus e(\bar{X}^A)$ maps to
    \begin{equation}
        \mapsto e(P^D)\oplus (\hat{g} e(P^D) +e(\bar{X}^A))
    \end{equation}
    The statement then follows from Lemma \ref{lem:low-lvl-non-CSS-measurement}.
\end{proof}
\begin{remark}
    Note that the only difference between Algorithm \ref{alg:high-lvl-non-CSS-measurement} and \ref{alg:low-lvl-non-CSS-wCNOTs} is the measurement error as described in Remark \ref{rem:phys-meas-errs}.
\end{remark}


\subsubsection{Non-CSS Syndrome Extraction}
\label{sec:non-CSS-syndrome-extraction}
In this section, we shall apply the low level implementation of a measurement in the context of syndrome extraction of a data code.
At a high-level, we shall utilize the following corollary, which implies that if an input state has physical errors, those error can be detected via subsequential measurement without modifying the stabilizer code.

\begin{corollary}[Subsequent Measurement]
    Suppose that  $X^{A}$ is a copy of $S^D$ with $g = \sigma^D$.
    Then Algorithm \ref{alg:low-lvl-non-CSS-wCNOTs} has the same output but with $\tilde{e}(P^D)=e(P^D)$.
\end{corollary}

Let us further elaborate the syndrome extraction as follows.

\begin{definition}[Non-CSS Code Gadget]
    \label{def:non-CSS-code-gadget}
    Consider a non-CSS code $D=S^D\to P^D\to \bar{S}^D$ in the form of a symplectic complex with (basis) checks $\cS^D$ and ($X,Z$-sector) qubits $\cQ_X^D,\cQ_Z^D\cong \cQ^D$ and stabilizer map $\sigma^D$, and isometry $\phi^D:S^D\to \bar{S}^D$.
    The corresponding \textbf{code gadget} is the collection of
    \begin{itemize}
        \item qubits $\cQ$ and check qubits $\cQ(S^D)$ where $\cQ(S^D)$ is in one-to-one correspondence with checks $\cS^D$.
        For notation consistency, write $Q(S^D)$ to be the $\F_2$ space generated by $\cQ(S^D)$ so that the identity map is an isometry between $Q(S^D)\cong S^D$.
        \item CNOT circuits such that a CNOT from check qubits $Q(S^D)$ to data qubits $Q^D$ can be implemented in $O(1)$ circuit depth based on the $X$-sector generator $\sigma^D_X:Q(S^D)\cong S^D\to Q_X^D$ of $\sigma^D$.
        Similarly, a CZ from check qubits $Q(S^D)$ to data qubits $Q^D$ can be implemented in $O(1)$ depth based on the $Z$-sector generator $\sigma^D_Z:Q(S^D)\cong S^D\to Q_Z^D$.
    \end{itemize}
    We shall always assume that qubits can be initialized transversally relative to the $X$ or $Z$ basis with unknown physical errors (mod logicals), and can be measured transversally perfectly.
\end{definition}

The following two algorithms then addresses the initialization and subsequent measurement of a syndrome extraction process.
\begin{algorithm}
\caption{\textsc{non-CSS Syndrome Extraction}}
\label{alg:non-CSS-syndrome-extract}
\begin{algorithmic}[1]
    \Require Data code $D$ gadget
    \Ensure Readout
    \begin{align}
        \hat{\sigma}^D \tilde{e}(P^D) +e(\bar{S}^D) \in \bar{S}^{D}
    \end{align}
    for some $\tilde{e}(P^D),e(\bar{S}^D)$, and output state stabilized by $D$ with physical $\tilde{e}(P^D)$
    \State Initialize $Q(S^D)$ in the $X$ basis, with unknown physical $e(\bar{S}^D)$ in the $Z$-sector
    \State Apply $O(1)$ depth CNOT and CZ circuit for $\sigma^D:S^D\to P^D$
    \State Measure out $Q(S^D)$ transversally in the $X$-basis with readout
    \begin{equation}
        \hat{\sigma}^D \tilde{e}(P^D) +e(\bar{S}^D)
    \end{equation}
    for some $\tilde{e}(P^D)$ so that the output state has physical measurement $\tilde{e}(P^D)$ mod logicals
\end{algorithmic}
\end{algorithm}


\begin{lemma}[non-CSS Syndrome Extraction]
    \label{lem:non-CSS-syndrome-extraction}
    Algorithm \ref{alg:non-CSS-syndrome-extract} is correct.
    Moreover, if an input state stabilized by $D$ with physical $e(P^D)$ (e.g., subsequent measurement) is given, then the output is as given by Algorithm \ref{alg:non-CSS-syndrome-extract} with $\tilde{e}(P^D) = e(P^D)$.
\end{lemma}
\begin{remark}[Fault Tolerance with Threshold]
    The physical errors can then be corrected via the general error correction procedure \cite{gottesman2013fault}.
    Specifically, after $T\ge d(D)$ rounds of subsequent syndrome extractions of Algorithm \ref{alg:non-CSS-syndrome-extract} where $d(D)$ is the code distance of $D$, one can build a spacetime fault complex.
    A threshold can then be obtained by applying min-weight correction relative to the fault complex.
    See Section \ref{sec:non-CSS-syndrome-fault-complex} for details.
\end{remark}
\begin{proof}
    The proofs follows from Lemma \ref{lem:low-lvl-non-CSS-wCNOTs}.
\end{proof}


\subsubsection{Non-CSS Logical Measurement}
\label{sec:non-CSS-logical-measure}
In this section, we shall consider the low level implementation of a logical measurement on a data code
\begin{equation}
    D=S^D \to P^D\to \bar{S}^D
\end{equation}
Conventionally, an ancilla $A$ is constructed which is then \textit{merged} with the data code $D$ to form a deformed code.
The checks of the deformed code are then measured in a manner similar to syndrome extraction in Section \ref{sec:non-CSS-syndrome-extraction}, so that the logical measurement can be extracted using the general error correction procedure \cite{gottesman2013fault} after $d(D)$ rounds of subsequential measurement.



\begin{theorem}[Non-CSS Logical Measurement]
    \label{thm:non-CSS-logical-measure}
    Algorithm \ref{alg:non-CSS-logical-measure} is correct.
    In particular, if no errors occur, then Algorithm \ref{alg:non-CSS-logical-measure} provides readout
    \begin{align}
        \hat{g}\tilde{e}(P^D) \in \bar{X}^A \\
        \delta^A x^{A} \in Q^{A}_X
    \end{align}
    so that the logical $\ell^{\star}(P^D)$ has measurement
    \begin{equation}
        \bra \hat{g}\tilde{e}(P^D)|\bar{\cX}^{A}\ket +\bra \delta^A x^{A} |p\cX^{A}\ket\in \F_2
    \end{equation}
    where $\cX^{A}\in X^A$ is the sum over all checks (vertices) in $X^{A}$ and $\bar{\cX}^A\in \bar{X}^A$ is the corresponding syndrome.
    Moreover, if the input state is also stabilized by logical $\ell^{\star}(P^D)$ (e.g., subsequent measurement), then we replace $\tilde{e}(P)\to e(P) =e(Q_X^A)\oplus e(P^D)$ in the output.
\end{theorem}

\begin{proof}
    For simplicity, we shall only consider the case where the input state is only stabilized by the data code (with some errors).
    The proof for when the input state is also stabilized by the intended logical follows similarly, and thus omitted.
    After initialization in line 1, the state is stabilized by the complex with the fully drawn out diagram
    \begin{equation}
    \label{eq:non-CSS-logical-code-line1}
    \begin{tikzpicture}[baseline]
    \matrix(a)[matrix of math nodes, nodes in empty cells, nodes={minimum size=20pt},
    row sep=0.2em, column sep=1.5em,
    text height=1.25ex, text depth=0.25ex]
    {&\red{S} & \red{Q(S)}\\ \\
    &\darkgreen{Q_X^A} & \darkgreen{\bar{Z}(Q^A)}\\
     \darkgreen{S'^D} &\darkgreen{P^D} & \darkgreen{\bar{S}'^D}\\
     \darkgreen{Z(Q^A)} & \darkgreen{Q_Z^A} &\\ \\
     \red{Q(\bar{S})} & \red{\bar{S}}&\\};
    \path[->,red,font=\scriptsize]
    (a-1-2) edge node[above]{$\rm{id}$}  (a-1-3)
    (a-7-1) edge node[above]{$\rm{id}$}  (a-7-2);
    \path[->,green!40!black,font=\scriptsize]
    (a-3-2) edge node[above]{$\rm{id}$}  (a-3-3)
    (a-4-1) edge node[above]{$\sigma^D$}  (a-4-2)
    (a-4-2) edge node[above]{$\hat{\sigma}^D$}  (a-4-3)
    (a-5-1) edge node[above]{$\rm{id}$}  (a-5-2);
    \end{tikzpicture}
    \end{equation}
    where $S'^D,\bar{S}'^D$ denote copies of $S^D,\bar{S}^D$ which were stabilized due to the input, and $Q(S)\oplus Q(\bar{S})$ denote the ($X,Z$-sectors) of check qubits where $S=X^A \oplus S^D \oplus Z^A$.
    The corresponding physicals are given by (mod logicals)
    \begin{equation}
        \begin{pmatrix}
            0 \\
            e(P)\\
            e(\bar{S})
        \end{pmatrix}, \quad
        e(P) =
        \begin{pmatrix}
            e(Q_X^A) \\
            e(P^D) \\
            0
        \end{pmatrix}
    \end{equation}
    Since the check qubits $Q(S)$ will ultimately be measured out, we rewrite the diagram by hiding the irrelevant info
    \begin{equation}
        \begin{tikzpicture}[baseline]
        \matrix(a)[matrix of math nodes, nodes in empty cells, nodes={minimum size=25pt},
        row sep=1.5em, column sep=1.5em,
        text height=1.25ex, text depth=0.25ex]
        { && \red{X^{A}} & \darkgreen{Q^{A}_X} & \red{\bar{Z}^{A}} \\
         &\red{S^{D}} & \darkgreen{P^{D}} & \red{\bar{S}^{D}} & \\
         \red{Z^A} & \darkgreen{Q_Z^A} & \red{\bar{X}^A}\\};
        \end{tikzpicture}
    \end{equation}

    After performing the CNOT and CZ circuits in line 2, the system is stabilized (with some errors) by
    \begin{equation}
    \begin{tikzpicture}[baseline]
    \matrix(a)[matrix of math nodes, nodes in empty cells, nodes={minimum size=25pt},
    row sep=2em, column sep=2em,
    text height=1.25ex, text depth=0.25ex]
    { && \red{X^{A}} & \darkgreen{Q^{A}_X} & \blue{\bar{Z}^{A}} \\
         &\red{S^{D}} & \darkgreen{P^{D}} & \blue{\bar{S}^{D}} & \\
         \red{Z^A} & \darkgreen{Q_Z^A} & \blue{\bar{X}^A}\\};
    \path[->,font=\scriptsize]
    (a-1-3) edge (a-1-4)
    (a-1-4) edge (a-1-5)
    (a-2-2) edge (a-2-3)
    (a-2-3) edge (a-2-4)
    (a-3-1) edge (a-3-2)
    (a-3-2) edge (a-3-3);
    \path[->,font=\scriptsize]
    (a-1-3) edge (a-2-3)
    (a-1-4) edge (a-2-4)
    (a-2-2) edge (a-3-2)
    (a-2-3) edge (a-3-3);
    \path[->,dashed,font=\scriptsize]
    (a-1-3) edge[bend right=80] (a-3-2)
    (a-1-4) edge[bend left=80] (a-3-3);
    \end{tikzpicture}
    \end{equation}
    where the colored nodes indicate that we still have omitted other arrows so that overall sequence is a complex if the check qubits are also shown.
    Equivalently, the full diagram is given by
    \begin{equation}
    \label{eq:non-CSS-logical-code-line3}
    \begin{tikzpicture}[baseline]
    \matrix(a)[matrix of math nodes, nodes in empty cells, nodes={minimum size=20pt},
    row sep=0.2em, column sep=1.5em,
    text height=1.25ex, text depth=0.25ex]
    {&\red{S} & \red{Q(S)}\\ \\
    &\darkgreen{Q_X^A} & \darkgreen{\bar{Z}(Q^A)}\\
     \darkgreen{S'^D} &\darkgreen{P^D} & \darkgreen{\bar{S}'^D}\\
     \darkgreen{Z(Q^A)} & \darkgreen{Q_Z^A} &\\ \\
     \red{Q(\bar{S})} & \red{\bar{S}}&\\};
    \path[->,red,font=\scriptsize]
    (a-1-2) edge node[above]{$\rm{id}$}  (a-1-3)
    (a-7-1) edge node[above]{$\rm{id}$}  (a-7-2);
    \path[->,green!40!black,font=\scriptsize]
    (a-3-2) edge node[above]{$\rm{id}$}  (a-3-3)
    (a-4-1) edge node[above]{$\sigma^D$}  (a-4-2)
    (a-4-2) edge node[above]{$\hat{\sigma}^D$}  (a-4-3)
    (a-5-1) edge node[above]{$\rm{id}$}  (a-5-2);
    \path[->,font=\scriptsize]
    (a-1-2) edge node[right]{$\sigma$} (a-3-2)
    (a-1-3) edge (a-3-3)
    (a-5-2) edge node[right]{$\hat{\sigma}$} (a-7-2)
    (a-5-1) edge (a-7-1);
    \end{tikzpicture}
    \end{equation}
    Note that due to the CNOT and CZ circuits, physicals propagate so that the corresponding physical maps to
    \begin{equation}
        \begin{pmatrix}
            \rm{id} && \\
            \sigma & \rm{id} & \\
            &  \hat{\sigma} &\rm{id}
        \end{pmatrix}
        \begin{pmatrix}
            0 \\
            e(P) \\
            e(\bar{S})
        \end{pmatrix}
        =
        \begin{pmatrix}
            0 \\
            e(P) \\
            \hat{\sigma} e(P) +e(\bar{S})
        \end{pmatrix}
    \end{equation}

    By Lemma \ref{lem:low-lvl-non-CSS-measurement}, after line 4, we obtain readout
    \begin{equation}
        \hat{\sigma} e(P) +e(\bar{S}) +\hat{\sigma} \tilde{\ell}(P) \in \bar{S}
    \end{equation}
    where $\tilde{\ell}(P)$ is some logical of only the green part in Diagram \eqref{eq:non-CSS-logical-code-line3}. Hence, $\tilde{\ell}(P)$ must be of the form
    \begin{equation}
        \begin{pmatrix}
            0 \\
            \ker \hat{\sigma}^D\\
            Q_Z^A
        \end{pmatrix}
    \end{equation}
    Indeed, the ancilla qubits are initialized in the $Z$-basis and thus $\tilde{\ell}(P)$ cannot act as $X$-Paulis on the ancilla.
    Also note that the output state after line 4 is stabilized by (with errors $e(P)+\tilde{\ell}(P)$)
    \begin{equation}
        \begin{pmatrix}
            0 \\
            \im \sigma^D\\
            Q_Z^A
        \end{pmatrix} \cap \ker \hat{\sigma} + \im \sigma
    \end{equation}
    Note that
    \begin{align}
        \hat{\sigma}
        \begin{pmatrix}
            0 \\
            \sigma^D s^D\\
            q_Z^A
        \end{pmatrix} &=0 \\
        \partial^A (\hat{h} s^D +q_Z^A) &=0
    \end{align}
    and thus by the fact that the ancilla has no nontrivial logicals, i.e., $\ker \partial^A =\im \partial^A$, there exists $z^A \in Z^A$ such that $q_Z^A = \hat{h}s^D +\partial^A z^A$ so that
    \begin{equation}
        \begin{pmatrix}
            0 \\
            \im \sigma^D\\
            Q_Z^A
        \end{pmatrix} \cap \ker \hat{\sigma} =
        \sigma
        \begin{pmatrix}
            0 \\
            S^D \\
            Z^A
        \end{pmatrix}
    \end{equation}
    Hence, the output state after line 4 is stabilized by the deformed code $\im\sigma$.
    Note that any Pauli of the form
    \begin{equation}
        \begin{pmatrix}
            0 \\
            \ker \hat{\sigma}^D \\
            Q_Z^A
        \end{pmatrix}
    \end{equation}
    is a logical of the code after line 4, i.e., commutes with $\sigma(S^D\oplus Z^A)$, and thus we can say that the output state after line 4 has physical $e(P)$ (mod logicals) instead of $e(P)+\tilde{\ell}(P)$ (mod logicals).

    Finally, by Lemma \ref{lem:overcomplete-non-CSS-physicals}, line 5 gives us readout in the $X$-sector $s(Q^A_X)\in Q^A_X$ such that
    \begin{equation}
        \begin{pmatrix}
            s(Q^A_X) \\
            0 \\
            0
        \end{pmatrix}
        +\tilde{e}(P) \in \ker \hat{\sigma} +
        \begin{pmatrix}
            0 \\
            P^D \\
            Q_Z^A
        \end{pmatrix}
    \end{equation}
    Note that any element in $\ker \hat{\sigma}$ (say by the Cleaning Lemma) can be written as
    \begin{equation}
        \underbrace{
        \begin{pmatrix}
            0\\
            \ell(P^D)\\
            \ell(Q_Z^A)
        \end{pmatrix}
        }_{\ell(P)} + \sigma x^A
    \end{equation}
    where $x^A\in X^A$ and $\ell(P)\in \ker \hat{\sigma}$. Therefore, we see that the readout satisfies
    \begin{equation}
        s(Q_X^A) = e(Q_X^A) +\delta^A x^A
    \end{equation}
    Moreover, the output state is stabilized by Paulis of the form
    \begin{equation}
        \im \sigma \cap
        \begin{pmatrix}
            0 \\
            P^D \\
            Q_Z^A
        \end{pmatrix}
        +
        \begin{pmatrix}
            0 \\
            0 \\
            Q_Z^A
        \end{pmatrix}
        =
        \begin{pmatrix}
            0 \\
            \im \sigma^D +g\ker \delta^A \\
            Q_Z^A
        \end{pmatrix}
    \end{equation}
    with physical measurement given by (mod logicals)
    \begin{equation}
        \tilde{e}(P)+\ell(P)+ \sigma x^A
    \end{equation}
    Note that $\ell(P)\in \ker \hat{\sigma}$ is also a logical of the final output code and thus can be omitted.
    Moreover, any element in $Q_Z^A$ is a logical of the final output code and thus can be omitted.
    Hence, the final state has physical measures (mod logicals)
    \begin{equation}
        \begin{pmatrix}
            e(Q_X^A) \\
            \tilde{e}(P^D)+gx^A \\
            0
        \end{pmatrix}
    \end{equation}
    Note that we discard the ancilla qubits and thus we are left with physical measuremenet $\tilde{e}(P^D) +gx^A$ relative to stabilizers in $\im \sigma^D +g\ker \delta^A =\im \sigma^D +\ell^{\star}(P^D)$.

\end{proof}

\begin{algorithm}
\caption{\textsc{Non-CSS Logical Measurement}}
\label{alg:non-CSS-logical-measure}
\begin{algorithmic}[1]
    \Require An input state stabilized by data code $D$ with unknown physical errors $e(P^D)$ mod logicals. Data code gadget $D$. Ancilla code gadget $A$ with respect to logical $\ell^{\star}(P^D)$ and relative connectivity gadget $g$.
    \Ensure Readout
    \begin{align}
        \hat{\sigma} \tilde{e}(P) +e(\bar{S}) \in \bar{S} \\
        \delta^A x^{A} +e(Q^{A}_X)\in Q^{A}_X
    \end{align}
    for some
    \begin{align}
        \tilde{e}(P)=e(Q^{A}_X)\oplus \tilde{e}(P^D)\oplus \tilde{e}(Q^A_Z) \\
        e(\bar{S})=e(\bar{Z}^{A})\oplus e(\bar{S}^D) \oplus e(\bar{X}^A)
    \end{align}
    and
    \begin{equation}
        \hat{\sigma}^D \tilde{e}(P^D) =\hat{\sigma}^D e(P^D)
    \end{equation}
    and output state stabilized by data code $D$ and logical $\ell^\star(P^D)$ with physical $\tilde{e}(P^D)+gx^{A}$ mod logicals.
    In particular, $\ell^\star(P^D)$ has measurement
    \begin{align}
        &\Lambda^D( \tilde{e}(P^D)+gx^{A},\ell^\star(P^D)) \\
        &= \bra \hat{\sigma}\tilde{e}(P)|\bar{\cX}^{A}\ket +\bra \delta^A x^{A} +e(Q_X^A) |p\cX^{A}\ket\in \F_2
    \end{align}
    where $\cX^{A}\in X^A$ is the sum over all checks (vertices) in $X^{A}$, and $\bar{\cX}^A\in \bar{X}^A$ is the corresponding syndrome.
    \State Initialize $P^{A}=Q^{A}_X\oplus Q^A_Z$ in the $Z$ basis with unknown error (in the $X$-sector) $e(Q^{A}_X)$
    \State Initialize qubits $Q(S)=Q(X^{A})\oplus Q(S^D) \oplus Q(Z^A)$ in the $X$ basis with unknown error $e(\bar{S})$ in the $Z$-sector
    \State Apply $O(1)$ depth CNOT and CZ circuit for $\sigma:S\to P$
    \State Measure out $Q(S)$ transversally in the $X$-basis with readout
    \begin{equation}
        \hat{\sigma} \tilde{e}(P) +e(\bar{S}) \in \bar{S}
    \end{equation}
    for some $\tilde{e}(P)\in P$ such that
    \begin{align}
        \tilde{e}(Q_X^A)&= e(Q_X^A) \\
        \tilde{e}(P^D) + e(P^D)&\in \ker \hat{\sigma}^D
    \end{align}
    so that the output state has physical error $\tilde{e}(P)$ mod logicals.
    \State Measure out $P^A$ transversally in the $Z$-basis with readout
    \begin{equation}
        \delta^A x^{A} +e(Q_X^{A})\in Q_X^{A}
    \end{equation}
    for some $x^{A}\in X^{A}$,
    so that the output state is stabilized by Paulis in $\im \sigma^D +g\ker \delta^A$ with physical measurement $\tilde{e}(P^D)+gx^{A}$ mod logicals (not only commute with Paulis of $D$ but also $\ell^{\star}(P^D)=g\ker\delta^A$)
\end{algorithmic}
\end{algorithm}

\section{Fast Non-CSS Surgery}
\label{sec:fast-surgery}

In this section, we generalize the framework for fast surgery \cite{baspin2025fast} to the non-CSS scenario using Theorem \ref{thm:symplectic-embedding} in the main text.
Specifically, following the notation of Theorem \ref{thm:symplectic-embedding}, consider a non-CSS code in the form of a symplectic complex
$D$ with the usual objects.
Suppose that there exists an ancilla CSS code in the form of a (co)complex ($A^\top$) $A$ with only trivial logicals, and maps $(g,h,p)$ such that the symplectic height-2 cone $C$ in Theorem \ref{thm:symplectic-embedding} is well-defined.

In contrast to conventional surgery in Section \ref{sec:surgery} of the main text, we must further assume that the ancilla code $A$ has meta-syndromes as follows
\begin{equation}
    A = Z^A \to Q_Z^A \to \bar{X}^A \to \bar{M}^A
\end{equation}
Which we grade as
\begin{equation}
    A = A_2 \to A_1 \to A_0 \to A_{-1}
\end{equation}
where we extend the (co)complex and treat the meta-syndrome in $\bar{M}^A$ as $(-1)$-(co)chains so that $d_0(A)$ is the meta-syndrome distance.
Also assume that the extended (co)complex possesses a map $f$ such that $(f,g,h)$ forms a chain map as follows.
\begin{equation}
    \begin{tikzcd}
    	& {M^A} & {X^A} & {Q_X^A} & {\bar{Z}^A} \\
    	& {S^D} & {P^D} & {\bar{S}^D} \\
    	{Z^A} & {Q_Z^A} & {\bar{X}^A} & {\bar{M}^A}
    	\arrow[from=1-2, to=1-3]
    	\arrow["f"', from=1-2, to=2-2]
    	\arrow["p"'{pos=0.4}, color={rgb,255:red,92;green,214;blue,92}, curve={height=6pt}, dashed, from=1-2, to=3-1]
    	\arrow[from=1-3, to=1-4]
    	\arrow["g"', from=1-3, to=2-3]
    	\arrow["p"'{pos=0.4}, color={rgb,255:red,92;green,214;blue,92}, curve={height=6pt}, dashed, from=1-3, to=3-2]
    	\arrow[from=1-4, to=1-5]
    	\arrow["h"', from=1-4, to=2-4]
    	\arrow["{p^\top}"{pos=0.7}, color={rgb,255:red,92;green,214;blue,92}, curve={height=6pt}, dashed, from=1-4, to=3-3]
    	\arrow["{p^\top}"{pos=0.7}, color={rgb,255:red,92;green,214;blue,92}, curve={height=6pt}, dashed, from=1-5, to=3-4]
    	\arrow[from=2-2, to=2-3]
    	\arrow["{\hat{f}}"', from=2-2, to=3-2]
    	\arrow[from=2-3, to=2-4]
    	\arrow["{\hat{g}}"', from=2-3, to=3-3]
    	\arrow["{\hat{h}}"', from=2-4, to=3-4]
    	\arrow[from=3-1, to=3-2]
    	\arrow[from=3-2, to=3-3]
    	\arrow[from=3-3, to=3-4]
    \end{tikzcd}
\end{equation}
where $M^A$ denotes the space dual to $\bar{M}^A$, so that the top row is the extended cocomplex $A^{\top}$.

Much like in \cite{baspin2025fast}, we will also make use of systolic expansion on the ancilla complex $A$.

\begin{definition}[Relative Expansion]
    \label{def:relative-expansion}
    The ancilla $A$ is said to be $(\epsilon,l)$-\textbf{expanding} relative to the code $D$ if
    \begin{equation}
        \forall f \in X^A: \quad |\delta^A f| \ \geq \ \min\left\{ \epsilon  \cdot \min_{z \in \ker \delta^A} |g (f+z)|, \ l \right\}
    \end{equation}
    where $|\cdot|$ is the Hamming weight of Definition \ref{def:complex-basis}.
\end{definition}

\begin{theorem}[Fast Surgery]
    \label{thm:fast-surgery}
    Take $D$ a code, and $A$ an ancilla that is $(\epsilon, d_1(D))$-expanding relative to $D$ (Definition \ref{def:relative-expansion}). As long as
    \begin{align}
        |p^{\top} e(Q_X^A)| +|\hat{g} e(P^D)| + |e(\bar{X}^A)| < d_0(A)/2 \\
        |e(Q_X^A)| < \min \{ d_1(A)/2,  \epsilon d_1(D)/4, d_1(D)/2 \}
    \end{align}
    then Algorithm \ref{alg:fast-surgery} measures $g \ker \delta^A \subset \ker \hat{\sigma}^D$, such that:
    \begin{enumerate}
        \item there exists a correction $\hat{e}(\bar{X}^A)$ such that
    \begin{equation}
        \bra  \hat{\sigma} \tilde{e}(P) +e(\bar{S}) + \hat{e}(\bar{X}^A) |\ell^\star(X^A)\ket = \bra \hat{g}\ell(P^D)| \ell^\star(X^A)\ket
    \end{equation}
        \item there exists $\hat{e}(P^D)$ and $u^A\in X^A$ with $|g(u^A)|<d_1(D)/2$ such that
        \begin{align}
            \mathfrak{e}_{\mathrm{phys}} + \hat{e}(P^D) \in  \ell(P^D) + g(u^A) + g \ker \delta^A
        \end{align}
        for $\mathfrak{e}_{\mathrm{phys}} = e(P^D)+ \ell(P^D) + gx^A$ the physical error of the system at the output of Algorithm \ref{alg:fast-surgery}.
    \end{enumerate}

\end{theorem}
\begin{proof}
    After applying Algorithm \ref{alg:fast-surgery}, we obtain the following readout from the measurement outcomes:
    \begin{align}
        \hat{\sigma} \tilde{e}(P) +e(\bar{S}) \in \bar{S}, &\quad \delta^A x^{A} +e(Q^{A}_X)\in Q^{A}_X \\
        \tilde{e}(P)=e(Q^{A}_X)\oplus \tilde{e}(P^D)\oplus \tilde{e}(Q^A_Z), &\quad
        e(\bar{S})=e(\bar{Z}^{A})\oplus e(\bar{S}^D) \oplus e(\bar{X}^A)
    \end{align}

    The measured value $m_{\text{raw}}$ of $\ell^\star(X^A) \in \ker \delta^A$ is then
    \begin{align*}
        m_{\text{raw}} = &\bra  \hat{\sigma} \tilde{e}(P) +e(\bar{S}) |\ell^\star(X^A)\ket \\ =&  \bra  p^{\top} e(Q_X^A) +\hat{g} \tilde{e}(P^D) + \partial^A \tilde{e}(Q_Z^A) + e(\bar{X}^A)|\ell^\star(X^A)\ket \\
        =& \bra \hat{g} \ell(P^D) + \partial^A \tilde{e}(Q_Z^A)|\ell^\star(X^A)\ket + \bra  p^{\top} e(Q_X^A) +\hat{g} e(P^D) + e(\bar{X}^A)|\ell^\star(X^A)\ket\\
        =&  \bra \hat{g} \ell(P^D)|\ell^\star(X^A)\ket + \bra  p^{\top} e(Q_X^A) +\hat{g} e(P^D) + e(\bar{X}^A)|\ell^\star(X^A)\ket
    \end{align*}
    The last equality uses that $\ell^\star(X^A)\in \ker(\delta^A)$, and hence $\bra  \partial^A \tilde{e}(Q_Z^A)|\ell^\star(X^A)\ket = 0$; i.e. the quantity $\bra \hat{g} \ell(P^D) + \partial^A \tilde{e}(Q_Z^A)|\ell^\star(X^A)\ket$ depends uniquely on the (co)homological classes of $\ell, \ell^\star$ which is as expected for a logical measurement.
    Further, note that in the absence of errors (or $e(Q_X^A),e(P^D)=0$) the outcome is entirely determined by this quantity, as we have $\bra  p^{\top} e(Q_X^A) +\hat{g} e(P^D) + e(\bar{X}^A)|\ell^\star(X^A)\ket = 0$, which yields the correct measurement $m_{\text{raw}}  = \bra \hat{g} \ell(P^D)|\ell^\star(X^A)\ket$.

    Otherwise, we can compute the meta-check syndrome:
    \begin{align*}
        & \partial^A (\hat{\sigma} \tilde{e}(P) +e(\bar{S}) ) \in \bar{M}^A\\
        =&\partial^A (p^{\top} e(Q_X^A) +\hat{g} \tilde{e}(P^D) + \partial^A \tilde{e}(Q_Z^A) + e(\bar{X}^A)) \\
        =& \partial^A (p^{\top} e(Q_X^A) +\hat{g} e(P^D) + e(\bar{X}^A))
    \end{align*}

    This syndrome is then leveraged into a correction $\hat{e}(\bar{X}^A)$, such that as long as $|p^{\top} e(Q_X^A) +\hat{g} e(P^D) + e(\bar{X}^A)| < d_0(A)/2$, we have $\hat{e}(\bar{X}^A) + p^{\top} e(Q_X^A) +\hat{g} e(P^D) + e(\bar{X}^A) \in \im \partial^A$. This gives us a rectified $m_{\text{corrected}}$:
    \begin{align}
        m_{\text{corrected}} = &\bra  \hat{\sigma} \tilde{e}(P) +e(\bar{S}) + \hat{e}(\bar{X}^A) |\ell^\star(X^A)\ket \\
        = &  \bra \hat{g} \ell(P^D)|\ell^\star(X^A)\ket + \bra  p^{\top} e(Q_X^A) +\hat{g} e(P^D) + e(\bar{X}^A) + \hat{e}(\bar{X}^A)|\ell^\star(X^A)\ket \\
        = & \bra \hat{g} \ell(P^D)|\ell^\star(X^A)\ket
    \end{align}
    Where the last equality is from $p^{\top} e(Q_X^A) +\hat{g} e(P^D) + e(\bar{X}^A) + \hat{e}(\bar{X}^A) \in \im \partial^A$.
    Meanwhile the physical error on the state at the output of Algorithm \ref{alg:fast-surgery} is:
    \begin{equation}
        \mathfrak{e}_{\mathrm{phys}} = \tilde{e}(P^D) + gx^A = e(P^D) + \ell(P^D) + gx^A
    \end{equation}
    for some $\ell(P^D) \in \ker \hat{\sigma}^D$. Our main problem here is we are not assuming that there already exists a decoder for the code complex like in Lemma \ref{lem:recovery-logical-msmt}, so we cannot guarantee something of the form $\mathfrak{e}_{\mathrm{phys}} + \hat{\mathfrak{e}}(P^D) = \sigma^D s^D + g s^A + g \mathfrak{x}^A$. Instead we need to wrestle with the dissonance between the physical system being affected by $gx^A$, while the recorded $\bar{S}^D$ syndrome:
    \begin{equation}
        \mathfrak{s}_{\mathrm{rec}} = \hat{\sigma}^D(e(P^D) + \ell(P^D)) + h(e(Q_X^A)) = \hat{\sigma}^D e(P^D)  + h(e(Q_X^A))
    \end{equation}
    is affected by $e(Q_X^A)$ instead.
    From assumptions, $|e(Q_X^A)| < d_1(A)/2$, so we can find $\hat{e}(Q_X^A)$ and some $u^A \in X^A$ with $\hat{e}(Q_X^A) + e(Q_X^A) = \delta^A u^A$, from the syndrome $\delta^A (\delta^A x^A + e(Q_X^A)) = \delta^A e(Q_X^A)$ that we compute from the destructive readout of $Q_X^A$. We update the recorded syndrome as follows:
    \begin{equation}
        \mathfrak{s}_{\mathrm{clean}} = \mathfrak{s}_{\mathrm{rec}} + h(\hat{e}(Q_X^A)) = \hat{\sigma}^D e(P^D) + h(\delta^A u^A) = \hat{\sigma}^D e(P^D) + \hat{\sigma}^D g (u^A).
    \end{equation}
    In a second step, pick $\hat{s}^A \in X^A$ such that $\delta^A \hat{s}^A = \delta^A (x^A + u^A)$. We also update the physical system by applying $g \hat{s}^A$, which gives
    \begin{equation}
        \mathfrak{e}_{\mathrm{clean}} = \mathfrak{e}_{\mathrm{phys}} + g\hat{s}^A \in e(P^D) + \ell(P^D) + g(u^A) + g \ker \delta^A
    \end{equation}
    This cleaning allows us to resolve the dissonance as we get $\hat{\sigma}^D \mathfrak{e}_{\mathrm{clean}} = \mathfrak{s}_{\mathrm{clean}}$. From here we use the relative expansion of Definition \ref{def:relative-expansion}:
    \begin{equation}
        |g(u^A)| \leq \frac{1}{\epsilon}|\delta^A u^A| \leq \frac{2}{\epsilon} |e(Q_X^A)| < d_1(D)/2
    \end{equation}
    Hence the leftover error is correctable.
\end{proof}

\begin{algorithm}
\caption{\textsc{Fast Non-CSS Surgery}}
\label{alg:fast-surgery}
\begin{algorithmic}[1]
    \Require Input state stabilized by $D$ with unknown physical $e(P^D)$. Ancilla code $A$ and maps $(f,g,h,p)$.
    \Ensure Readout
    \begin{align}
        \label{eq:cone-readout}
        \hat{\sigma} \tilde{e}(P) +e(\bar{S}) \in \bar{S} \\
        \label{eq:transversal-readout}
        \delta^A x^{A} +e(Q^{A}_X)\in Q^{A}_X
    \end{align}
    for some
    \begin{align}
        \label{eq:tilde-e-P}
        \tilde{e}(P)=e(Q^{A}_X)\oplus \tilde{e}(P^D)\oplus \tilde{e}(Q^A_Z) \\
        e(\bar{S})=e(\bar{Z}^{A})\oplus e(\bar{S}^D) \oplus e(\bar{X}^A)
    \end{align}
    where non-tilde elements $e$ denote (presumably small) initialization errors,
    and
    \begin{equation}
        \label{eq:tilde-e-relation}
        \hat{\sigma}^D \tilde{e}(P^D) =\hat{\sigma}^D e(P^D).
    \end{equation}
    The output state is stabilized by $D$ and logicals of $D$ in $g\ker \delta^A$ with physical $\tilde{e}(P^D)+gx^{A}$.
    In particular, if $\ell^{\star}(X^A)\in \ker \delta^A$, then $\ell^\star(P^D)=g\ell^{\star}(X^A)$ has $\F_2$ measurement
    \begin{equation}
        \Lambda^D( \tilde{e}(P^D)+gx^{A},\ell^\star(P^D)) = \bra \hat{\sigma}\tilde{e}(P)|\ell^\star(X^A)\ket+\bra \text{Eq. } \eqref{eq:transversal-readout} |p\ell^\star(X^A)\ket
    \end{equation}
    \State Initialize ancilla qubits $Q^{A}$ in the $Z$-basis with unknown error (in the $X$-sector) $e(Q^{A}_X)$
    \State Initialize check qubits $Q(S)$ in the $X$-basis with unknown error $e(\bar{S})$ in the $Z$-sector
    \State Apply $O(1)$ depth CNOT and CZ circuit for $\sigma:S\to P$ of the height-2 cone $C$
    \State Measure out $Q(S)$ transversally in the $X$-basis with readout
    \begin{equation}
        \hat{\sigma} \tilde{e}(P) +e(\bar{S}) \in \bar{S}
    \end{equation}
    where $\tilde{e}(P),e(\bar{S})$ satisfy Eq. \eqref{eq:tilde-e-P}-\eqref{eq:tilde-e-relation},
    so that the output state has physical $\tilde{e}(P)$.
    \State Measure out ancilla qubits $Q^A$ transversally in the $Z$-basis with readout
    \begin{equation}
        \delta^A x^{A} +e(Q_X^{A})\in Q_X^{A}
    \end{equation}
    for some $x^{A}\in X^{A}$,
    so that the output state is stabilized by Paulis $\in \im \sigma^D +g\ker \delta^A$ with physical $\tilde{e}(P^D)+gx^{A}$.
\end{algorithmic}
\end{algorithm}

\bibliography{main.bbl}

\end{document}